\documentclass[11pt]{article}
\usepackage[paper=a4paper, left = 2.5cm, right = 2.5cm, top = 2.5cm, bottom = 2.5cm]{geometry}
\usepackage[utf8]{inputenc}
\usepackage[authoryear,round]{natbib}
\usepackage{graphicx}
\usepackage{amsmath}
\usepackage{amsfonts}
\usepackage{amsthm}
\usepackage{amssymb}
\usepackage{mathrsfs,dsfont}
\usepackage{multirow}
\usepackage{longtable}
\usepackage{arydshln}
\usepackage{enumitem}
\usepackage{bm}
\usepackage{bbm}
\usepackage{setspace}
\usepackage{rotating}
\usepackage{subfigure}
\usepackage{color}
\usepackage{bigstrut}
\usepackage{setspace}
\usepackage[dvips]{epsfig}
\usepackage{graphicx}
\usepackage{eurosym}
\usepackage{booktabs}
\usepackage{nicefrac}
\usepackage{tabls}
\usepackage{placeins}
\usepackage{footmisc}

\usepackage{caption}
\DeclareCaptionFont{tada}{\fontsize{9pt}{9.5pt}\selectfont #1}
\usepackage{float}

\usepackage{xr-hyper}  
\usepackage[
breaklinks,
colorlinks=true,
pagebackref=false,
pdfauthor={},%
pdftitle={Efficient Sampling for Realized Variance Estimation in Time-Changed Diffusion Models},%
citecolor=blue,
linkcolor=blue,
urlcolor=blue
]{hyperref}

\allowdisplaybreaks

\newcommand{\N}{\mathbb{N}}

\newcommand{\E}{\mathbb{E}}

\newcommand{\RV}{\operatorname{RV}}

\newcommand{\IV}{\operatorname{IV}}
\newcommand{\rIV}{\operatorname{rIV}}
\newcommand{\IQ}{\operatorname{IQ}}

\newcommand{\plim}{\operatorname{plim}}

\newcommand{\Prob}{\mathbb{P}}

\newcommand{\proofpart}[1]{{\bf{#1.}}}

\newcommand{\Fil}{\mathbb{F}}
\newcommand{\Filint}{\Fil^{\lambda,\varsigma}}
\newcommand{\FilintN}{\Fil^{\lambda,\varsigma, N}}
\newcommand{\F}{\mathcal{F}}

\newcommand{\Fint}{\mathcal{F}^{\lambda,\varsigma}}
\newcommand{\FintN}{\mathcal{F}^{\lambda,\varsigma, N}}

\newcommand{\Stil}{\int_0^T\varsigma^2(r)dN(r)}
\newcommand{\Sno}{\int_0^T\varsigma^2(r)\lambda(r)dr}

\newcommand{\btau}{\boldsymbol{\tau}}
\newcommand{\Ui}{U_i}

\newcommand{\edit}[1]{\textcolor{black}{#1}}

\theoremstyle{definition} 
\newtheorem{thm}{Theorem}
\newtheorem{lem}[thm]{Lemma}
\newtheorem{cor}[thm]{Corollary}
\newtheorem{prop}[thm]{Proposition}

\theoremstyle{definition}
\newtheorem{defn}[thm]{Definition}

\newtheorem{rmk}[thm]{Remark}

\newtheoremstyle{assumption}
{3pt}
{3pt}
{}
{}
{\bf}
{.}
{.5em}
{\thmname{#1} (\thmnote{#3}\thmnumber{#2})}

\theoremstyle{assumption}
\newtheorem{ass}{Assumption}

\usepackage[colorinlistoftodos,textsize=footnotesize]{todonotes}
\newcommand{\Comments}{1}
\newcommand{\mynote}[2]{\ifnum\Comments=1\textcolor{#1}{#2}\fi}
\newcommand{\mytodo}[2]{\ifnum\Comments=1%
	\todo[linecolor=#1!80!black,backgroundcolor=#1,bordercolor=#1!80!black]{#2}\fi}

\ifnum\Comments=1               
\fi

\ifnum\Comments=1               
\fi

\ifnum\Comments=1               
\fi

\begin{document}

\title{\textsc{Efficient Sampling for Realized Variance Estimation in Time-Changed Diffusion Models}}

\author{
	{\normalsize 
	Timo \textsc{Dimitriadis}${}^{1}$, 
	Roxana \textsc{Halbleib}${}^{2}$,
	Jeannine \textsc{Polivka}${}^{3}$,}
	\\
	{\normalsize
	Jasper \textsc{Rennspies}${}^{4}$,
	Sina \textsc{Streicher}${}^{5}$ and
	Axel Friedrich \textsc{Wolter}${}^{6}$}
}
 
\normalsize

\date{\today}

\maketitle

\begin{abstract}
 	This paper analyzes the benefits of sampling intraday returns in intrinsic time for the realized variance (RV) estimator.
	We theoretically show in finite samples that depending on the permitted sampling information, the RV estimator is most efficient under either hitting time sampling that samples whenever the price changes by a pre-determined threshold, or under the new concept of realized business time that samples according to a combination of observed trades and estimated tick variance.
	The analysis builds on the assumption that asset prices follow a diffusion that is time-changed with a jump process that separately models the transaction times.
	This provides a flexible model that allows for leverage specifications and Hawkes-type jump processes and separately captures the empirically varying trading intensity and tick variance processes, which are particularly relevant for disentangling the driving forces of the sampling schemes.
	Extensive simulations confirm our theoretical results and show that for low levels of noise, hitting time sampling remains superior while for increasing noise levels, realized business time becomes the empirically most efficient sampling scheme.
	An application to stock data provides empirical evidence for the benefits of using these intrinsic sampling schemes to construct more efficient RV estimators as well as for an improved forecast performance.\\[0.1cm]


\end{abstract}

\vspace{3cm} 
 
\footnotetext[1]{Corresponding author.  Faculty of Economics and Business, Goethe University Frankfurt, 60629 Frankfurt am Main, Germany, \href{mailto:dimitriadis@econ.uni-frankfurt.de}{dimitriadis@econ.uni-frankfurt.de}.}
\footnotetext[2]{Institute of Economics, University of Freiburg, Germany; email: \href{mailto:roxana.halbleib@vwl.uni-freiburg.de}{roxana.halbleib@vwl.uni-freiburg.de}}
\footnotetext[3]{University of St.~Gallen, Switzerland} 
\footnotetext[4]{Institute of Economics, University of Freiburg, Germany; email: \href{mailto:jasper.rennspies@vwl.uni-freiburg.de}{jasper.rennspies@vwl.uni-freiburg.de}}
\footnotetext[5]{KOF Swiss Economic Institute, ETH Zürich, 
Switzerland; email: \href{mailto:streicher.sina@gmail.com}{streicher.sina@gmail.com}}
\footnotetext[6]{Department of Computer and Information Science, University of Konstanz, Germany; email: \href{mailto:axel-friedrich.wolter@uni-konstanz.de}{axel-friedrich.wolter@uni-konstanz.de}}

\pagebreak
\onehalfspacing

\section{Introduction}

The estimation and forecasting of the variance of daily stock returns plays a major role in risk management, portfolio optimization and asset pricing. Accurate estimates of the daily variation of asset prices are commonly obtained by using intraday information as in the realized variance (RV) estimator introduced by \citet*{andersen1998a} and \citet*{andersen2001a, andersen2001b}. Together with \citet*{barndorff2002a} and \citet*{meddahi2002}, they show that under the assumption that the logarithmic price process follows a standard continuous-time diffusion model, \edit{RV is an unbiased and consistent estimator of the quadratic variation, which coincides with the integrated variance (IV) in the absence of jumps \citep*{barndorff2008, Barndorff2011SubsamplingRK, Andersen2012}.}

Despite the theoretically appealing approaches of subsampling \citep{zhang2005}, realised kernels \citep{barndorff2008} and pre-averaging \citep{PodolskijVetter2009} for robustifying the RV estimator to market mircrostructure noise (MMN), the standard RV estimator at low frequencies such as sampling every five minutes is still regularly employed in empirical work, see e.g.\ \citet{LiuPattonSheppard2015, bollerslev2018risk, bollerslev2020good, bollerslev2022zero, bates2019crashes, bucci2020realized, reisenhofer2022harnet, alfelt2023singular, Patton2023Bespoke} among many others.\footnote{More fundamentally, the bibliographic review of \citet{Hussain2023} analyses 2920 papers and summarizes that ``5-minute interval data appear to be the most favored choices in terms of high-frequency data usage.''}
Reasons for the standard RV's ongoing popularity include its simple and intuitive implementation, the fact that low(er) frequencies can be used at which MMN is not a major concern, that its convergence rate is substantially faster compared to the previously mentioned approaches,
and that it still performs comparably well in empirical studies \citep{LiuPattonSheppard2015}.

While most of the literature focuses on sampling returns \emph{equidistantly in calendar time} such as every five minutes, financial markets do not tick in calendar time. 
Instead, their intraday trading activity and tick variance (price variance of adjacent transactions or quotes) is time-varying, which might provide important information about the market's pulse and especially its riskiness. 
In this paper, we study the theoretical and empirical benefits of using intraday returns sampled in intrinsic time scales to  efficiently estimate the daily IV through the RV estimator.
These time scales accelerate the clock time when the trading or price variations are intense, and they slow the time down when the markets are calm. 
In particular, we differentiate between the time scale driven by the trading activity (Transaction Time Sampling - TTS), the intraday price volatility (Business Time Sampling - BTS), and observed absolute price changes (Hitting Time Sampling - HTS).
For TTS and BTS, we distinguish their implementation into \emph{intensity} and \emph{realized/jump-based} versions, where the latter use the observed amount of trades on a given day whereas the former rely on estimated intensities.
In contrast to e.g., \citet[Section 4]{bandi2008} who derive an optimal sampling frequency given equidistant sampling points, we focus on the ``inverse'' question of how to optimally allocate the sampling points under a given frequency. 
\edit{By \emph{optimal} or \emph{efficient}, we mean a sampling scheme that, among a class of unbiased schemes, attains the smallest mean squared error (MSE), which in this case equals its estimation variance.}

Summarizing our main contributions, we show that using HTS, which samples such that the absolute returns are (approximately) equal, provides a theoretical lower bound for the efficiency of the RV estimator in finite samples in terms of its \edit{MSE}. 
Furthermore, the newly introduced realized BTS (rBTS) scheme, which samples according to a combination of the observed ticks and the (estimated) variance at these ticks, arises as most efficient when restricting attention to sampling schemes that do not use the  observed high-frequency prices for the construction of the sampling points.
This restriction is motivated by the empirical presence of MMN, which has a particularly severe impact on HTS as its sampling times are obtained directly from the noise-contaminated price observations.
In our simulations and the empirical application, both HTS and rBTS exhibit an excellent and overall comparable performance, and clearly dominate the classically used sampling in calendar or tick time.
While HTS dominates for very low frequencies where MMN is (almost) absent, rBTS arises as most efficient when the sampling frequency exceeds the 5 minute level.

Our theory builds on the assumption of a price process that follows a stochastic diffusion that is time-changed with a  jump (e.g., doubly stochastic Poisson or Hawkes) process.
We call this the \emph{tick-time stochastic volatility (TTSV) model}.
It is a joint stochastic model for the  asset prices together with their transaction (or quote) arrival times.
The prices in this model follow a pure jump process that accommodates the time-varying \emph{trading intensity} and \emph{tick variance} processes within its diffusive component.
The spot variance becomes the product of these two time-varying components that behave empirically different for stock markets as portrayed in Figures \ref{fig:TTM_snippet} and  \ref{fig:EstimatedIntensities} below.

The TTSV model is a simple and transparent framework to study statistical (finite sample) properties of the RV estimator with respect to various choices of sampling schemes.
This is achieved by having the trading intensity and tick variance as two separately evolving processes that jointly govern the price variability.
The separate modeling of trading intensity and tick variance particularly allows for a comparative theoretical analysis of sampling according to calendar time, tick time in the sense of observed ticks or trading intensity, business time as measured by (realized) intraday volatility and hitting time by homogenizing absolute price changes.

A theoretical alternative is to work under \emph{discretized} diffusion models as e.g.\ employed in \citet{Jacod2017, Jacod2019, Jacod2018, DaXiu2021, Li2022remedi}, where a continuous diffusion process is augmented with a process separately modeling the arrivals of the transactions.
Similar to the TTSV model, the resulting price process is a pure jump process with price changes at the explicitly modeled arrivals of the transactions.
We illustrate in Appendix \ref{sec:ComparisonDiscretization} that these discretized diffusions are closely related to the TTSV model.
While similar (finite sample or asymptotic) efficiency results might be derived by relying on discretized diffusions, the TTSV model is attractive due to its simplicity and transparency in distinguishing between trading intensity and tick variance.
Some of our results (in particular, Theorem~\ref{thm:EfficientSampling} (b) and (c)), however, require strong independence conditions on the underlying TTSV processes, which could possibly be weakened when working with discretized diffusions.
The TTSV model, however, also allows for the analysis of the price-dependent HTS scheme (opposed to e.g., \citet[Assumption O (2c)]{Li2022remedi}, \citet[Assumption O (ii)]{Jacod2017} and \citet[Assumption A on page 302]{ait2014high}) and it yields the novel realized BTS scheme due to the explicit modeling of the trading times, hence refining the (asymptotic) efficiency results of \citet{Barndorff2011SubsamplingRK}.

Although the idea of intrinsic time sampling is not new to the literature, especially with regard to its empirical benefits \citep*{clark1973,oomen2005,oomen2006,hansen2006,Andersenal2007,Andersenal2010,aitsahalia2011}, its theoretical advantages over the classical calendar time sampling (CTS) scheme are still largely unexplored, especially in finite samples.
Exceptions are \citet*{oomen2005,oomen2006}, who study the statistical properties of RV under intrinsic time sampling schemes, however, based on a compound Poisson price assumption \citep*{press1967}, whose volatility pattern is solely driven by the trading intensity (see also \citet*{griffin2008}).
Hence, this model misses a substantial source of daily return variation, i.e., the one due to the tick variance, as illustrated in Figures \ref{fig:TTM_snippet} and  \ref{fig:EstimatedIntensities} below.
Furthermore, \citet[Corollary 2]{Barndorff2011SubsamplingRK} show that (intensity) BTS arises as an asymptotically efficient \emph{deterministic} sampling scheme for (subsampled) realized kernel estimators.
Our results however also apply to finite sampling frequencies and allow for sampling based on observed ticks and prices (instead of being deterministic) and can hence accommodate the HTS and realized BTS schemes.
\cite{Fukasawa2010RV} analyses the asymptotic MSE of the RV estimator under endogenous sampling schemes, assuming a continuous semi-martingale for the price process that is observed whenever the price changes by a fixed quantity.
\cite{Fukasawa2010RV} shows that, asymptotically, HTS is most efficient. 
In this light, Theorem~\ref{thm:EfficientSampling} (a) can be interpreted as a finite-sample analogue of his result, albeit established in a different setting.
\citet{FukasawaRosenbaum2012}, \citet{robert2012volatility} and \citet{Vetter2017} provide further asymptotic results under endogenous sampling times.

Pure jump processes, as the TTSV model, have already proven to be valuable alternatives to continuous diffusion models to describe financial prices, as they not only capture empirically observed random trading times and price discontinuities, but also offer a flexible framework to address MMN contamination or to price derivatives; see e.g., \citet*{press1967}, \citet*{carr2004}, \citet*{engle2005}, \citet*{oomen2005,oomen2006}, \citet*{Liesenfeld2006} and \citet*{ShephardYang2017}. 
These processes can be further framed and generalised within stochastic time-changed structures, which are mathematically and empirically very effective, but have received so far only moderate attention in the financial econometrics literature \citep*{clark1973,carr2004}. 

The decomposition of spot variance in trading intensity and tick variance has already been  addressed by \citet*{jones1994}, \citet*{ane2000},  \citet*{plerou2001}, \citet*{gabaix2003}, \citet*{Dahlhaus2014}, \citet*{dahlhaus2016}, among others, when studying the intraday trading behaviour in relation to the intraday clock volatility pattern in order to measure spot variance or to test for normality of intraday returns sampled in transaction time scales. They find that, while the intraday trading is highly correlated with the intraday spot variance, the tick variance affects the spot variance as well, although it has a flatter intraday shape. 
Our empirical observation on stock markets complements these findings and reveals that the intraday tick variance and the trading intensity follow mirrored ``J'' patterns (also see \citet*{Admati1988}, \citet*{oomen2006} or \citet*{dong2014}), which jointly result in the well known ``U'' shape of the intraday spot variance, as documented by \citet{Harris1986}, \citet{Wood1985}, \citet{andersen1997} and \citet{Bauwens2001}.

We validate our theoretical results in extensive simulations, where we also examine the impact of a leverage effect through an asymmetric Hawkes-type process and different specifications of MMN on the bias and the MSE of the RV estimator.
Our empirical results show that, as predicted by our theory, the HTS scheme provides the most efficient RV estimates in the absence of noise.
However, the HTS scheme is most sensitive to noise as its sampling times directly rely on absolute changes of the noisy price process.
In contrast, the rBTS scheme is more robust to noise and is superior for the typical sampling frequencies between 1 and 5 minutes under noisy price processes.
The rBTS scheme also clearly dominates a classical implementation of (intensity) BTS, different implementations of TTS and the baseline case of CTS.  

The empirical application considers 27 liquid stocks traded at the New York Stock Exchange (NYSE).
It provides clear empirical evidence for the benefits of using HTS and realized BTS for increasing the statistical quality of the RV estimator in terms of MSE and QLIKE loss in both an in-sample estimation and out-of-sample forecast environment based on the Heterogeneous AutoRegressive (HAR) model of \citet{Corsi2009}.
For the in-sample evaluation, we follow the method of \citet*{Patton2011RV} that facilitates the empirical comparison of competing RV estimators, in our case computed from the different sampling schemes.
The empirical results particularly stress the practical relevance of the HTS and the realized BTS scheme by showing their superiority in a model-free environment.

The remainder of the paper is structured as follows. 
In Section \ref{theory}, we introduce the TTSV model and derive theoretical efficiency results for finite sampling frequencies for the RV estimator.
Section \ref{sec:Simulation} presents a comprehensive simulation study that analyses the performance of RV under different sampling schemes and Section \ref{sec:Application} provides an empirical application to real data. 
We conclude in Section \ref{sec:Conclusions}.
\edit{Appendix~\ref{sec:ProofsMain} provides proofs for our main results.}

The supplemental material contains a comparison to discretized diffusions in Appendix \ref{sec:ComparisonDiscretization}, additional finite sample theory in a setting where sampling can use information from the end of the trading day in Appendix~\ref{sec:EfficientSamplingUptoT}, and a specific comparison to the results of \citet{oomen2006} in Appendix~\ref{sec:ComparisonOomen2006}.
\edit{All proofs---other than those in Appendix~\ref{sec:ProofsMain}---are collected in Appendix~\ref{app:proofs}.} 
Appendix~\ref{sec:RemainderTerms} discusses generalizations of some theoretical results to mildly dependent processes and  Appendix~\ref{sec:AdditionalResults} contains additional empirical results.

\section{Theory} 
\label{theory}

This section introduces some preliminaries in Section \ref{concepts} and presents the TTSV model in Section \ref{sec:TTSV_model}.
Sections \ref{sec:EfficiencyFiniteSample} and \ref{sec:SamplingSchemes} establish finite sample efficiency results for the RV estimator, \edit{which is complemented by additional theory in Appendix~\ref{sec:EfficientSamplingUptoT} that allows for employing information from the entire trading day.}

\subsection{Preliminaries}
\label{concepts}
	
Throughout the paper, all random objects are defined on a filtered probability space $\left(\Omega,\F, \Fil,\Prob\right)$ with filtration $\Fil = \{\mathcal{F}_t\}_{t\geq0}$ that we specify in Assumption \eqref{ass:filtration} below.
If not stated otherwise, all (in)equalities of random variables are understood to hold almost surely.
Let $\{P(t)\}_{t\geq0}$ denote the stochastic process representing the logarithmic price process of an asset, which we assume to be a continuous-time stochastic process that is right-continuous with left limits.
We sometimes abuse notation and simply write $P(t)$, which we also do for other stochastic processes.
\edit{We denote the quadratic variation of the process $P(t)$ over  $[0,T]$ by $[P]_T$.}

For  $0 \le s \le t$, we define the logarithmic return over the interval $[s,t]$ by 
\begin{equation*} 
	\label{eq:return} 
	r(s,t) := P(t)-P(s).
\end{equation*} 
Then, the (model free) \textit{spot (or instantaneous) variance} of the logarithmic price $P$ at time $t$ is\footnote{We consider spot variance in \emph{calendar time} (instead of some intrinsic time) as this conveniently allows to link it to the trading intensity and tick variance as later formalized in Proposition \ref{prop:vola_decomposition}.}
\begin{equation}
	\label{eqn:DefSpotVolatility}
	\sigma^2(t):=\lim_{\delta \downarrow 0} \frac{1}{\delta} \E \left[  r^2 (t, t+\delta) \; \middle| \;  \mathcal{F}_t\right].
\end{equation}

In this paper, we are interested in estimating the \edit{price variability} within a given time period $[0,T]$, where we focus on the case of $T$ being one trading day, i.e., the daily return is given by $r_{\mbox{\scriptsize daily}}:=r\left(0,T\right)=P\left(T\right)-P\left(0\right)$.
Here, \edit{this price variability} is measured by the \textit{integrated variance} (IV) associated with the logarithmic price process $P(t)$ over the interval $[0,T]$ \citep{barndorff2002a, andersen2006}.
\edit{Formally, the IV is defined as}
\begin{equation}
	\label{eqn:GenIVDef}
	\IV\left(0,T\right) := \int_0^T \sigma^2(r)dr.
\end{equation}
Proposition \ref{prop:IVvola} below provides a more formal justification for the IV as our object of interest given that \edit{in expectation}, it equals the variance of the daily asset return.

We primarily focus on the specific choice of a \emph{sampling scheme} for sparsely sampled intraday returns for estimating IV.
Given a filtration $\mathbb{G} = \{\mathcal{G}_t\}_{t\geq0}$ with $\mathcal{G}_t \subset \mathcal{F}_t$, a $\mathbb{G}$-adapted stopping time sampling scheme $\boldsymbol{\tau}$ is a sequence of \edit{increasing} $\mathbb{G}$-adapted stopping times on $[0,T]$,
\begin{equation}
   	\label{eqn:sampling_scheme}
	\boldsymbol{\tau} = \{\tau_0, \tau_1, ...\} \subseteq [0,T],
\end{equation}
such that $\tau_{j-1} \le \tau_j$ for all $j\in\mathbb{N}$. We require $\tau_0=0$ and that for almost all $\omega \in \Omega$ there exists an $n(\omega)\in \mathbb{N}$ such that $\tau_{n(\omega)}(\omega) = T$ \edit{and that $\tau_{j-1} < \tau_j$ for all $j \le n(\omega)$.}
We give specific examples how $\btau$ can be chosen in Section~\ref{sec:SamplingSchemes}. 

Given the sampling times $\btau$, the corresponding intraday returns are
\begin{equation} 
	\label{eq:return_intraday} 
	r_{j} := r(\tau_{j-1},\tau_j) = P(\tau_j ) - P(\tau_{j-1}),    \qquad j=1,\ldots,M,
\end{equation} 
where we associate to a sampling scheme $\btau$ the (random) number of intraday returns $M = M(\btau) = \inf\{n:\tau_n=T\}$. 
Based on the $M \in \N$ intraday returns $r_{j}$ from the grid $\btau$, we follow \cite{andersen1998a}, among many others, and define the \textit{realized variance} (RV) estimator as
\begin{equation}
	\label{eq:RVDef}
	\RV (\btau) := \sum_{j=1}^{M} r_{j}^2,
\end{equation}
where we stress the dependence on the employed sampling scheme with the argument $\btau$.

\subsection{The Tick-Time Stochastic Volatility Model}
\label{sec:TTSV_model}

We model the ticks and log-prices based on a diffusion $B$ with stochastic tick variance $\varsigma$, where $B$ is time-changed by a jump process $N$ (e.g., Poisson- or Hawkes-type) that models the individual ticks.
We refer to this as the Tick-Time Stochastic Volatility (TTSV) model,
\begin{equation}
	\label{eq:TTSV_model}
    \edit{
	P(t) = P(0) + \int_0^t \varsigma(r) \, dU(r),}
\end{equation}
for $t\in [0,T]$, \edit{where $U(r) = B(N(r))$}.
Formally, we build the model on the following assumption:

\begin{ass}
	\label{ass:filtration}
	We assume that there exists a filtered probability space $(\Omega, \mathcal{F},\Fil,\mathbb{P})$, where the filtration\footnote{The minimal filtration that satisfies Assumption~\eqref{ass:filtration} is \edit{the completed right-continuous version of} $\Fil^*=\big\{ \sigma(N(s),\lambda(s),\varsigma(s),B(N(s)), \; 0\leq s\leq t) \big\}_{t \in [0,T]}$.}
    $\Fil=\{\mathcal{F}_t\}_{t \in [0,T]}$ satisfies the usual assumptions (completeness and right-continuity), and there exist:
	\begin{enumerate}[label=(\alph*)]
		\item 
		a counting process $\{N(t)\}_{t \in [0,T]}$, which is an $\Fil$-adapted jump process with a scalar, positive and $\Fil$-predictable intensity process $\{\lambda(t)\}_{t \in [0,T]}$ that is left-continuous with right-hand limits and $\int_0^t\lambda(r)dr<\infty$ a.s. for all $t \in [0,T]$;
	
		\item 
		a tick volatility process $\{\varsigma(t)\}_{t \in [0,T]}$ that is a positive, $\Fil$-predictable and left-continuous process with right-hand limits; 
	
		\item 
		and a (not necessarily $\Fil$-adapted) Brownian motion $\{B(s)\}_{s\geq0}$ such that $\{B(N(t))\}_{t \in [0,T]}$ is $\Fil$-adapted and such that for any jump point $t_i=\inf\{t\geq0,N(t)=i\}$, $i\in\mathbb{N}$, the increment of the Brownian motion $U_i:=B(N(t_i))-B(N(t_{i-1}))$ is independent of $\mathcal{F}_{t_i-}$, i.e., $U_i|\mathcal{F}_{t_i-}\sim\mathcal{N}(0,1)$.

        \item 
        Moreover, the moments $\E \left[ \left( \int_0^T \varsigma^2(r)dN(r) \right) ^2 \right]$
    	and $\E\left[ [P]_T^2 \right]$ are finite, where the quadratic variation of a pure jump process is the sum of the squared increments, $[P]_t := \sum_{0\le t_i \le t} (\Delta P_{t_i})^2$.
        
    \end{enumerate}
\end{ass}

The TTSV model provides a joint model for the tick arrivals $N(t)$ together with the log-price process $P(t)$ that can capture both, time-varying, stochastic trading intensity and tick variance patterns.
At the same time, $P(t)$ is a semi-martingale as a time-changed diffusion model \citep{Monroe1987, Liptser2012}. 
In fact, Proposition \ref{prop:martingale_price} shows that $P$ is an actual martingale, complying with the regularly imposed assumption of efficient markets \citep{DelbaenSchachermayer1994}.
\begin{prop}
	\label{prop:martingale_price}
	Under Assumption~\eqref{ass:filtration}, the TTSV price process $P$, as defined in~\eqref{eq:TTSV_model}, is an  $\Fil$-martingale.
\end{prop}

In the TTSV model, we assume to observe the jump times $N(t)$ together with the logarithmic prices at these times.
We treat the jump times $N(t)$ as transaction times, whereas they could also be other measures of interest such as quote arrivals, volume-related quantities or aggregates of these measures.
The ``intensity'' processes $\lambda(t)$ and $\varsigma(t)$ are latent, and can for example be modeled as standard It\^{o} diffusions, or Hawkes process type intensities; see \citet[Example 2.3]{dahlhaus2016} for a range of possible specifications.\footnote{The price process in (\ref{eq:TTSV_model}) could further be augmented with a finite-variation predictable mean component \citep*{andersen2003}. However, we follow \citet{oomen2006} (see also  \citet{hansen2006}, \citet{aitsahalia2011}, among others) and set it to zero for simplicity.}

In the general form of Assumption \eqref{ass:filtration}, the transaction times $N(t)$ can follow a general jump process with intensity $\lambda(t)$, which implies that $\E \big[ N(t)-N(s) \mid \mathcal{F}_s \big] = \E \big[ \int_s^t \lambda(r) dr \mid \mathcal{F}_s \big]$ holds a.s.\ for all $0 \le s \le t \le T$, i.e., the expected number of arrivals in the period $[s,t]$ is characterized by the accumulated intensity $\int_s^t \lambda(r) dr$; see \citet{bauwens2009} for details.
Besides doubly stochastic (and non-homogeneous) Poisson processes that are characterized by independent arrivals,  Assumption \eqref{ass:filtration} also allows more general intensity-based models such as autoregressive intensity processes \citep{hamilton2002} or self-exciting Hawkes processes \citep{hawkes1971}, which can additionally capture the observed dependence and memory of the trade arrivals on financial markets.

Assumption \eqref{ass:filtration} also allows for capturing ``leverage effects'' as the jump intensity $\lambda$ and the tick-volatility $\varsigma$ can depend on (the sign of) past price changes.
Part (c) of Assumption~\eqref{ass:filtration} governs the price changes at the observed jumps. 
It essentially rules out anticipative dependence of the calendar-time processes $\lambda$ or $\varsigma$ on $B$, in the sense that the path of the intensities following a jump point is independent of the next increment of the Brownian motion. 
Assumption~\eqref{ass:filtration} further contains moment conditions, which ensure that the IV and the integrated quarticity (IQ) in Theorem \ref{thm:MSE_IV} below are finite.

\begin{figure}[tb] 
	\centering
	\includegraphics[width=1\textwidth]{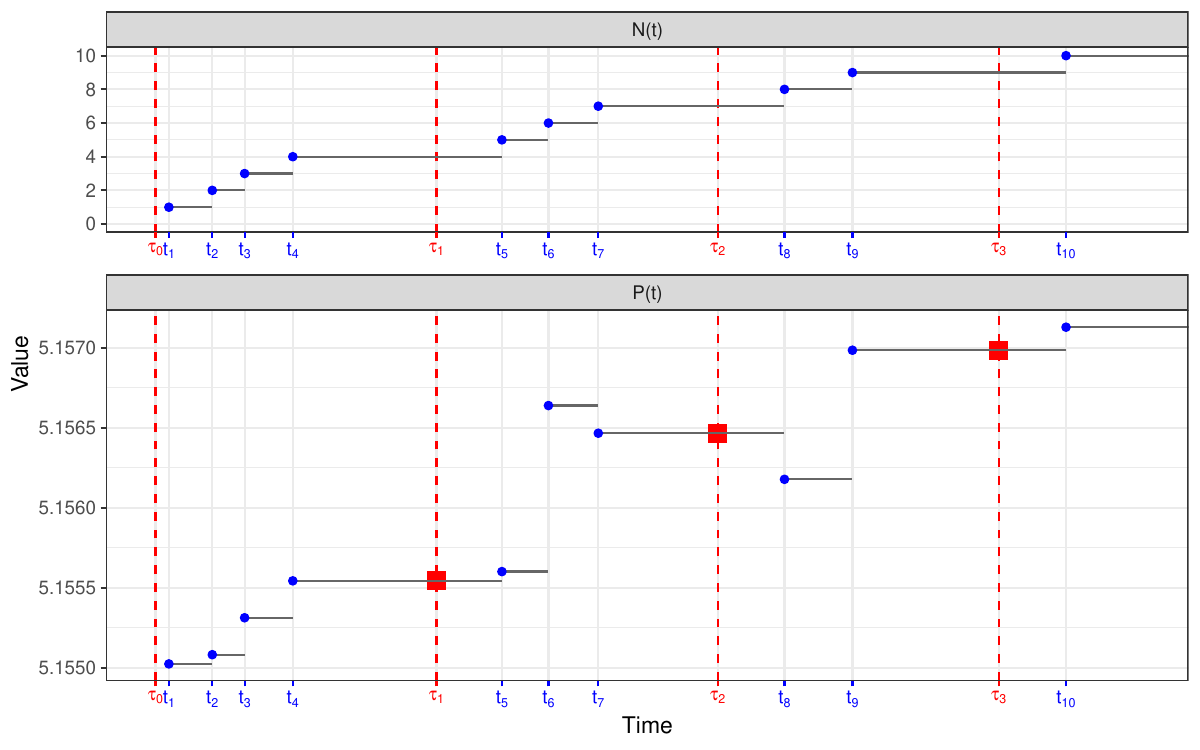} 
	\caption{Illustration of the arrival and sampling times in the TTSV model: 
		The upper panel shows the evolution of the jump process $N(t)$ generating the ticks (arrival times) $t_i$.
		The lower panel shows the log-price process $P(t)$, which exhibits price jumps at the ticks $t_i$ of $N(t)$ and is constant in between.
		The vertical red lines represent the sampling times of an exemplary sampling scheme $\btau$ (that does not have to be equidistant in calendar time), and the red squares show the resampled prices based on the previous tick method.
	}
	\label{fig:TTSVModelIllustration}
\end{figure}

\begin{figure}[tb] 
	\begin{center}
		\includegraphics[width=1\textwidth]{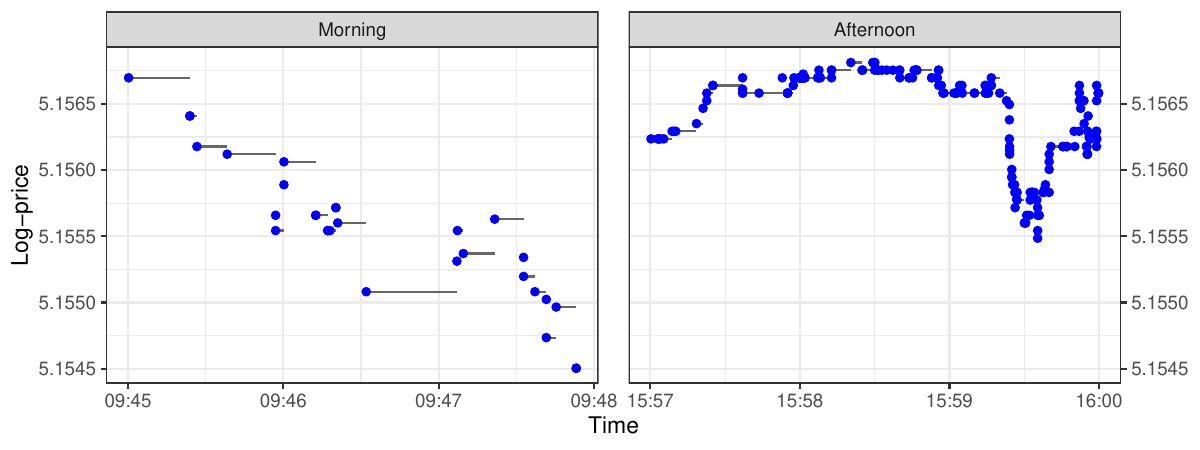}
	\end{center} 
	\caption[$~~$]{IBM transaction log-prices on May 1, 2015 for three minutes in the morning between 9:45am and 9:48am and in the \edit{afternoon} between 15:57pm and 16:00pm. 
	We observe a clear pattern of much more ticks in the \edit{afternoon} and a much higher ``tick-by-tick'' variance in the morning that is typical for stocks traded at the NYSE.}
	\label{fig:TTM_snippet}
\end{figure}

\begin{figure}[tb] 
	\centering
	\includegraphics[width=1\textwidth]{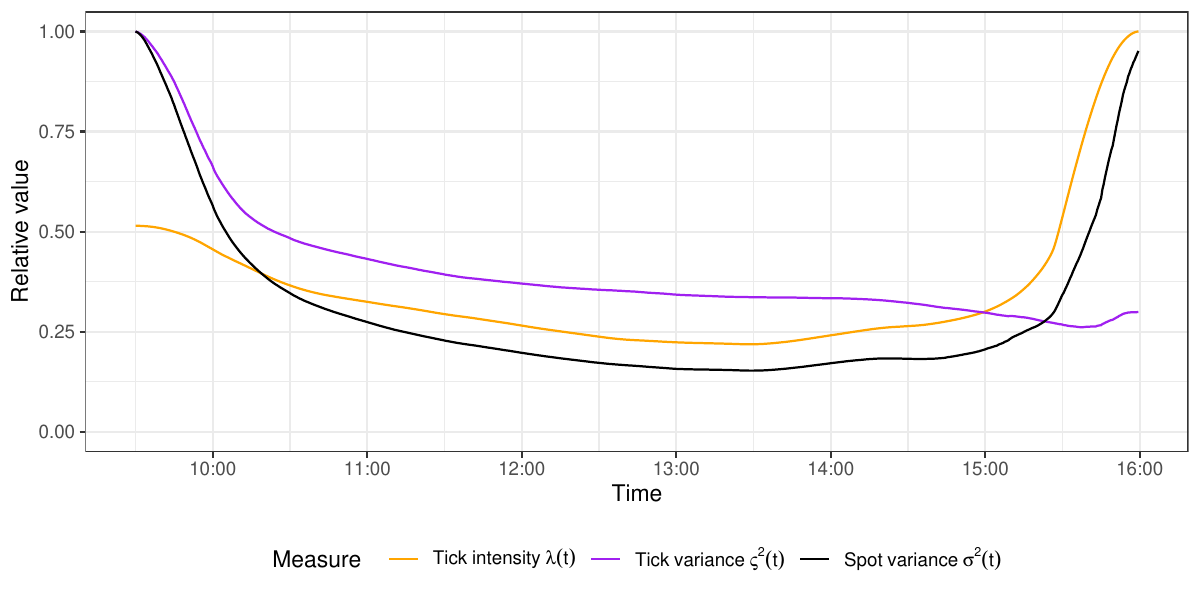}
	\caption{Estimates of the trading intensity $\lambda(t)$, tick variance $\varsigma^2(t)$ and spot variance $\sigma^2(t)$, averaged over all trading days in the year 2018.
	We use the nonparametric kernel estimators for $\lambda(t)$ and $\varsigma^2(t)$ of \cite{dahlhaus2016}, that we augment with a ``mirror image'' bias correction of \citet{DiggleMarron1988}, similar to  \citet[equation (17)]{oomen2006}.
	Following Proposition \ref{prop:vola_decomposition}, the estimate of the spot variance $\sigma^2(t)$ is obtained as the product of the estimated $\lambda(t)$ and $\varsigma^2(t)$.
	}
	\label{fig:EstimatedIntensities}
\end{figure}

In the following, we provide a detailed empirical motivation of the TTSV model:
The jump process $N(t)$ models the ticks (i.e., the transaction or quote times) through its \emph{arrival times} $t_i, i \ge 0$, that satisfy $t_i\in[0,\infty)$ and $t_i<t_{i+1}$ for all $i=1,\dots,N(T)$.
As illustrated by the blue points and black lines in the upper panel of Figure \ref{fig:TTSVModelIllustration}, the sample path of $N(t)$ is a right-continuous step function with jumps of magnitude one at the arrival times $t_i$ such that $N(t) = i$ for $t \in [t_i,t_{i+1})$.
The stochastic intensity process $\lambda(t)$ of $N(t)$ is motivated by the empirical observation that the amount of trading varies drastically throughout the day.
E.g., at the NYSE, there is a much higher trading activity just before market closure than throughout the rest of the day. 
Figure \ref{fig:TTM_snippet} shows the log-prices of the IBM stock traded on the NYSE on May 1, 2015 between 9:45am and 9:48am and between 15:57pm and 16:00pm. 
We see that there are drastically more trades in the \edit{afternoon} than in the morning, which is caused by many traders closing their position due to various reasons, including settlement rules of exchange markets \citep{Admati1988}.
Figure \ref{fig:EstimatedIntensities} shows a non-parametric estimate of the trading intensity $\lambda(t)$ for the IBM stock (details are provided in the figure caption), which confirms this finding.

As $N(t)$ is piecewise constant between its arrival times $t_i$, it holds for all $0 \le s < t \le T$ that
\begin{equation} 
	\label{eq:representation}
	P(t) = P(s) + \sum_{s< t_i\leq t} \varsigma(t_i) U_i,
    \qquad \text{where} \qquad 
	\Ui = B(N(t_i))-B(N(t_{i-1})),
\end{equation}
where the index $i$ in $U_i$ corresponds to the $i$'th observed tick $t_i$.
As graphically illustrated with the blue dots and black lines in the lower panel of Figure \ref{fig:TTSVModelIllustration}, this implies that the log-price $P(t)$ exhibits jumps of magnitude $\varsigma(t_i) \Ui$ at the arrivals of $N(t)$, and it is constant in between.

The stochastic tick volatility $\varsigma(t)$ is essential for the model as one  observes empirically varying tick volatility patterns throughout the day on financial markets.
E.g., Figure \ref{fig:TTM_snippet} shows that at the NYSE, the tick variance of the log-price changes is much higher in the morning than in the \edit{afternoon}, which is  illustrated more formally by the nonparametric estimate of the tick variance  $\varsigma^2(t)$ in Figure \ref{fig:EstimatedIntensities}.
This finding is mainly caused by traders who trade overnight information in the beginning of the day, which triggers large oscillations in the transaction prices and thus, a high tick volatility that calms down until lunch time \citep{Dahlhaus2014}.

Conditionally on an arrival $t_i$, the price change $\varsigma(t_i) \Ui$ is normally distributed with mean zero and variance $\varsigma^2(t_i)$, hence justifying the term \emph{tick variance}. 
Generalizing the conditional Gaussianity of $\varsigma(t_i) \Ui$  in \eqref{eq:representation} might be an interesting avenue for future research.
Nevertheless, due to the stochastic nature of the processes $N(t)$, $\lambda(t)$ and $\varsigma(t)$, the unconditional distribution of the log-prices in the TTSV model is much more general than Gaussian.

\begin{prop}
	\label{prop:vola_decomposition}
	Let Assumption \eqref{ass:filtration} hold and assume that for each $t\in[0,T]$ there exists an $\epsilon>0$ such that $\varsigma^2(r)\lambda(r)$ is bounded for all $r \in [t, t+\epsilon]$ by a random variable $Z(t)$ with $E[Z(t)]<\infty$. 
    Then, the spot variance as given in \eqref{eqn:DefSpotVolatility} satisfies the following decomposition,
	\begin{equation}
		\sigma^2(t) = \varsigma^2(t+) \lambda(t+),
	\end{equation}
	where, for any process $X$, we denote the right-limit as $X(t+):= \lim_{\delta \downarrow 0} X(t+\delta)$.
\end{prop}

Proposition \ref{prop:vola_decomposition}, which is similarly stated in \citet{dahlhaus2016}, shows that in the TTSV model, the spot variance at time $t$ conveniently decomposes into the (right-hand limits of the) trading intensity $\lambda(t)$ and the tick variance of the price jumps $\varsigma^2(t)$, hence combining the two different sources of intraday variation as illustrated, for example, in Figure \ref{fig:EstimatedIntensities}.

Together with the general definition of IV in \eqref{eqn:GenIVDef}, Proposition \ref{prop:vola_decomposition} shows that the IV of the log-price following the TTSV model is given by
\begin{equation}
	\label{eq:IV_TTSV_Definition}
	\IV(0,T) = \int_{0}^T\sigma^2(r) dr = \int_{0}^T \varsigma^2(r+) \lambda (r+) dr = \int_{0}^T \varsigma^2(r) \lambda (r) dr.
\end{equation}
The use of IV as the measure of (daily) return variability in the TTSV model is further motivated by the following result.

\begin{prop}
	\label{prop:IVvola}
	Under Assumption \ref{ass:filtration}, it holds that
	\begin{align*} 
		\E \left[ r^2_{\mbox{\scriptsize daily}}  - \IV \left(0,T\right)\right] = 0.
	\end{align*}
\end{prop}

Hence, under the TTSV model, the variance of the daily return equals the expected IV, which shows that (estimates of) the IV can be interpreted as a measure of daily return variation, similar to classical diffusion processes \citep[Corollary 1 and Theorem 2]{andersen2003}.

For our purposes of analyzing the efficiency of alternative sampling schemes, the TTSV model is particularly useful as it disentangles the time-varying trading activity via the \emph{trading intensity} $\lambda(t)$, and the time-varying \emph{tick variance} through $\varsigma^2(t)$.
As their intraday dynamics differ markedly in empirical data as shown in Figures \ref{fig:TTM_snippet} and \ref{fig:EstimatedIntensities}, the separate model components for $\lambda(t)$ and $\varsigma(t)$ are crucial for some of the results of this paper.

The TTSV model is closely related to many classical models.
For deterministic arrival times $t_1, \dots, t_N$ and a constant tick volatility $\varsigma(t)$, it nests a simple Gaussian random walk in transaction time.
Furthermore, the compound Poisson process used by \citet*{oomen2005,oomen2006} arises when $N(t)$ follows a doubly stochastic Poisson process and when $\varsigma(t)$ is constant.
While this setup allows for modeling tick arrivals as a separate component, it models all time variation in volatility through fluctuations in the arrival intensity. This restriction to a constant tick volatility is a clear limitation.

Lastly, a standard modelling choice is the continuous-time diffusion \citep{barndorff2002a} (without drift and jump terms)
\begin{align}
	\label{eq:diffusion}
	dP(t) = \sigma_\text{diff}(t) \, dB(t), \qquad t\in[0,T],
\end{align}
which is, compared to the TTSV model, not based on a time-change.
In order to explicitly model the stochastic tick arrivals within these diffusion models, \citet{Fukasawa2010RV, Jacod2018, Jacod2017, Jacod2019} apply discretization schemes, where the tick arrivals (or alternatively, the sampling points) are modeled as random times at which one observes (a possibly generalized version of) the diffusion in \eqref{eq:diffusion}.
Similar to the TTSV model, the observed prices are then modeled as a pure jump process with random arrival times, however, with the conceptual difference that the former applies a \emph{time-change} with a jump process while the latter uses \emph{discretization}.

We provide a detailed comparison of the TTSV model to these discretization schemes in Appendix \ref{sec:ComparisonDiscretization}.
While both modeling approaches have their individual merits and limitations, we use the TTSV model in this paper for the following reasons:
First, the TTSV model offers an inherent and transparent decomposition of the spot variance into the empirically relevant components of sampling intensity and tick variance, which directly enables the derivation of particularly insightful results for classically used sampling schemes.
Second, the simplicity of the TTSV model facilitates the derivation of finite sample MSE results---albeit partly under strong independence assumptions.
While such results may also be attainable with discretized diffusion models, we conjecture that doing so would be considerably more involved.
Third, as illustrated in Appendix \ref{sec:ComparisonDiscretization}, the novel realized BTS scheme does not arise as naturally within the discretized diffusion framework.

\subsection{Efficient Sampling}
\label{sec:EfficiencyFiniteSample}

In this section, we derive the bias and MSE of the RV estimator based on general sampling schemes $\btau$ with a \edit{fixed (expected) amount of intraday returns}.
\edit{Our main target is to find an optimal sampling scheme that is efficient in the sense of attaining the smallest MSE among a class of unbiased sampling schemes.}

\begin{thm}
	\label{thm:unbiasedness}
	Under Assumption \eqref{ass:filtration} and for any $\Fil$-adapted sampling scheme $\btau$, the RV estimator in \eqref{eq:RVDef} is an unbiased estimator for the IV:
	\begin{equation}
		\E \big[ \RV(\btau) \big] = \E \big[ \IV(0,T) \big].
	\end{equation}
\end{thm}

As the RV estimator is unbiased for \emph{any} $\Fil$-adapted sampling scheme, there is no theoretical distinction between different sampling schemes $\btau$ in terms of a bias.
We, however, continue by showing that the choice of $\btau$ entails a difference in the estimation efficiency.
To this end, we derive a closed-form expression for the finite-sample MSE of the RV estimator depending on the sampling grid $\btau$.

\begin{thm} 
	\label{thm:MSE_IV}
	Under Assumption \eqref{ass:filtration} and for any $\Fil$-adapted sampling scheme $\boldsymbol{\tau}$, the MSE of the RV estimator is given by
	\begin{align} 
		\label{eq:MSE}
		\E \left[ \big( \RV(\btau) - \IV(0,T) \big)^2 \right] 
		= \frac{2}{3} \E\left[ \sum_{j=1}^M r^4(\tau_{j-1}, \tau_j) \right] + \E\left[ \IQ(0,T) \right],
	\end{align}
	where $\IQ(0,T) = \int_0^T\varsigma^4(r)\lambda(r)dr$ is the integrated quarticity (IQ) of the TTSV model.\footnote{We call $\IQ(s,t) = \int_{s}^{t}\varsigma^4(r)\lambda(r)dr$ the integrated quarticity of the TTSV model as its definition is specific for the TTSV model. If, instead, the integrated quarticity would be defined based on the spot variance as $\int_{s}^{t} \sigma^4(r) dr$, this would result in a slightly different notion of $\int_{s}^{t}\varsigma^4(r)\lambda^2(r)dr$ by using Proposition \ref{prop:vola_decomposition}.}
\end{thm}

Theorem~\ref{thm:MSE_IV} provides a finite sample result for the MSE of any $\Fil$-adapted sampling scheme $\boldsymbol{\tau}$ under general dependence assumptions that for example, allow for Hawkes-type processes including a leverage effect; see the discussion after Assumption~\eqref{ass:filtration}.
In \eqref{eq:MSE}, the MSE is bounded from below by $\E[\IQ(0,T)]$, which merely depends on the underlying process but is invariant to the employed sampling scheme.
Most important for our purposes is the term
$\frac{2}{3} \E\left[ \sum_{j=1}^M r^4(\tau_{j-1}, \tau_j) \right]$,  which depends on the fourth power of the returns, sampled according to $\btau$.
\edit{
    By applying the Cauchy-Schwarz inequality, this term is minimized by a sampling scheme that aims at homogenizing the absolute values of the intraday returns---as e.g., HTS.
    As the MSE expression~\eqref{eq:MSE} in Theorem~\ref{thm:MSE_IV} is only shown to hold for any $\Fil$-adapted sampling scheme $\boldsymbol{\tau}$, it is unclear how a feasible and $\Fil$-adapted scheme could be set up in practice that minimizes \eqref{eq:MSE} \emph{exactly}, especially as the TTSV price process is discontinuous.\footnote{\edit{A trivial---but clearly not $\Fil$-adapted---approach to minimizing~\eqref{eq:MSE} for a given $M$ would be to allocate $\btau$ among all observed tick times so as to minimize the sum of the fourth power of the resulting returns. However, such a sampling scheme would presumably not yield an unbiased RV estimator, rendering the MSE expression~\eqref{eq:MSE} inapplicable. Moreover, it would be computationally very demanding, particularly on days with many ticks and for large values of $M$.}}
    We will later consider feasible and $\Fil$-adapted sampling schemes that aim at making intraday returns as homogeneous as possible---either in terms of their magnitude or in quantities related to their second moment---depending on the setting.
    }

Theorem \ref{thm:MSE_IV} applies to a very general class of sampling schemes that can access the history of all the processes driving the prices in the TTSV model.
In the following, we also consider subclasses of sampling schemes that use less information about the price process and, in particular, are \emph{not} allowed to depend directly on the observed prices.
The intuitive reason is that the actual price observations are affected by MMN, which distorts the MSE result in Theorem~\ref{thm:MSE_IV}.
As we will see in our simulations, this distortion is particularly severe for sampling schemes as HTS that directly rely on the observed high-frequency prices.

Therefore, we define the following two restricted filtrations that determine the precise information that (alternative) sampling schemes can use:
\begin{align*}
    \FilintN :=  \{\FintN_t\}_{t\in[0,T]},
	\qquad \text{ and }  \qquad
    \Filint& := \{\Fint_t\}_{t\in[0,T]},
\end{align*}
where  $\FintN_t = \sigma \big( \lambda(s), \varsigma(s), N(s); \; 0 \le s \le t \big)$ and $\Fint_t = \sigma \big( \lambda(s), \varsigma(s); \; 0 \le s \le t \big)$. 
By considering sampling schemes adapted to the filtrations $\FilintN$ or $\Filint$, we ensure that the possibly noisy price observations do not directly determine the sampling times.
In the \edit{$\FilintN$-adapted} case, we allow for a dependence of the sampling times on the realized tick pattern of the particular day. We refer to the case of $\FilintN$-adapted sampling as ``realized'' or ``jump-based'' sampling and to the case of $\Filint$-adapted as ``intensity-based'' sampling.

We continue to investigate the MSE for the specific classes of sampling schemes introduced above. For this, we first state the two following corollaries, which express the MSE for sampling schemes $\btau$ that are adapted to the reduced filtrations $\FilintN$ and $\Filint$. The first corollary states that the MSE depends on the realized IV (rIV), which we define as
\begin{align}
	\label{eq:rIVdef}
	\rIV(s,t) := \int_s^t\varsigma^2(r)dN(r) = \sum_{s \leq t_i \leq t} \varsigma^2(t_i),
\end{align}
and interpret as a jump-process based and hence ``realized'' version of the classical IV given in \eqref{eqn:GenIVDef} and \eqref{eq:IV_TTSV_Definition}.  

\begin{cor} 
	\label{cor:MSE_jump_based}
	Under Assumption \eqref{ass:filtration}, and given that $U_i^2$ is independent of the paths of $\lambda$, $\varsigma$, and $N$,
    the MSE of the RV estimator for any $\FilintN$-adapted sampling scheme $\boldsymbol{\tau}$ is
	\begin{align} 
		\label{eq:MSE_realized_sampling}
		\E \left[ (\RV(\btau) - \IV(0,T))^2 \right] & = 2 \E \left[\sum_{j=1}^M \rIV(\tau_{j-1}, \tau_{j})^2 \right] + \E \left[ \IQ(0,T) \right] + \E[R(\btau)],
	\end{align}
	where  
	\begin{equation}
		\label{eq:remainder_1}
		R(\btau) := 4\sum_{j=1}^M \left((P_{\tau_{j}}-P_{\tau_{j-1}})^2-([P]_{\tau_{j}}-[P]_{\tau_{j-1}})\right) \rIV(\tau_{j-1},\tau_{j}).
	\end{equation}
\end{cor}

The MSE formula from Corollary~\ref{cor:MSE_jump_based} provides intuition on the relative efficiency of $\FilintN$-adapted  sampling schemes:
Invoking $\E[R(\btau)]=0$, a condition that holds under independence assumptions that are formalized in Theorem~\ref{thm:EfficientSampling} below, the Cauchy-Schwarz inequality directly implies that the MSE can be minimized by specifying $\btau$ such that $\rIV(\tau_{j-1}, \tau_{j})$ is as homogeneous as possible (in expectation).
Notice that the additional requirement in Corollaries \ref{cor:MSE_jump_based} and \ref{cor:MSE_intensity_based} that the $U_i^2$ are independent of the entire paths of $\lambda$, $\varsigma$ and $N$ still allows for leverage effects, as the jump process and the tick variance can depend on the past sign of $U_i$.
In Appendix~\ref{sec:RemainderTerms}, we provide informal theoretical arguments that, under process dependencies that decay fast enough over time (as in Hawkes processes), the remainder term $\E[R(\btau)]$ is approximately equal for all sampling schemes, given that sparse sampling is employed.
We note here already that our simulations confirm this finding.

\begin{cor} 
	\label{cor:MSE_intensity_based}
	Under Assumption \eqref{ass:filtration}, and given that $U_i^2$ is independent of the paths of $\lambda$, $\varsigma$, and $N$, the MSE of the RV estimator for any $\Filint$-adapted sampling scheme $\btau$ is
	\begin{align} 
		 \label{eq:MSE_intensity}
		\E \left[ (\RV(\btau) - \IV(0,T) )^2 \right] 
        &= 2 \E \left[ \sum_{j=1}^M \IV(\tau_{j-1}, \tau_j)^2 \right] + 3 \E \left[ \IQ(0,T)  \right] + \E \left[ R(\btau) \right] + \E \big[ \widetilde{R}(\btau)\big],
	\end{align}
	where 
    $R(\boldsymbol{\tau})$ is as in \eqref{eq:remainder_1} and for $\widetilde{N} := \left\{ N(t)-\int_0^t \lambda(r)dr \right\}_{t\in[0,T]}$, we define
	\begin{equation}
		\label{eq:remainder_2}
		\widetilde{R}(\boldsymbol{\tau}) := 4 \sum_{j=1}^M \IV (\tau_{j-1}, \tau_j) \E\left[ \int_{\tau_{j-1}}^{\tau_j} \varsigma^2(r) d\widetilde{N}(r) \Bigg| \mathcal{F}_{\tau_j}^{\lambda, \varsigma} \right].
	\end{equation}
\end{cor}

Corollary~\ref{cor:MSE_intensity_based} shows that restricting attention to $\Filint$-adapted sampling schemes $\btau$ leads to a similar formula as in Corollary~\ref{cor:MSE_jump_based}.
However, efficiency is now characterized by homogeneity of $\IV(\tau_{j-1}, \tau_j)$ (opposed to the \emph{realized} IV in Corollary~\ref{cor:MSE_jump_based}), and the result is subject to the further remainder term $\widetilde{R}(\boldsymbol{\tau})$.

The following theorem summarizes these results by imposing conditions under which the remainder terms ${R}(\boldsymbol{\tau})$ and $\widetilde{R}(\boldsymbol{\tau})$ vanish in expectation.
\begin{thm}
	\label{thm:EfficientSampling}
	For a given constant $\overline{M} = \E[M(\btau)] \in \mathbb{N}$, we consider sampling schemes $\btau$ with respect to different filtrations. 
	Under Assumption \eqref{ass:filtration}, the MSE of the RV estimator is minimized
	\begin{enumerate}[label=(\alph*)]
		\item 
		among all $\Fil$-adapted sampling schemes, by a sampling scheme such that 
		 $| r(\tau_{j-1}, \tau_j) | = \sqrt{\E[\IV(0,T)] \big/ \overline{M}}$; 
		
		\item  
		\label{thm:MSE_minimization_realized}
		among all $\FilintN$-adapted sampling schemes, by a sampling scheme such that $\rIV(\tau_{j-1}, \tau_j) = \E[\IV(0,T)] \big/ \overline{M}$  under the additional assumption that $B$ is independent from $\lambda$, $\varsigma$ and $N$;
        
		\item 
		among all $\Filint$-adapted sampling schemes, by a sampling scheme such that $\IV (\tau_{j-1}, \tau_j) = \E[\IV(0,T)] \big/ \overline{M}$ under the additional assumptions that $B$ is independent from $\lambda$, $\varsigma$ and $N$ and that $N$ is a doubly stochastic Poisson process with intensity $\lambda$.
	\end{enumerate}
\end{thm}

Roughly speaking, all three parts of Theorem~\ref{thm:EfficientSampling} suggest \emph{homogenizing} the sampled returns.
These parts mainly differ by the quantity that is homogenized, which will naturally be contained in the filtration the sampling schemes are adapted to.
\edit{It is important to note that in all three parts of Theorem~\ref{thm:EfficientSampling}, adaptiveness to a certain filtration is required. This makes it unclear how the condition of homogenizing returns can be satisfied \emph{exactly} in practice, rendering these lower bounds infeasible in implementation. In Section~\ref{sec:SamplingSchemes}, we therefore consider feasible sampling schemes that satisfy the homogeneity conditions approximately.}

Theorem~\ref{thm:EfficientSampling} (a) establishes that the most general finite sample efficiency is achieved when sampling times are chosen such that the absolute return values coincide throughout a trading day, hence pertaining to the HTS scheme.
Parts (b) and (c) examine settings where the price information is not used for the construction of the sampling times.
These restricted settings are practically relevant, as the observed high-frequency returns are regularly contaminated by MMN, which can make their use in constructing the sampling times problematic as will be illustrated in our simulations.

On a technical level, the additional independence assumptions in parts (b) and (c) ensure that the remainder terms $R(\btau)$ and $\widetilde{R}(\btau)$ from Corollaries~\ref{cor:MSE_jump_based} and \ref{cor:MSE_intensity_based} vanish in expectation. 
As exemplified in Appendix~\ref{sec:RemainderTerms}, we conjecture that these remainder terms have a minor dependence on the employed sampling schemes, suggesting that the efficiency results of parts (b) and (c) also continue to hold for processes with mild dependencies, as reflected in our simulations.

While Theorem~\ref{thm:EfficientSampling} describes idealized conditions for efficient sampling, the following Section~\ref{sec:SamplingSchemes} discusses their practical implementation.

\subsection{Sampling Schemes}
\label{sec:SamplingSchemes}

Most practically relevant sampling schemes $\btau$ that aim to homogenize a certain quantity, as formalized through Theorem~\ref{thm:EfficientSampling}, can be specified based on a (weakly) increasing and possibly stochastic \emph{accumulated sampling intensity} process $\{\Phi(t)\}_{t \in [0,T]}$. 
For example, for the classical CTS scheme, $\Phi(t) = t$ equals the identity.
In contrast, different variants of transaction- and business-time sampling are based on combinations of the accumulated trading intensity, tick variance and the observed tick arrivals.
If $\Phi$ is differentiable on $(0,T)$, its derivative is denoted by $\phi$ and has the interpretation of a sampling intensity.

Given an accumulated sampling intensity process $\Phi$, the sampling times $\tau_j$, $j=0,\dots, M$ are chosen as the generalized inverse of $\Phi$,
\begin{align}
	\label{eq:ChoosingTau_j}
	\tau_j = \inf \big\{ t \in [0,T]: \; \Phi(t) \ge j \cdot \delta \big\},
\end{align}	
for some possibly stochastic threshold $\delta > 0$.
This ensures that we sample \emph{equidistantly in the accumulated sampling intensity} with $\tau_0 = 0$ and $\tau_{M} = T$.\footnote{If $\Phi(t)$ is continuous, \eqref{eq:ChoosingTau_j} implies that $\Phi(\tau_j) - \Phi(\tau_{j-1}) = j \delta -  (j-1)\delta = \delta$ is constant for all $j=1,\dots, M$.
For the discontinuous versions of $\Phi(t)$ (such as sampling every $K \in \mathbb{N}$ transactions), this only holds approximately.}
We then obtain the prices at sampling times $\tau_j$ with the ``previous tick method'' that is consistent with the TTSV modeling assumption, as illustrated with the red squares in the lower panel of Figure \ref{fig:TTSVModelIllustration}.

In this paper, we focus on the following common sampling schemes that arise by choosing different measures for the sampling intensity:

\begin{enumerate}
    \item 
	\textbf{Calendar Time Sampling (CTS)}, for which $\Phi^{\mbox{\tiny CTS}}(t) = t$, such that we have a constant sampling intensity $\phi^{\mbox{\tiny CTS}}(t) = 1$.
    CTS returns homogenize calendar time between sampling points $\tau_j^{\mbox{\tiny CTS}}=j {T}/{M}$ for $j=0,\ldots,M$, and its simple implementation makes it the most widespread sampling scheme in finance.
    It, however, neglects any information on intraday trading and volatility patterns. 
    
	\item 
	\textbf{Intensity Transaction Time Sampling (iTTS)}, for which the data is sampled equidistantly in the \emph{trading intensity} $\phi^{\mbox{\tiny iTTS}}(t) = \lambda(t)$ of the TTSV model,
    i.e., $\Phi^{\mbox{\tiny iTTS}}(t) = \Lambda(0,t)$, where $\Lambda(s,t) := \int_s^t \lambda(r) dr$. 
    Sampling according to iTTS homogenizes the returns according to the trading intensity.
    
	\item
	\textbf{Realized Transaction Time Sampling (rTTS)}, for which the data is sampled equidistantly in the \emph{observed number of transactions}, such that $\Phi^{\mbox{\tiny rTTS}}(t) = N(t)$.
	This implies that we sample every
	$N(\tau^{\mbox{\tiny rTTS}}_{j}) - N(\tau^{\mbox{\tiny rTTS}}_{j-1}) = \delta$ observed ticks (given that $\delta$ is integer-valued) such that rTTS homogenizes returns with respect to the observed transactions.
    
	\item 
	\textbf{Intensity Business Time Sampling (iBTS)}, for which the data is sampled equidistantly in integrated \emph{spot variance} $\phi^{\mbox{\tiny iBTS}}(t) = \sigma^2(t) = \varsigma^2(t) \lambda(t)$, i.e., we choose $\Phi^{\mbox{\tiny iBTS}}(t) = \IV(0,t)$. Hence, iBTS homogenizes the returns according to the spot variance.

	\item 
	\textbf{Realized Business Time Sampling (rBTS)}, where the data is sampled equidistantly in the \emph{tick variance-weighted observed number of transactions}.
	In particular, we choose $\Phi^{\mbox{\tiny rBTS}}(t) = \sum_{t_i\leq t} \varsigma^2(t_i) = \int_0^t\varsigma^2(r)dN(r)$,
	such that the returns are (approximately) homogenized with respect to \emph{realized} IV.
    
\end{enumerate}

While CTS is deterministic, iTTS and iBTS are $\Filint$-adapted, and rTTS and rBTS are $\FilintN$-adapted, at least given that a deterministic threshold $\delta$ is used.
For a practical implementation of iTTS, iBTS, and rBTS, we have to estimate the intensity processes $\lambda$ and/or $\varsigma$, which we do by averaging over past trading days. 

The above sampling schemes $\btau$ result in $M = M(\btau) = \Phi(T)/\delta$ sampled returns per day, which is in general a stochastic quantity.
In practice, it is, however, often desirable to fix $M$ for the following reasons:
First, fixing $M$ allows for a convenient comparison across sampling schemes. 
We will do this later on in simulations and the empirical application. 
Second, as argued in \cite{zhang2005}, among many others, the value of $M$ is the main driver of the bias of the RV estimator in the presence of MMN.
By fixing $M$, we particularly ``stabilize'' the effect of noise on the RV estimator,
as this prevents the RV from being more affected by noise on higher volatility days than on lower volatility days.

In empirical work, one often deviates from the stopping time assumption and fixes $M$ by choosing $\delta = \Phi(T)/M$.
In practice, when estimating RV at the end of a trading day, the information $\Phi(T)$ is observable or can be estimated.
Formally, the sampling schemes are no longer adapted to the filtrations $\Filint$ or $\FilintN$, but rather to their enlargements by $\sigma(\Phi(T))$, where $\Phi(T)$ corresponds to the given sampling scheme.
While the theoretical results of Section \ref{sec:EfficiencyFiniteSample} do not formally apply to that setting, we show in simulations (see Figure~\ref{fig:RMSE_Stopping}) that the effect is negligible.
Moreover, Appendix \ref{sec:EfficientSamplingUptoT} derives finite sample theory with results analogous to cases (b) and (c) of Theorem~\ref{thm:EfficientSampling}, where the sampling times are allowed to depend on information up to time $T$.

We finally describe the HTS scheme that is already analyzed in \citet{Fukasawa2010RV, Vetter2017, FukasawaRosenbaum2012}, and which is \emph{not} based on an accumulated intensity process:
\begin{enumerate}
	\setcounter{enumi}{5}
	\item 	 
	\textbf{Hitting Time Sampling (HTS)}, where the data is sampled whenever the observed price change exceeds a fixed threshold $\delta \in \mathbb{R}_+$, i.e, $\tau_0 = 0$ and, given some $\tau_{j-1} \in [0,T]$ for $j \ge 1$, we set
	\begin{align}
        \label{eq:HTS_Implementation}
		\tau_j = \inf \big\{ t \in [0,T]: \quad \vert P(t) - P(\tau_{j-1})\vert \ge \delta \big\}.
	\end{align}
	This results in a random number $M = M_\delta$ of samples per day, and we set $\tau_{M} = T$. 
    HTS homogenizes the absolute return values, at least approximately for the TTSV model, as the discontinuity of the price process does in general not allow to find times where $\vert P(\tau_{j}) - P(\tau_{j-1}) \vert = \delta$ holds exactly; see Figure~\ref{fig:HTS_Overshooting}. HTS is model-free and does not require estimation of any underlying intensity processes.
\end{enumerate}

Reconsidering our main result, Theorem~\ref{thm:EfficientSampling}, we see that HTS is tailored to the most general case (a), where the absolute return values should coincide.
Similarly, rBTS aims at homogenizing rIV, which is the most efficient among the  $\FilintN$-adapted sampling schemes, and iBTS homogenizes IV, which is the most efficient among the $\Filint$-adapted sampling schemes.

It is important to note that Theorem~\ref{thm:EfficientSampling} suggests \emph{idealized} sampling schemes, which are, however, not necessarily feasible due to the discontinuity of the underlying processes in the TTSV model as well as in practice.
For HTS, this leads to a common ``overshooting'' effect, where the absolute returns are only guaranteed to be larger than $\delta$.
This overshooting effect is particularly pronounced for small values of $\delta$ and for days with little trading activity; see Figure~\ref{fig:HTS_Overshooting}.
\edit{Although other $\Fil$-adapted schemes---such as sampling whenever the price process crosses an equidistant grid, ignoring repeated crossings of the same grid level---could also homogenize absolute returns, we find their performance similar to HTS and therefore do not pursue them further.}

For HTS, it is unfortunately not possible to fix the number of samples $M$, which is often desirable, as argued above.\footnote{Even with a large number of values for $\delta$ and trial and error, it might be impossible to obtain certain values of $M$ given an observed price path.}
Through Theorem~\ref{thm:EfficientSampling}, it is only feasible to fix the expected number of samples $\overline{M}$ by choosing $\delta^2 = \E[\IV(0,T)] \big/ \overline{M}$, at least in the absence of MMN, by ignoring the overshooting effect, and by estimating $\E[\IV(0,T)]$, e.g., by a standard RV estimator based on CTS returns.

\begin{figure}[tb] 
	\centering
	\includegraphics[width=1\textwidth]{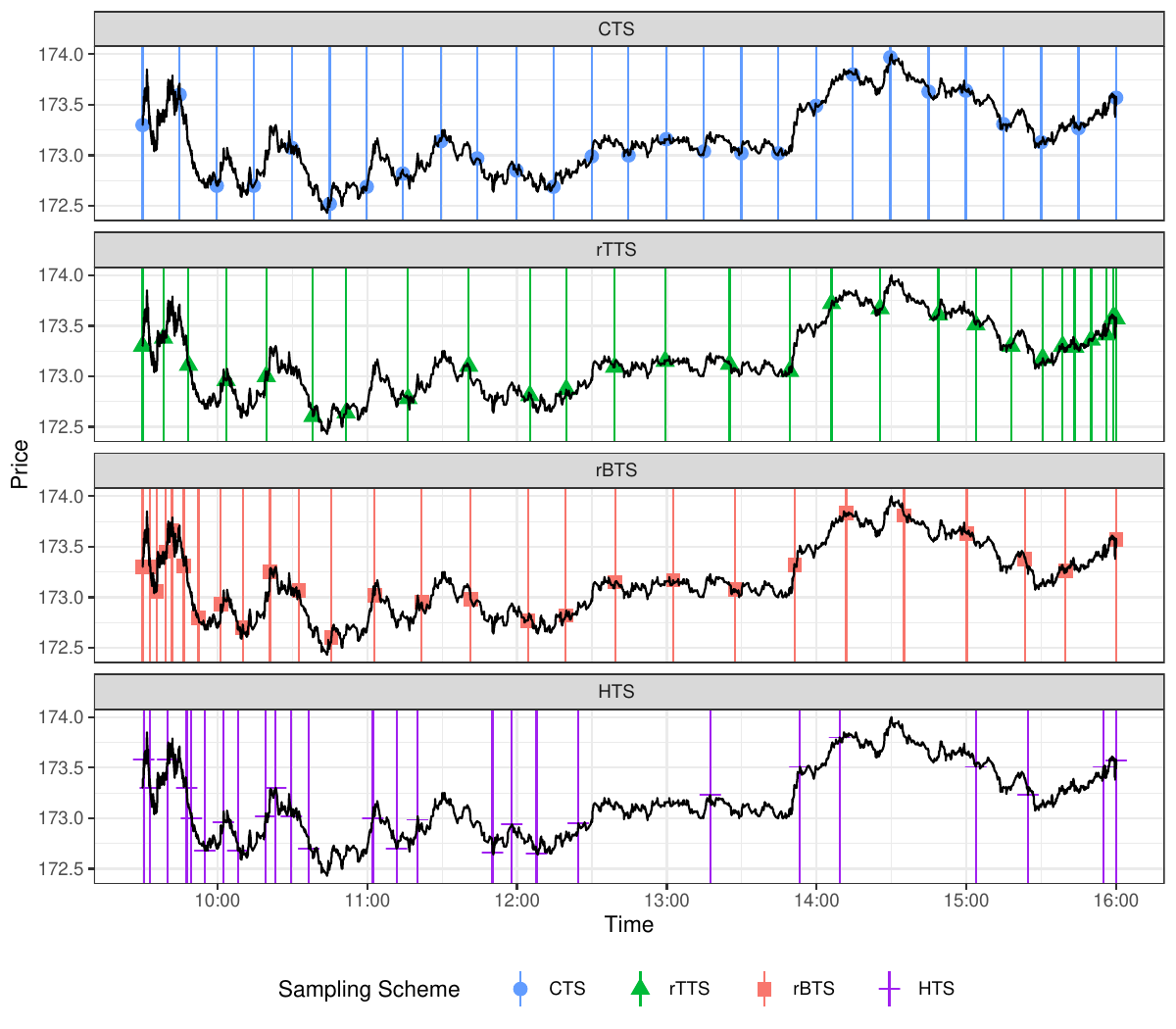}
	\caption{IBM log-price on May 1, 2015 together with the CTS, rTTS, rBTS and HTS sampling schemes for $M=26$, i.e., corresponding to intrinsic time 15 minute returns.
		For the rBTS scheme, we estimate the tick variance $\varsigma^2(\cdot)$ as the average of the estimates over the past 50 days using the estimator of \citet{dahlhaus2016}.
		For HTS, we choose the threshold $\delta=0.00158$ that happens to result in exactly 26 sampled observations on the given day.}
	\label{fig:sampling_comparison_IBM}
\end{figure}

Figure \ref{fig:sampling_comparison_IBM} shows the price path of IBM on May 1, 2015, 
with estimates of the sampling times $\btau$ with $M=26$ under the four sampling schemes CTS, rTTS, rBTS and HTS, presented in the four panels.
The figure reveals a substantial variation of the sampling times across the sampling schemes:
While the sampling points are equidistant in time for CTS, we sample more often in the afternoon with rTTS, but more often in the morning with rBTS and HTS.
In particular, the empirically observed difference between rTTS and rBTS highlights the importance and necessity of a refined price model, such as the TTSV model, that can separately accommodate the different intraday patterns of the trading intensity and tick variance.

\begin{rmk} 
    The efficiency results of Theorem~\ref{thm:EfficientSampling} (b) and (c) extend the theoretical findings of \citet[Proposition~1]{oomen2006}, who considers sampling based on observed and expected transactions in a restricted version of the TTSV price process based on a doubly stochastic Poisson process with a constant tick variance.
    Disregarding whether sampling schemes are allowed to use the information $\Phi(T)$ (also see Appendix~\ref{sec:EfficientSamplingUptoT}), the sampling schemes of  \citet{oomen2006}  are closely related to our $\FilintN$-adapted sampling.
    In summary, \cite{oomen2006} finds that in his model, sampling with respect to the observed transactions (i.e., rTTS $\mathrel{\widehat=}$ rBTS) is more efficient than sampling with respect to the sampling intensity that represents the expected number of transactions (i.e., iTTS $\mathrel{\widehat=}$ iBTS).
    This finding is consistent with the results of our Theorem~\ref{thm:EfficientSampling} (b) and furthermore, with Theorem \ref{prop:MSE_tick} and Corollary~\ref{cor:EfficientSamplingAppendix} in Appendix~\ref{sec:EfficientSamplingUptoT}, where we thoroughly illustrate the comparison for the setting where information on $\Phi(T)$ is used for sampling.\footnote{
    	The past literature on sampling schemes often uses inconsistent terminologies, which requires special care when comparing the results among different papers.
    	E.g., \cite{oomen2006} refers to BTS as sampling with respect to the ``expected number of transactions'' and to TTS as sampling with respect to the ``realized number of transactions'', which matches our definitions of iTTS and rTTS, respectively.
    	Furthermore, \cite{griffin2008} differentiate between the tick and transaction time sampling, where the former samples with respect to transactions with non-zero price changes.}
\end{rmk}

\section{Simulation Study}
\label{sec:Simulation}

We now compare the statistical properties of the RV estimator in \eqref{eq:RVDef} based on different sampling schemes in simulations under general (leverage-type) process and noise specifications.
In addition to validating our theoretical derivations, the aim of the simulation study is to analyze the impact of MMN on the sampling schemes and to quantify the efficiency gains of \emph{intrinsic time} sampling.

We simulate $D=5000$ days with $T=23400$ (seconds) from the TTSV price process
\begin{equation}
    \label{eq:SimPrice}
	dP(t) = \varsigma(t) dB\left({N(t)}\right), \qquad t \in [0,T],
\end{equation} 
where we distinguish the following two settings.

In the first specification, which we denote as the ``independent TTSV process'', $N(t)$ is a doubly stochastic Poisson process independent of $B$.
For the underlying intensities, we use the diffusive specifications,
\begin{align}
	\label{eq:SimLambda}
	\lambda(t) &= \lambda_\text{det}(t) c_\lambda \exp \big(0.01 \lambda^{*}(t) - \bar \lambda^{*}\big), 
	\quad \text{where} \quad
	d\lambda^{*}\left(t\right) = -0.0002 \lambda^{*}(t) dt + dB_1(t), \\
	\label{eq:SimVarsigma}
	\varsigma(t) &= \varsigma_\text{det}(t) c_\lambda^{-1/2} \exp \big(0.005 \varsigma^{*}(t)- \bar \varsigma^{*} \big), 
	\quad \text{where} \quad
	d\varsigma^{*}\left(t\right) = -0.0002 \varsigma^{*}(t) dt + dB_2(t),
\end{align}
for $t \in [0,T]$, where $B_1$ and $B_2$ (and $B$) are independent Brownian motions.
The processes $\lambda(t)$ and $\varsigma(t)$ in \eqref{eq:SimLambda}--\eqref{eq:SimVarsigma} consist of deterministic components $\lambda_\text{det}(t)$ and $\varsigma_\text{det}(t)$ that are the same for every simulated day and give the processes a common characteristic shape, and the multiplicative stochastic diffusions $\lambda^{*}(t)$ and $\varsigma^{*}(t)$ that add some day-by-day randomness.
We obtain the deterministic components $\lambda_\text{det}(t)$ and $\varsigma_\text{det}(t)$ as averages of their estimates using the estimators of \cite{dahlhaus2016}, computed over all trading days of the IBM stock in the year 2018.
The factor $c_\lambda \in \{2000, 8000, 32000\} / \int_0^T \lambda_\text{det}(t) \mathrm{d}t$ in \eqref{eq:SimLambda} allows to control the amount of expected ticks per day to equal $\{2000, 8000, 32000\}$, while its inclusion in \eqref{eq:SimVarsigma} preserves the expected IV, making it invariant to the choice of $c_\lambda$.

The components $\lambda^{*}(t)$ and $\varsigma^{*}(t)$ are Ornstein-Uhlenbeck processes driven by independent Brownian motions $B_i(t)$, $i=1,2$.
Their exponential transformations ensure the positivity of $\lambda(t)$ and $\varsigma(t)$, and the coefficients $\bar \lambda^{*}$ and $\bar \varsigma^{*}$ are the daily averages (over all $t \in [0,T]$) of $\exp (0.01 \lambda^{*}(t))$  and $\exp(0.005 \varsigma^{*}(t))$, respectively, such that the exponential functions have unit mean and serve as multiplicative noise.
We use Euler discretizations with 23400 steps to simulate the diffusions in \eqref{eq:SimPrice}--\eqref{eq:SimVarsigma}.

For the second specification, which we denote as the ``Hawkes-type TTSV process'', $N(t)$ is a Hawkes process with intensity $\lambda(t)$, which, along with the tick variance, is defined as follows
\begin{align}
	\lambda(t) &=  \lambda_\text{det}(t) \widetilde{c}_\lambda \exp \big(0.005\lambda^{*}(t) - \bar \lambda^{*}\big) + \sum_{t_k < t} \nu_\lambda(t-t_k), \label{eqn:HawkesLambda} \\
	\varsigma(t) &= \varsigma_\text{det}(t) \widetilde{c}_\varsigma \widetilde{c}_\lambda^{-1/2} \exp \big(0.0025 \varsigma^{*}(t)- \bar \varsigma^{*} \big) + \sum_{t_k < t} \nu_\varsigma(t-t_k). \label{eqn:HawkesVarsigma}
\end{align}
These intensities extend the specifications in \eqref{eq:SimLambda}--\eqref{eq:SimVarsigma} by incorporating dependent Brownian motions $B_1$ and $B_2$ with a correlation of $0.3$ and, more importantly, by including summands corresponding to self-exciting Hawkes-type intensities with an additional leverage specification \citep{hawkes2018hawkes, laub2021elements}.
For the sequence of jump time $t_1,t_2,\dots$ of the process $N$, and $\Delta P(t_{k}) = P(t_{k}) - P(t_{k-1})$, we set
\begin{align*}
	 \nu_\lambda(t-t_k) &= 
	 \begin{cases} 
	 	0.05 \bar{\lambda}_\text{det} \exp(-0.25 \bar{\lambda}_\text{det} (t-t_k)) \quad &\text{if} \quad \Delta P(t_{k}) > 0, \\
	 	0.1 \bar{\lambda}_\text{det} \exp(-0.25 \bar{\lambda}_\text{det} (t-t_k)) \quad &\text{if} \quad \Delta P(t_{k}) \le 0,
	 \end{cases}
	\\
	\nu_\varsigma(t-t_k) &= 
	\begin{cases} 
		0 \quad &\text{if} \quad \Delta P(t_{k}) > 0, \\
		0.1 \bar{\varsigma}_\text{det} \exp(-0.5 (t-t_k)) \quad &\text{if} \quad \Delta P(t_{k}) \le 0,
	\end{cases}
\end{align*}
where $\bar{\lambda}_\text{det}$ and $\bar{\varsigma}_\text{det}$ are the daily averages (over all $t \in [0,T]$) of  $\lambda_\text{det}(t)$  and  $\varsigma_\text{det}(t)$, respectively.
Here, past price changes have a self-exciting effect on the intensities that declines exponentially with the time elapsed since that observation, $t-t_k$.
Consistent with the classical leverage effect, positive price changes $\Delta P(t_{k}) > 0$ at the previous ticks $t_k$ have a different (weaker) impact than negative price changes $\Delta P(t_{k}) \le 0$.

As above, the constant $\widetilde{c}_\lambda \in \{2000, 8000, 32000\} \cdot (1-\eta) / \int_0^T \lambda_\text{det}(t) \mathrm{d}t$, with $\eta = 0.5 (0.05 \bar{\lambda}_\text{det} + 0.1 \bar{\lambda}_\text{det})/(0.25 \bar{\lambda}_\text{det})$, controls the expected number of ticks per day; see \citet[Eq.~(3.6)]{laub2021elements} for details.
As we are not aware of a closed-form formula for the expected $\varsigma(t)$ to account for the self-exciting effect stemming from the latter sum in \eqref{eqn:HawkesVarsigma}, we choose $\widetilde{c}_\varsigma \approx 0.855, 0.837, 0.741$ for the settings of $2000, 8000$, and $32000$ expected ticks, respectively. 
These choices ensure that all simulation processes have approximately the same expected IV while maintaining control over the expected number of ticks.
For the Hawkes-type intensities in \eqref{eqn:HawkesLambda}--\eqref{eqn:HawkesVarsigma}, we employ the simulation method described in \citet[Algorithm~3.1]{dassios2013exact}.

\begin{figure}[tb] 
	\centering
    \includegraphics[width=1\textwidth]{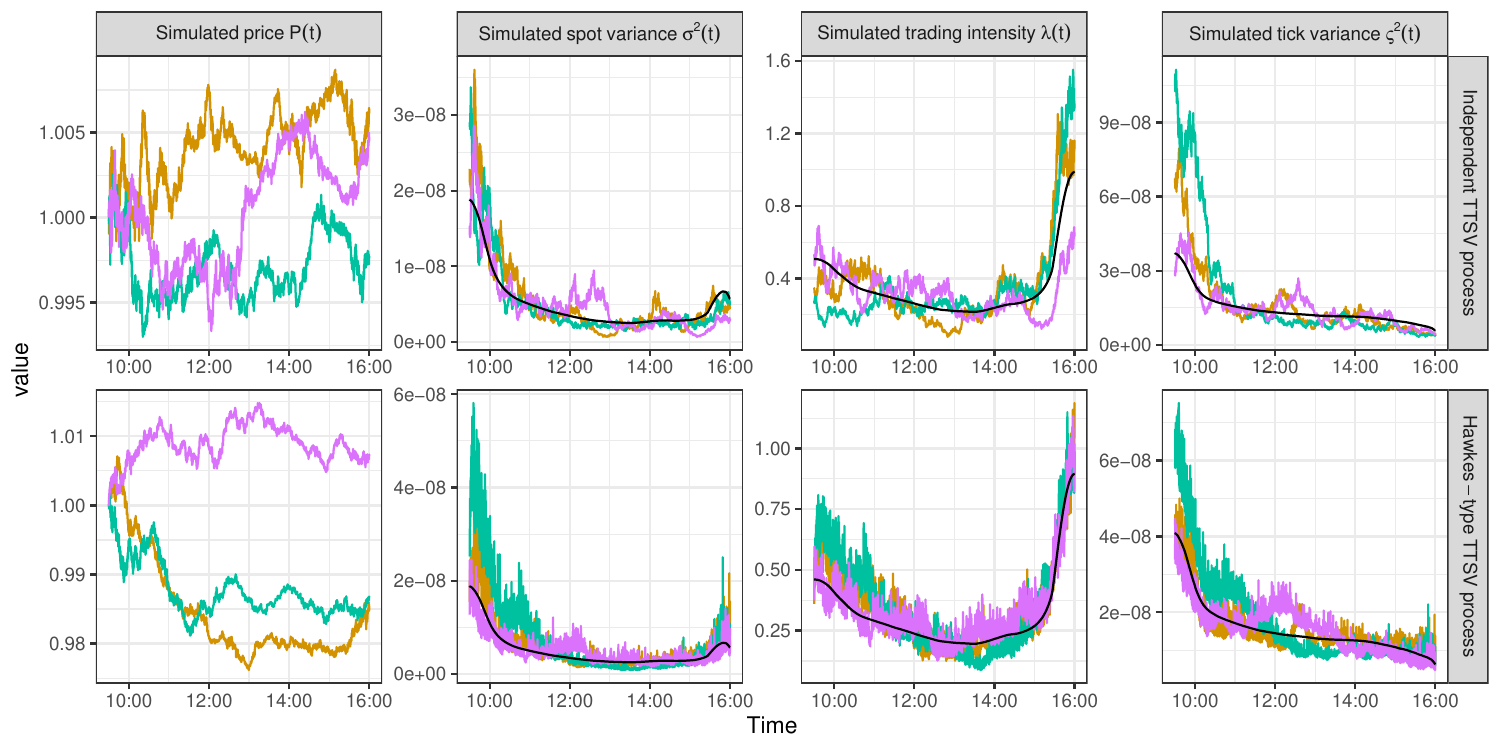}
	\caption{Simulated paths of the asset price as described in Section~\ref{sec:Simulation}, the spot variance $\sigma^2(t)$, the trading intensity $\lambda(t)$, and the tick variance $\varsigma^2(t)$ for three exemplary days in green, orange and pink.
		The black lines show the (appropriately rescaled according to the expected behavior of the Hawkes processes) deterministic components $\lambda_\text{det}(t)$, $\varsigma^2_\text{det}(t)$ and the resulting $\sigma^2_\text{det}(t) = \lambda_\text{det}(t) \, \varsigma^2_\text{det}(t)$ of our simulation setup that are obtained as the estimates from the IBM stock averaged over all tradings days in the year 2018.
	}
	\label{fig:SimExamples}
\end{figure}

The parameters of the two simulation processes above are chosen to mimic real financial data, while also providing sufficient daily variation (across different days) in the simulated intensities $\lambda(t)$ and $\varsigma(t)$, as can be seen from the three exemplary sample paths of $\lambda(t)$, $\varsigma^2(t)$, $\sigma^2(t)$ and $P(t)$ for both processes shown in Figure \ref{fig:SimExamples}.

For both simulation processes, we contaminate the log-price process with either i.i.d.\ or ARMA(1,1) noise with and without a diurnal heteroskedasticity component.
Given the randomly simulated trading times $t_1, \dots, t_{N(T)}$, we set
\begin{align}
	\label{eq:SimulateNoise}
	\widetilde{P}(t_i) = P(t_i) + v_i,
\end{align}
where $v_i$ is independent of all other processes.
For the  i.i.d.\  noise, we let $v_i \stackrel{i.i.d.}{\sim} \mathcal{N}(0, \sigma^2_v)$ for $i=1,\dots, N(T)$, where $\sigma_v = c_{N} \cdot 1.2 \cdot 10^{-4}$.
Here, the factor $1.2 \cdot 10^{-4}$ corresponds to the magnitude of the average tick standard deviation (for the standard setting of 8000 expected ticks per day), and the pre-factor $c_{N} \in \{0, 0.25, 0.5, 1\}$ governs the relative noise level ranging from no noise $c_{N} = 0$ to a high noise setting $c_{N} = 1$, where the noise variance equals the average tick variance.
In the results below, we refer to the factor $c_N$ by writing ``$100 \cdot c_N \%$ noise''.
We emphasize that our ``$100 \%$ noise'' setting is consistent with the findings and simulation setups of \citet{Jacod2017} and \citet{Li2022remedi}.\footnote{In more detail, our $100 \%$ i.i.d.\ noise setting employs a noise standard deviation of $\sigma_v = 1.2 \cdot 10^{-4}$ for values of $\sqrt{\IV} \approx 1.1 \cdot 10^{-2}$.
In contrast, \citet[Section~4.1]{Jacod2017} use the much higher estimated noise standard deviation from their Figure~9 of approximately $5.6 \cdot 10^{-4}$ for Citigroup data in the year 2011 in relation to values of $\sqrt{\IV}$ of around $10^{-2}$.
Moreover, \citet[Figure~5]{Li2022remedi} obtain noise standard deviation estimates of approximately $\{0.7, 1.1\} \cdot 10^{-4}$ (obtained as the square root of the autocovariance function at lag 0) for the Coca-Cola stock in the year 2018, where the pre-factors $\{0.7, 1.1\}$ refer to two different noise estimators.}

For the ARMA noise process, we let $v_i =\varepsilon_{i} +  0.5 v_{i-1}  + 0.5 \varepsilon_{i-1}$, where $\varepsilon_{i} \sim \mathcal{N}(0, \sigma^2_{\varepsilon,i})$, and $\sigma^2_{\varepsilon,i}$ is either constant or follows a diurnal V-shaped piecewise linear function.
The latter assigns double the variance at market opening and closing compared to the middle of the trading day, following \citet{Kalnina2008} and \citet{Jacod2017}.
For each of the five choices in $c_{N} $, we specify $\sigma^2_{\varepsilon,i}$ such that the average standard deviation of $v_i$ over the day equals $c_{N} \cdot 1.2 \cdot 10^{-4}$ to make it comparable in magnitude to the i.i.d.\ noise setting.

For all sampling schemes except HTS, we fix the value of $M$ by using information on the respective accumulated intensity $\Phi(T)$ at the end of each trading day in \eqref{eq:ChoosingTau_j}.
While this formally violates the stopping-time condition \eqref{eqn:sampling_scheme} in Theorems~\ref{thm:MSE_IV} and \ref{thm:EfficientSampling}, we illustrate in Figure~\ref{fig:RMSE_Stopping} that the results are invariant to this violation.
As fixing $M$ is not possible for the HTS scheme, we fix $\delta$, for which we choose a sequence of 17 values ranging from approximately $0.00022$ to $0.0054$.
These values yield reasonable sampling frequencies allowing for a comparison with the other sampling schemes.
Note that for HTS and a fixed $\delta$, the number of samples per days is random and can vary substantially across trading days.

While the CTS and rTTS schemes can be implemented straightforwardly, the iTTS, iBTS and rBTS schemes require the intensities $\lambda(t)$, $\varsigma^2(t) \lambda(t)$, and  $\varsigma^2(t)$, respectively.
For this, we use rolling averages over the past $50$ trading days of the nonparametric estimators $\widehat{\lambda}(t)$, $\widehat{\lambda}(t) \,  \hat{\varsigma}^2(t)$ and $\hat{\varsigma}^2(t)$, respectively, which are proposed in \citet{dahlhaus2016}, who also show consistency of these estimators under i.i.d.\ noise.

\begin{figure}[tb]
	\centering
	\includegraphics[width=1\textwidth]{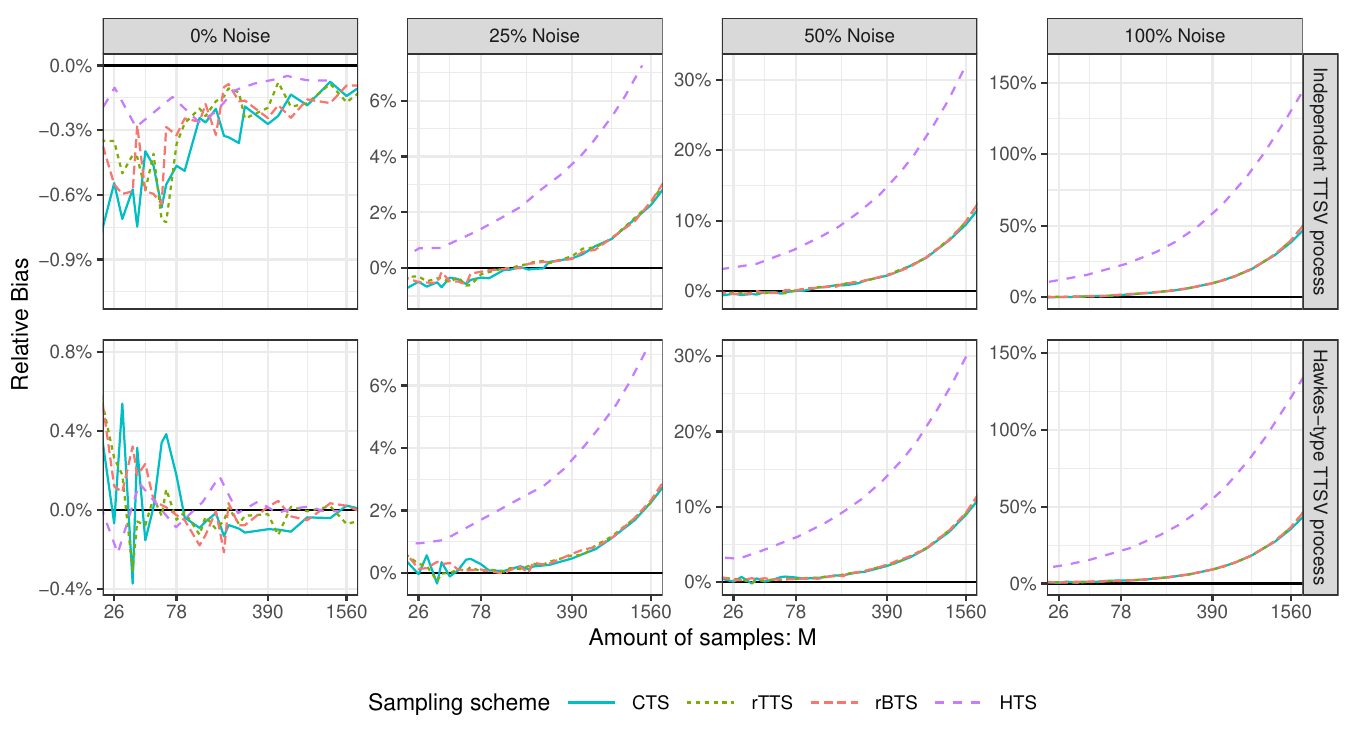}
	\caption{
		Relative bias (in percent) of the RV estimator using different sampling schemes in color plotted against the (for HTS average) sampling frequencies $M$ on the horizontal axis.
        The plot columns refer to the noise magnitude described below \eqref{eq:SimulateNoise} and the plot rows refer to the two process specifications described after \eqref{eq:SimPrice}.
        }
	\label{fig:Bias_Processes}
\end{figure}

Figure \ref{fig:Bias_Processes} shows the relative bias, i.e., the bias standardized by the respective daily value of IV, of the RV estimator for the considered sampling schemes, a range of $M$ values, and for the two process specifications\footnote{For the Hawkes-type TTSV-process, we compare the estimated RV values against the \emph{realized} IV, which can easily be computed as $\rIV(0,T) = \int_0^T\varsigma^2(r)dN(r) = \sum_{0 \leq t_i \leq T} \varsigma^2(t_i)$. In contrast, $\IV(0,T) = \int_0^T \varsigma^2(r) \lambda(r) dr$ is much more difficult to approximate in our simulations due to the combination of a continuous time diffusion with the Hawkes-type jumps with exponential decays defined in \eqref{eqn:HawkesLambda}--\eqref{eqn:HawkesVarsigma}. Note that $\E[\rIV(0,T)] = \E[\IV(0,T)]$.} described above.
Results are shown for four magnitudes of i.i.d.\ noise and values of $c_\lambda$ and $\widetilde{c}_\lambda$ that yield 8000 expected ticks per day.

For the specification without noise, we can confirm the unbiasedness of the RV estimator of Theorem~\ref{thm:unbiasedness} for all sampling schemes and both process specifications.
For an increasing amount of noise, the RV estimator exhibits the usual positive bias that grows with the sampling frequency.
Notably, the HTS sampling scheme reacts more strongly to increasing noise levels, even for the lowest considered sampling frequencies, where the other sampling schemes are (almost) unbiased.
Importantly, the results hold equivalently for both the independent and the Hawkes-type TTSV processes, thereby illustrating the broad applicability of Theorem~\ref{thm:unbiasedness}.

We continue to shed light on the increased bias under noise of the HTS scheme:
Using the notation $r(s,t) = P(t) - P(s)$ and $\widetilde{r}(s,t) = \widetilde{P}(t) - \widetilde{P}(s)$, heuristic arguments for the RV estimator under noise, $\widetilde{\RV}(\btau)$, yield
\begin{align}
    \widetilde{\RV}(\btau) &= \sum_{j=1}^M \widetilde{r}(\tau_{j-1}, \tau_j)^2 \nonumber \\ 
    &= \sum_{j=1}^M r(\tau_{j-1}, \tau_j)^2 + \sum_{j=1}^M (v_{N(\tau_j)} - v_{N(\tau_{j-1})})^2 + 2 \sum_{j=1}^M r(\tau_{j-1}, \tau_j) (v_{N(\tau_j)} - v_{N(\tau_{j-1})}) \nonumber \\
    & = \mathrm{IV}(0,T) + \mathcal{O}_P\big(M^{-1/2}\big) \nonumber \\
    &\qquad \quad + \sum_{j=1}^M (v_{N(\tau_j)} - v_{N(\tau_{j-1})})^2 + 2 \sum_{j=1}^M r(\tau_{j-1}, \tau_j) (v_{N(\tau_j)} - v_{N(\tau_{j-1})}).
    \label{eqn:BiasApproxTerms}
\end{align}
In the following, we ignore the asymptotically vanishing $\mathcal{O}_P\big(M^{-1/2}\big)$ term arising from a standard central limit theorem for the (noise-free) RV estimator.
Then, \eqref{eqn:BiasApproxTerms} indicates that the bias is driven by two terms: the \emph{variance} of the noise differences at the sampling points and the \emph{covariance} between the sampled (noise-free, efficient) returns and the noise differences.

\begin{figure}[tb]
	\centering
	\includegraphics[width=1\textwidth]{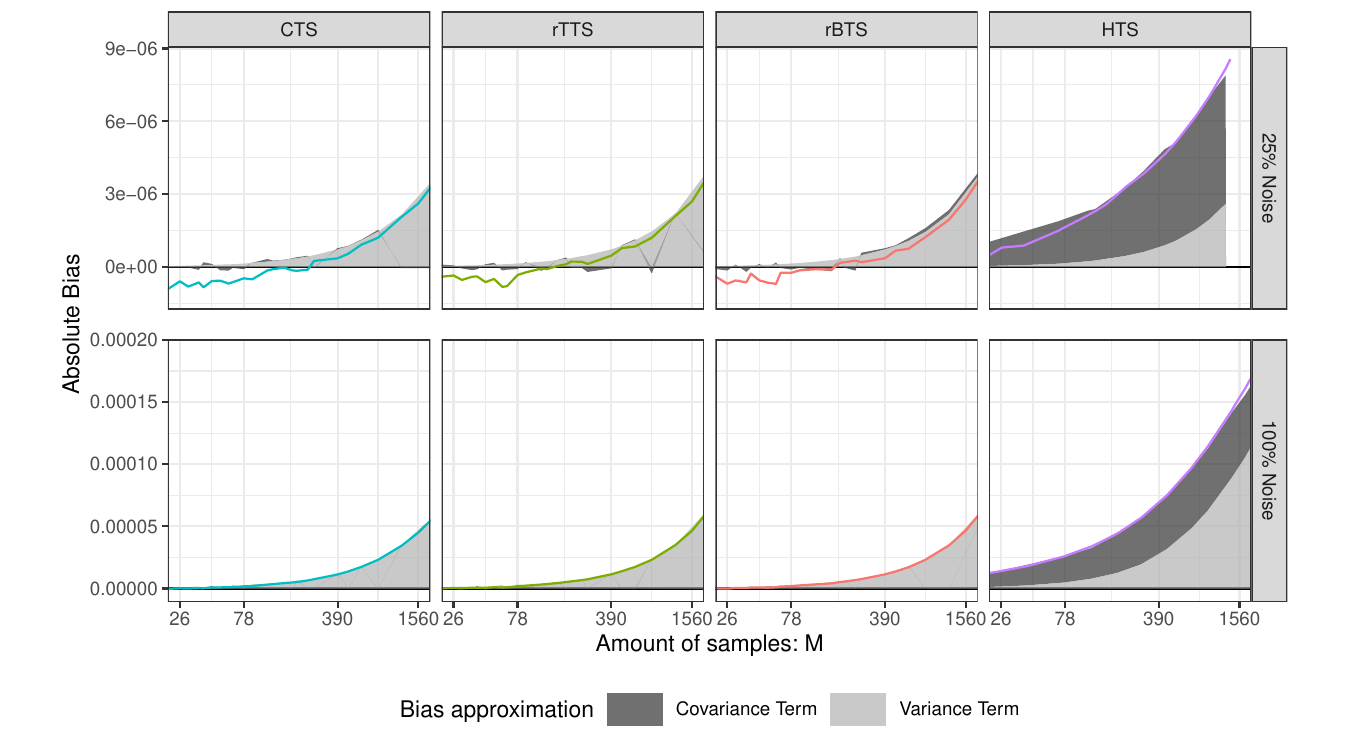}
    \caption{
    	Bias of the RV estimator in the independent TTSV process using different sampling schemes in the plot columns (and in color) plotted against the (for HTS average) sampling frequencies $M$ on the horizontal axis.
        The gray areas depict the ``variance'' and ``covariance'' terms from the bias approximation in \eqref{eqn:BiasApproxTerms}, estimated from corresponding simulations.
        The plot rows refer to two different noise magnitudes described below \eqref{eq:SimulateNoise}.
        }
	\label{fig:Bias_Approximation}
\end{figure}

Figure~\ref{fig:Bias_Approximation} displays the bias for the four sampling schemes under the independent TTSV process with 8000 expected ticks per day and 25\% or 100\% i.i.d.\ noise. 
\edit{The colored lines represent the empirical bias obtained from the simulations, i.e., these lines match the respective lines from the second and fourth plot in the upper panel of Figure~\ref{fig:Bias_Processes}.}
The shaded gray areas correspond to the two approximation terms from \eqref{eqn:BiasApproxTerms}, which help explain the sampling-scheme-dependent differences in bias.
We estimate these terms from the simulated data according to the formulas in \eqref{eqn:BiasApproxTerms}.
While the variance term is of a similar magnitude for all sampling schemes, the HTS scheme stands out \edit{as the only scheme} with a notably large positive covariance term---the main cause of HTS’s elevated bias, as we explain in the following.

For CTS, rTTS, and rBTS, the efficient returns $r(\tau_{j-1}, \tau_j)$ are independent of the noise terms as the sampling points do not depend on the noise on the given day.
\edit{In contrast, HTS determines the next sampling time $\tau_j$ as the first time point $t \ge \tau_{j-1}$, where the absolute noisy price change, $\big| \widetilde{r}(\tau_{j-1}, t) \big| = \big| r(\tau_{j-1}, t) + (v_{N(t)} - v_{N(\tau_{j-1})}) \big|$ exceeds $\delta$.}\footnote{As our price process in \eqref{eq:TTSV_model} (such as real prices at financial markets) generates discrete price paths that are only observed at the realizations of $N$, the absolute values of the HTS returns slightly overshoots the threshold $\delta$ as can be seen in Figure~\ref{fig:HTS_Overshooting}. 
As shown in Theorem~\ref{thm:unbiasedness} that applies to arbitrary $\Fil$-adapted sampling schemes, this should not be the underlying reason for the increased bias of HTS.}
\edit{
    Hence, given a fixed previous sampling point $\tau_{j-1}$, HTS is particularly likely to sample at time points $\tau_{j}$ for which the two quantities $r(\tau_{j-1}, \tau_j)$ and $(v_{N(\tau_j)} - v_{N(\tau_{j-1})})$ share the same sign, and hence accumulate in the noisy return $\widetilde{r}(\tau_{j-1}, \tau_j)$.
    This behavior results in a positive covariance term in \eqref{eqn:BiasApproxTerms} and in our simulations, we observe associated correlations ranging between $0.15$ and $0.5$ for HTS.
}

\begin{figure}[tb]
	\centering
	\includegraphics[width=1\textwidth]{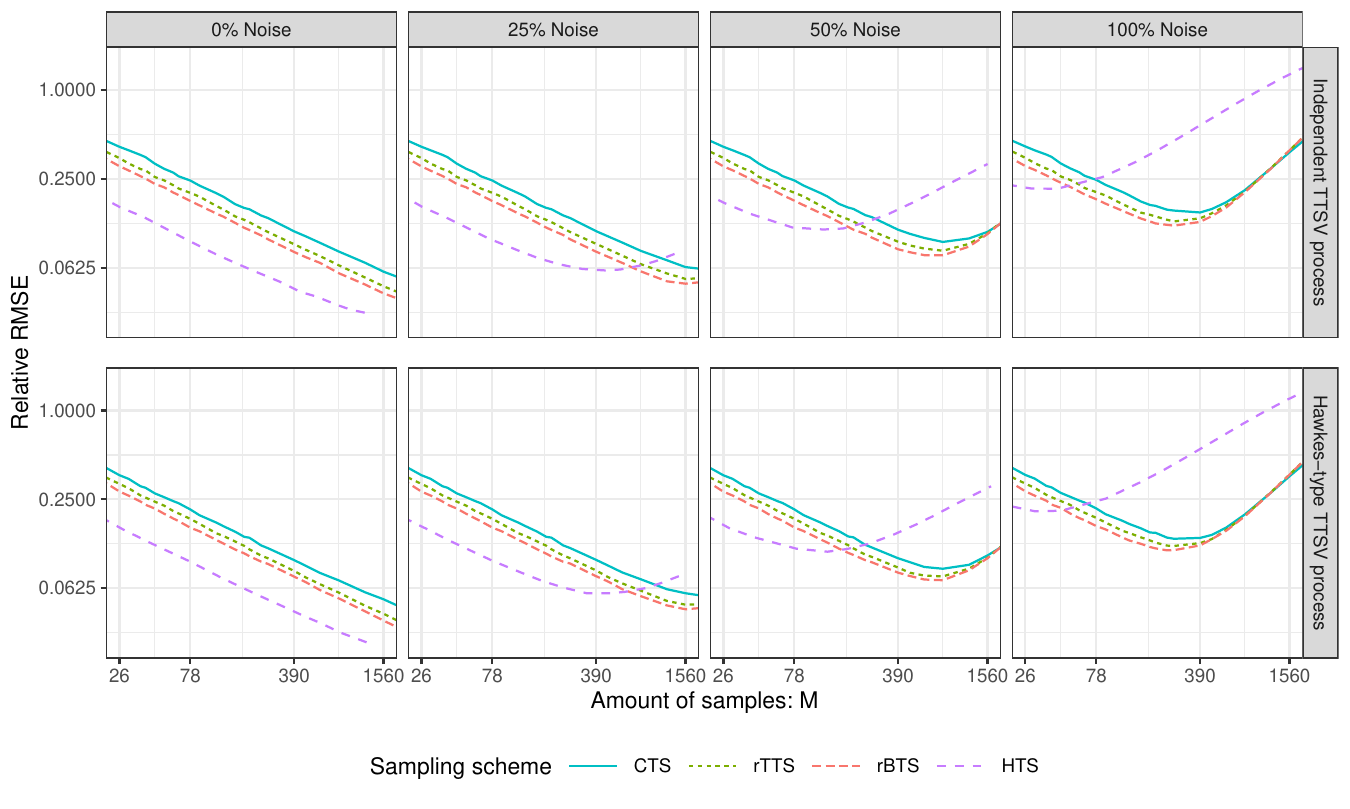}
	\caption{
		Relative RMSE of the RV estimator using different sampling schemes in color plotted against the (for HTS average) sampling frequencies $M$ on the horizontal axis.
        The plot columns refer to the noise magnitude described below \eqref{eq:SimulateNoise} and the plot rows refer to the two process specifications described after \eqref{eq:SimPrice}.}
	\label{fig:RMSE_Processes}
\end{figure}

Figure \ref{fig:RMSE_Processes} presents the relative RMSE of the RV estimator\footnote{
	The relative RMSE over the trading days $d=1,\dots,D$ is formally given as
	\begin{align*}
		\frac{\sqrt{\sum_{d=1}^D \big( \RV_d(\btau) - \IV_d(0,T) \big)^2}}{\sum_{d=1}^D  \IV_d(0,T)},
	\end{align*}
	ensuring that the square root and the normalization are taken ``outside'' of the MSE.
    This way, the plots indeed analyze the MSE while presenting results in a conveniently interpretable scale.}
for  the different sampling schemes.
As in Figure \ref{fig:Bias_Processes}, we show results for both simulation processes and four noise levels in the subplots.
In the absence of noise, HTS clearly yields the lowest RMSE across all sampling frequencies and both process specifications, as implied by Theorem~\ref{thm:EfficientSampling}.
Furthermore, rBTS and rTTS also improve upon the classically used CTS scheme, in line with part (b) of Theorems~\ref{thm:EfficientSampling}.
As the noise level increases, the RMSE rises across all sampling schemes and frequencies, reflecting the growing bias illustrated in Figures \ref{fig:Bias_Processes}--\ref{fig:Bias_Approximation}.

The pronounced bias for the HTS scheme leads to the finding that, as the sampling frequency increases, rBTS yields RV estimates with lower RMSE than HTS.
The crossing point at which rBTS becomes more efficient than HTS primarily depends on the noise magnitude and ranges from $M \approx 780$ to $M \approx 39$, corresponding to sampling frequencies between 30 seconds and 10 minutes.
Similar to the bias, the MSE results are very similar for the independent and the Hawkes-type TTSV processes, hence illustrating the broad applicability of Theorem~\ref{thm:EfficientSampling}.
This observation also supports the insight from Appendix~\ref{sec:RemainderTerms} that the remainder terms in Corollaries~\ref{cor:MSE_jump_based} and \ref{cor:MSE_intensity_based} are approximately equal across sampling schemes, even under mild dependence.

Appendix~\ref{sec:AdditionalResults} contains additional simulation results summarized as follows:
First, Figure~\ref{fig:RMSE_TickPerDay} analyzes the effects of a varying expected number of $\{2000, 8000, 32000\}$ trades per day while keeping the expected IV unchanged.
Under noise, HTS performs worse as the number of ticks increases, which is mainly explained by an increased relationship of the noise relative to $\varsigma(t)$:
More ticks are generated through a higher level of $\lambda(t)$, which results in a lower $\varsigma(t)$ as the expected IV is held constant.
Second, Figure~\ref{fig:RMSE_NoiseProcess} illustrates that our results are robust to the standard and diurnal ARMA noise specifications.
Third, Figure~\ref{fig:RMSE_Intensity} confirms parts (b) and (c) of Theorem~\ref{thm:EfficientSampling}, i.e., that the \emph{realized} TTS and BTS sampling variants outperform the \emph{intensity} variants, and that using the true (oracle) intensities yields slightly better RV estimation performance than using their estimated counterparts.
Fourth, Figure~\ref{fig:RMSE_Stopping} shows that employing stopping-times for rTTS and rBTS, as opposed to fixing $M$ (see Section~\ref{sec:SamplingSchemes}), produces essentially the same RMSE results.

\section{Empirical Applications}
\label{sec:Application}

We start to illustrate the gains in estimation accuracy that HTS and rBTS entail for the RV estimator in Section \ref{sec:appl_Estimation}, and continue to analyze different sampling schemes in a forecasting environment in Section \ref{sec:appl_Forecasting}.

\subsection{Comparing Estimation Accuracy}
\label{sec:appl_Estimation}

In this application, we assess the estimation accuracy of the RV estimator for the different sampling schemes using data on 27 liquid stocks from the NYSE TAQ database.\footnote{We use the 27 stocks with the ticker symbols  AA, AXP, BA, BAC, CAT, DIS, GE, GS, HD, HON, HPQ, IBM, IP, JNJ, JPM, KO, MCD, MMM, MO, MRK, NKE, PFE, PG, UTX, VZ, WMT, and XOM.}
We filter the raw prices according to \citet[Section 3]{barndorff2009}.
Based on the filtered prices, we compute the five sampling schemes CTS, rTTS, iBTS, rBTS, and HTS as described in Section \ref{sec:SamplingSchemes}.
We use all trading days from January 1, 2012, to March 31, 2019, for evaluating the estimation accuracy and up to 50 trading days before January 1, 2012, to estimate the intensities required for the iBTS and rBTS methods.
We estimate the underlying trading intensity and tick variance with the non-parametric and noise-robust estimators of \citet{dahlhaus2016} and average the estimated intensities over the past 50 trading days in a rolling fashion.

For the above sampling schemes, we choose a fixed number of $M \in \{13,26,39,78,130,260,390\}$ log-returns per day, which correspond to intrinsic time sampling frequencies of $390/M$ minutes. 
As in the simulations, fixing $M$ is done using the information on $\Phi(T)$ available at the end of each trading day.
For HTS, however, fixing the threshold $\delta$ leads to a random number of samples $M_\delta$ per day, which can vary considerably. 
To address this variability, we proceed as follows: 
For each $M$, asset, and trading day, we select the HTS result corresponding to the threshold $\delta$ for which the realized $M_\delta$ is closest to the given $M$. 
For $\delta$, we use 29 equally spaced values for $\log_{10}(\delta)$ between $-3.7$ and  $-2.3$.
Table~\ref{tab:applSigPosNegValuesAggregation} shows that averaging $M_\delta$ over time and assets before matching to $M$ does not meaningfully change the results for HTS.\footnote{\label{fn:HTS_Matching}Table~\ref{tab:applSigPosNegValuesAggregation} also shows results when we (i) match \emph{monthly averages} by averaging $M_\delta$ over all days within each month before matching to the $M$-grid; (ii) use \emph{all-time averaging} over all trading days in the sample; and (iii) apply \emph{all-time and asset-wise averaging} across all days and assets.}

We evaluate the competing RV estimators with the data-based ranking method of \citet{Patton2011RV}, which addresses the challenge that the estimation target, IV, is not observable, even ex post.
Specifically, we use the subsequent trading day’s IV estimate as a proxy, assuming it is unbiased but noisy.
By using a future RV estimator as the proxy, the method of \cite{Patton2011RV} ``breaks'' the correlation between the estimation errors of the RV estimators under consideration and the proxy.
In practice, one should use an unbiased proxy that is unlikely to be affected by MMN.
While choosing a potentially inefficient estimator still gives an asymptotically valid test, its power might be lower \citep{LiuPattonSheppard2015, HogaDimi2022}.
To balance these points, we set the proxy to the next day`s RV computed from 5 minute CTS returns throughout our analysis.
Using different reasonable choices for the proxy such as sampling frequencies of 1, 10, or 15 minutes, \edit{or daily squared returns (see Figures~\ref{fig:appl_RV_eval_vsCTS_SqRet} and \ref{fig:appl_RV_eval_vsrBTS_SqRet}),} does not meaningfully change our results.
We test for significance of the pairwise loss differences with respect to a benchmark estimator to be specified below (which is in general different from the proxy) by using the  \cite{DieboldMariano1995} test, with inference drawn by using the stationary bootstrap of \citet{PolitisRomano1994} that is shown to be valid in this setting by \citet[Proposition 2]{Patton2011RV}.

\begin{table}[tb]
	\centering  
	\begin{tabular}{llrrlrrllrrlrrrr}
		\toprule
        \multicolumn{7}{c}{Sampling vs. CTS} & $\qquad$ &  \multicolumn{7}{c}{Sampling vs. rBTS}  \\
		\cmidrule{1-7}  	\cmidrule{9-15} 
		&& \multicolumn{2}{c}{MSE} &&  \multicolumn{2}{c}{QLIKE} &&&& \multicolumn{2}{c}{MSE} &&  \multicolumn{2}{c}{QLIKE} \\
		\cmidrule{3-4}  	\cmidrule{6-7}  	\cmidrule{11-12} 	\cmidrule{14-15} 	
	   Sampling &  & pos & neg  &  & pos & neg& & Sampling & &pos & neg &  & pos & neg \\ 
		\midrule
        &&&&&&&& CTS &  & 0 & 56 &  & 2 & 90 \\  
        rTTS &  & 46 & 0 &  & 64 & 8 && rTTS &  & 3 & 42 &  & 0 & 89 \\ 
        iBTS &  & 43 & 1 &  & 95 & 0 && iBTS &  & 4 & 29 &  & 14 & 27 \\ 
        rBTS &  & 56 & 0 &  & 90 & 2 \\ 
        HTS &  & 56 & 3 &  & 86 & 4  && HTS &  & 33 & 19 &  & 73 & 10 \\ 
		\bottomrule
	\end{tabular}
    \caption{Percentage values of significantly positive (``pos'') and negative (``neg'')  MSE and QLIKE loss differences between the sampling schemes mentioned in the column ``Sampling'' against the one in the title using the method of \cite{Patton2011RV}.
	The percentage values are computed over the 27 assets and the seven employed values of $M$ for the respective estimators.}
	\label{tab:applSigPosNegValues}
\end{table}

Table \ref{tab:applSigPosNegValues} summarizes the results by reporting the percentage of significantly positive and negative loss differences (at the $5\%$ level) compared to the baseline sampling schemes, aggregated across the 27 assets and the seven considered sampling frequencies.
We use CTS and rBTS as the baseline schemes for comparison in the two panels: CTS as the most commonly employed sampling method in the literature, and rBTS to enable a direct comparison to HTS, as motivated by our simulation results.
We deliberately compare estimators with the same sampling frequency \emph{across} sampling schemes as a direct comparison of sampling schemes is the main focus of the paper.
The table shows results based on both the MSE and QLIKE loss functions.

\begin{figure}[p]
	\centering
	\includegraphics[width=0.9\textwidth]{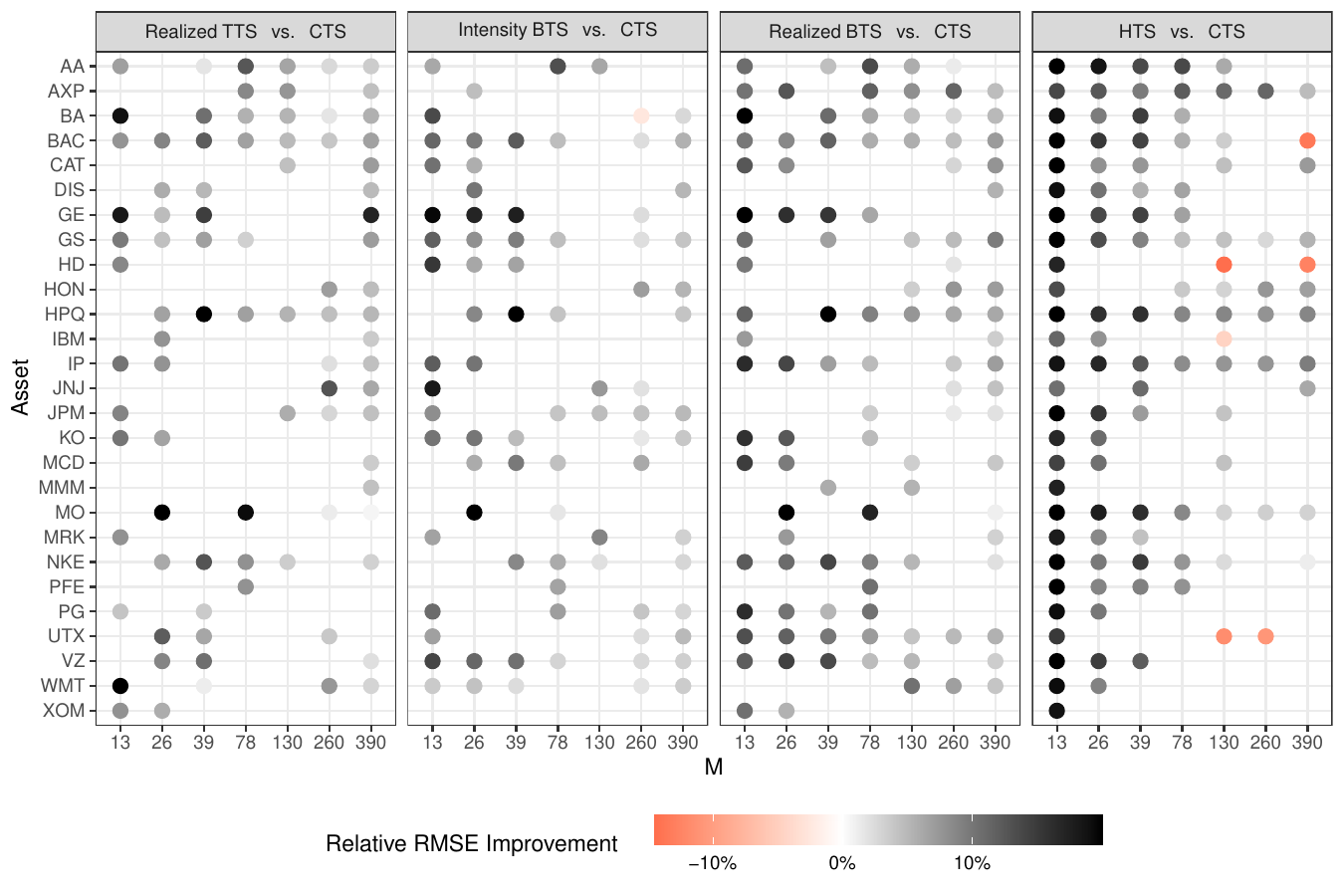}
	\includegraphics[width=0.9\textwidth]{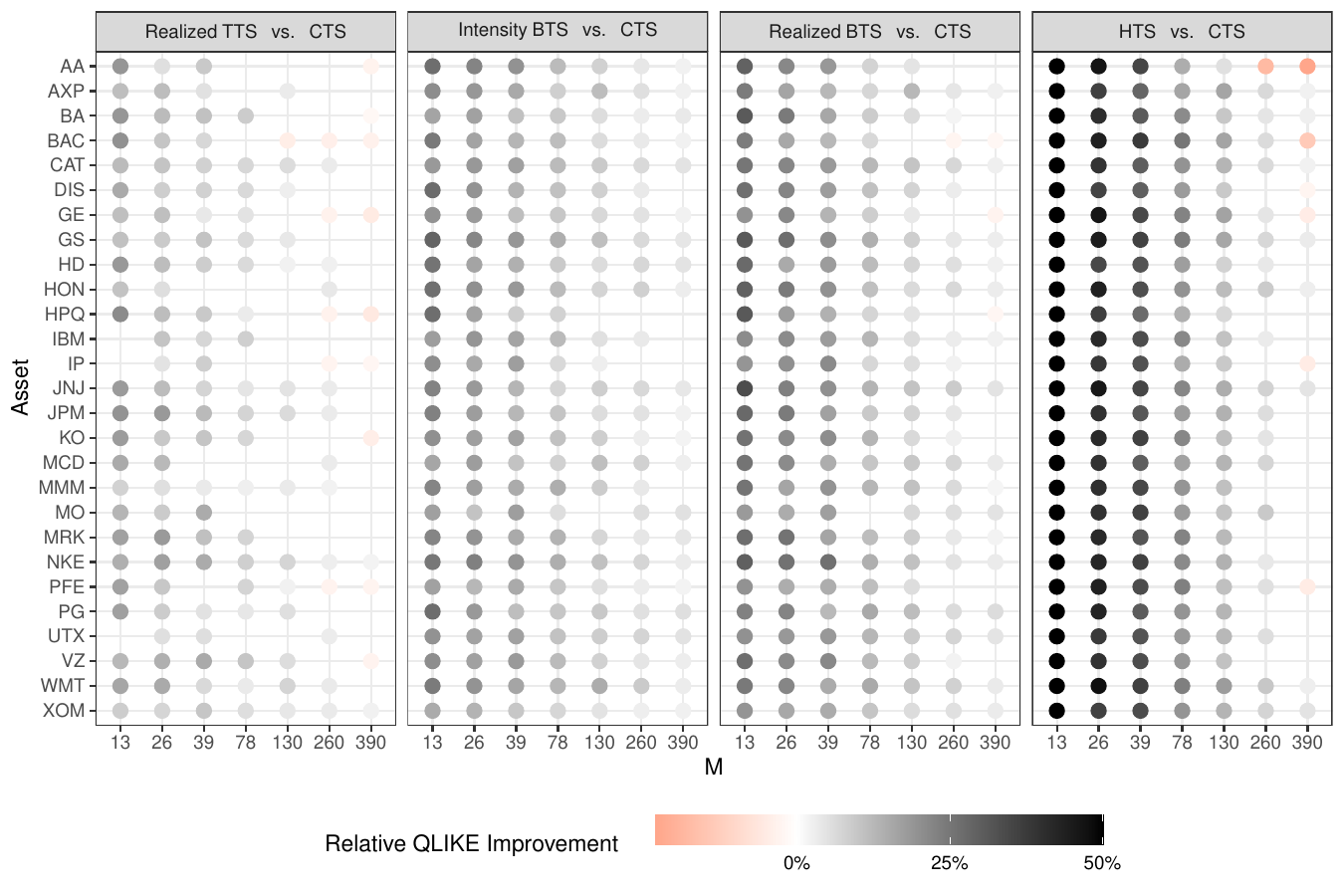}
	\caption{RMSE (top) and QLIKE (bottom) loss differences for the RV estimator based on different sampling schemes and a range of sampling frequencies $M$ for the 27 considered assets.
	Each point corresponds to a (at the $5\%$ level) significant  loss difference of the corresponding RV estimator to a \emph{benchmark CTS RV estimator} with the same sampling frequency.
	Insignificant loss differences are omitted.
	The color scale of the points shows the relative improvement in terms of RMSE/QLIKE, where black (red) colors refer to an improvement (decline).}
	\label{fig:appl_RV_eval_vsCTS}
\end{figure}

\begin{figure}[p]
	\centering
	\includegraphics[width=0.9\textwidth]{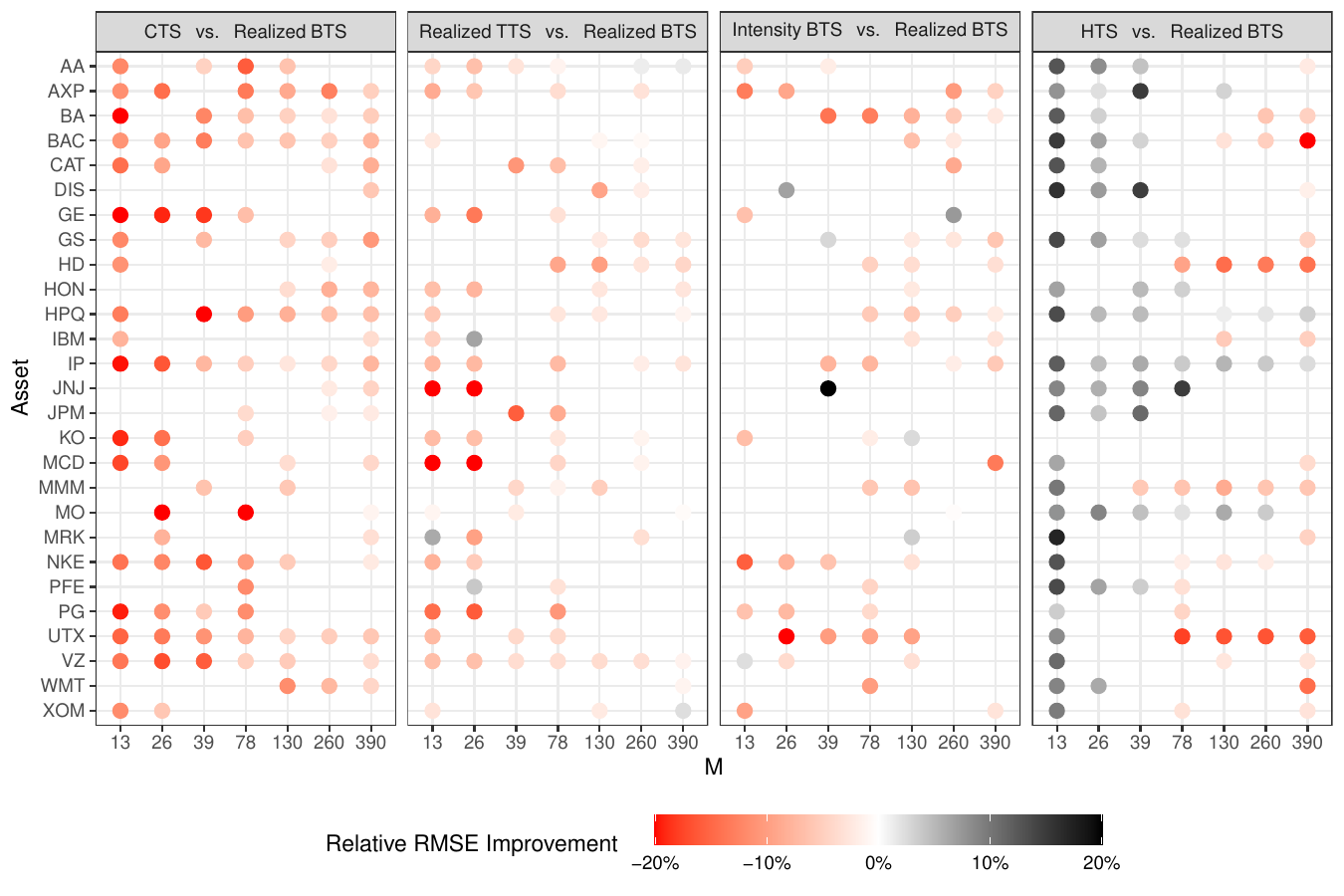}
	\includegraphics[width=0.9\textwidth]{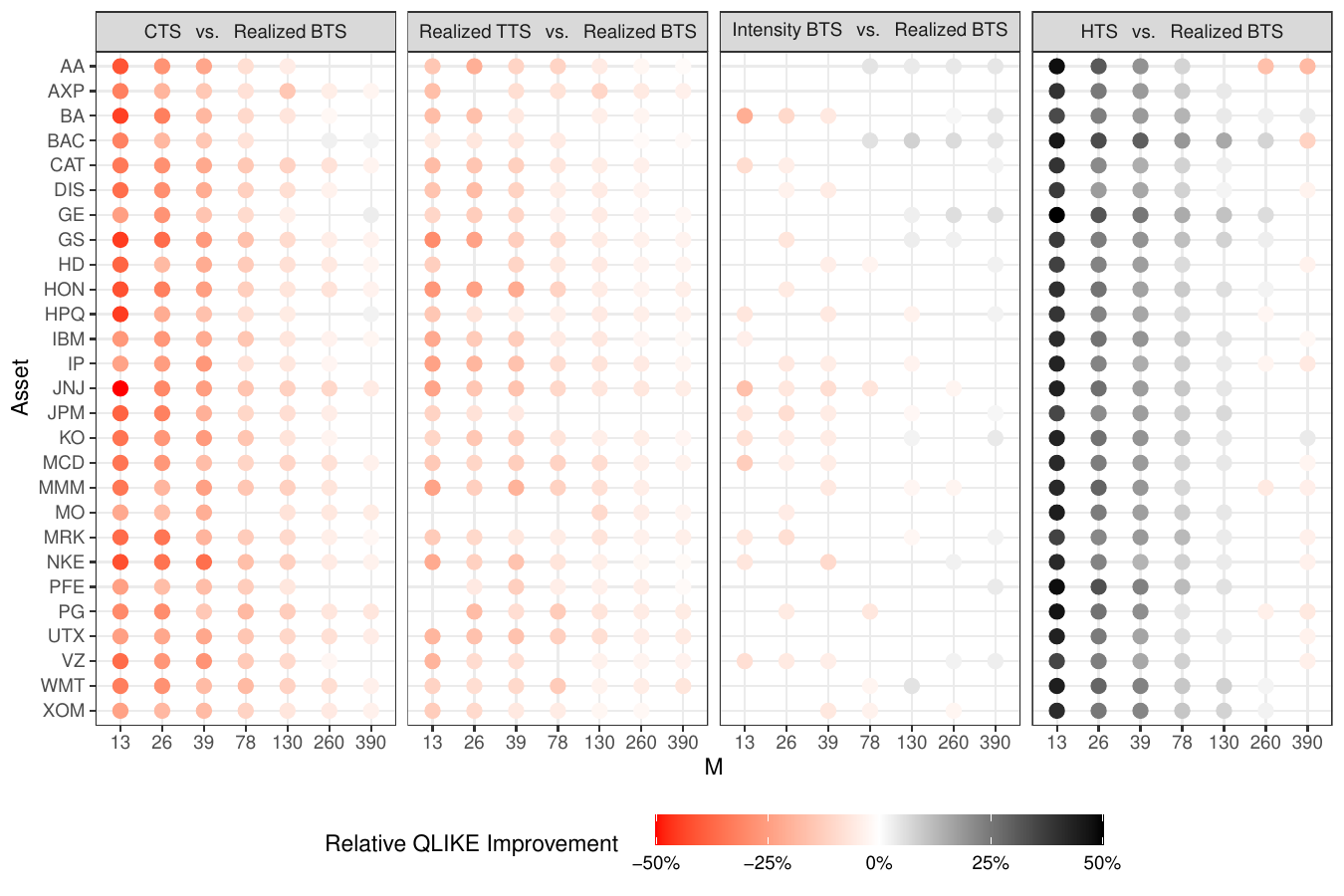}
	\caption{RMSE (top) and QLIKE (bottom) loss differences for the RV estimator based on different sampling schemes and a range of sampling frequencies $M$ for the 27 considered assets.
	Each point corresponds to a (at the $5\%$ level) significant  loss difference of the corresponding RV estimator to a \emph{benchmark rBTS RV estimator} with the same sampling frequency.
	Insignificant loss differences are omitted.
	The color scale of the points shows the relative improvement in terms of RMSE/QLIKE, where black (red) colors refer to an improvement (decline).}
	\label{fig:appl_RV_eval_vsrBTS}
\end{figure}

Detailed results for each asset and sampling frequency are given in Figures~\ref{fig:appl_RV_eval_vsCTS} and \ref{fig:appl_RV_eval_vsrBTS}, comparing to CTS and rBTS as the baseline schemes, respectively.
The upper panels show RMSE and the lower panels QLIKE results.
Black (red) points indicate that the considered estimator is significantly better (worse) than the benchmark at the 5\% level; absence of a point denotes an insignificant difference.
The color intensity indicates the magnitude of the relative improvements in RMSE (capped at $\pm 20\%$) or in QLIKE (capped at $\pm 50\%$).

When comparing the more elaborate (rTTS, iBTS, rBTS, HTS) sampling schemes against the baseline CTS scheme in Figure \ref{fig:appl_RV_eval_vsCTS} and the left panel of Table \ref{tab:applSigPosNegValues}, we observe far more significantly positive than negative loss differences.
This pattern is even more pronounced for the QLIKE loss function, relating to the known fact that evaluation results are often more stable for QLIKE than for MSE loss \citep{Patton2011}.
Figure \ref{fig:appl_RV_eval_vsCTS} further shows that the increases are particularly pronounced at lower sampling frequencies, which are still regularly used in empirical work such as in  \citet{LiuPattonSheppard2015, bollerslev2018risk, bollerslev2020good, bollerslev2022zero, bates2019crashes, bucci2020realized, reisenhofer2022harnet, alfelt2023singular, Patton2023Bespoke}.
Consistent with our simulation findings, the most frequent and substantial improvements can be observed for the HTS (at lower frequencies) and the rBTS schemes.

Figure~\ref{fig:appl_RV_eval_vsrBTS} and the right panel of Table~\ref{tab:applSigPosNegValues} show that rBTS consistently outperforms CTS, rTTS, and iBTS, with efficiency gains again being more pronounced under the QLIKE loss.
The direct comparison between rBTS and HTS reveals that, in line with our simulation results, HTS dominates rBTS at lower sampling frequencies below 5 minutes ($M \le 78$), where noise has a negligible effect.
In contrast, rBTS outperforms HTS at frequencies above 5 minutes ($M > 78$) for most of the considered stocks.

\edit{
To assess how our sparsely sampled RV estimators perform compared to a state-of-the-art noise-robust benchmark, Figure~\ref{fig:appl_RV_eval_vs_PAVG} compares them to the pre-averaging RV of \citet{Jacod2009}, computed from all tick-level data with non-zero price changes and with a bandwidth of $0.5 \sqrt{m_\text{ticks}}$, where $m_\text{ticks}$ is the daily number of ticks.
Because the pre-averaging RV is independent of the sampling frequency $M$, all sampling-based RV estimators (for different $M$) are compared to a single pre-averaging estimator in Figure~\ref{fig:appl_RV_eval_vs_PAVG}. The resulting presentation therefore differs slightly from Figures~\ref{fig:appl_RV_eval_vsCTS} and \ref{fig:appl_RV_eval_vsrBTS}.
For the evaluation proxy, we use daily squared returns, since other choices---either a sparsely sampled CTS RV in Figure~\ref{fig:appl_RV_eval_vs_PAVG_CTSProxy} or the pre-averaging RV in Figure~\ref{fig:appl_RV_eval_vs_PAVG_PAVGProxy}---can bias the results.
As noted above, using daily squared returns reduces the test’s power but avoids this undesired sensitivity.}

\edit{
Figure~\ref{fig:appl_RV_eval_vs_PAVG} shows that our sparsely sampled RV estimators slightly outperform the pre-averaging estimator, particularly the rTTS and rBTS variants at sampling frequencies between $M = 78$ and $M = 390$. Although the HTS estimator exhibits some advantages at very low frequencies ($M < 78$) over the other sampling schemes in Figures~\ref{fig:appl_RV_eval_vsCTS} and \ref{fig:appl_RV_eval_vsrBTS}, RV at these frequencies does not outperform the pre-averaging benchmark in Figure~\ref{fig:appl_RV_eval_vs_PAVG}.
Our overall findings with respect to the pre-averaging RV estimator are consistent with the empirical study of \citet{LiuPattonSheppard2015}, who find that the classical RV estimator is difficult to outperform in practice.}

\begin{figure}[tb]
	\centering
	\includegraphics[width=0.9\textwidth]{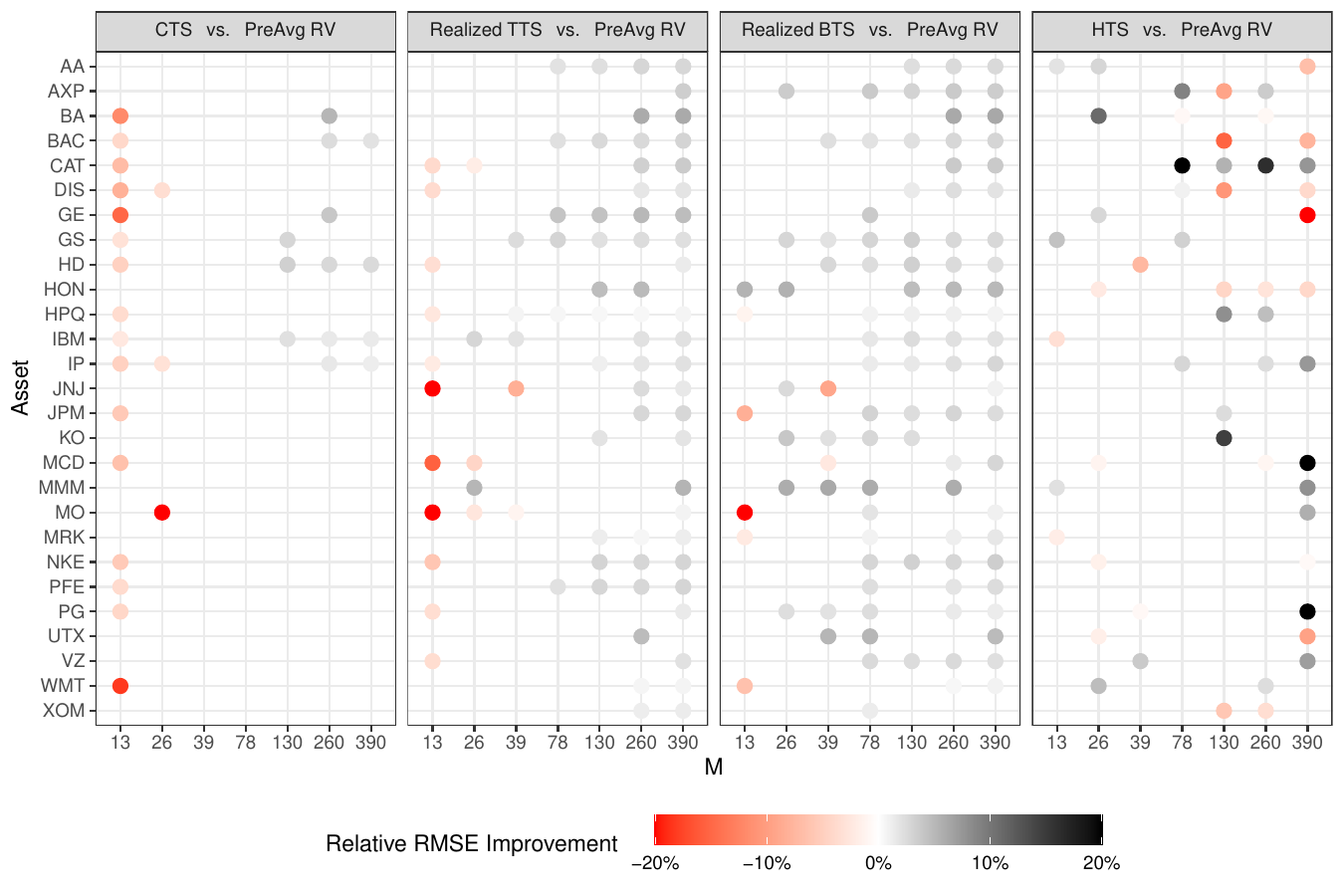}
	\caption{RMSE loss differences for the RV estimator based on different sampling schemes and a range of sampling frequencies $M$ for the 27 considered assets.
    Here, we use the (leaded) daily squared return as the proxy in the evaluation framework of \citet{Patton2011RV}.
	Each point corresponds to a (at the $5\%$ level) significant  loss difference of the corresponding RV estimator to a \emph{benchmark pre-averaging RV estimator} using all tick-level returns.
	Insignificant loss differences are omitted.
	The color scale of the points shows the relative improvement in terms of RMSE, where black (red) colors refer to an improvement (decline).}
	\label{fig:appl_RV_eval_vs_PAVG}
\end{figure}

In summary, our empirical analysis confirms our theoretical and simulation-based findings.
First, the more elaborate sampling schemes (rTTS, iBTS, rBTS, HTS) that take into account intraday variation clearly outperform CTS.
Second, rBTS and HTS perform best within this class\edit{, and can also outperform the noise robust pre-averaging estimator using all tick level data.}
Third, their relative effectiveness depends on the sampling frequency: HTS excels at (very) low frequencies, while rBTS proves to be more robust at higher ones.
The empirical superiority of the HTS and especially the rBTS schemes further underscores the practical value of the TTSV modeling framework, which enables the convenient derivation of the rBTS scheme.

\subsection{Comparing Forecast Performance}
\label{sec:appl_Forecasting}

We next assess how the gains in estimation accuracy of HTS and rBTS translate into improved forecast performance following the empirical analysis of \citet[Section 5.6]{LiuPattonSheppard2015}.
To this end, we use the Heterogeneous AutoRegressive (HAR) model of \citet{Corsi2009}, 
\begin{align}
	\label{eq:HARmodel}
	\RV_d(\btau) = \beta_0 + \beta_D \RV_{d-1}(\btau) + \beta_W \frac{1}{5} \sum_{j=1}^5  \RV_{d-j}(\btau) +  \beta_M \frac{1}{22} \sum_{j=1}^{22}  \RV_{d-j}(\btau) + \varepsilon_d, 
\end{align}
that models RV on day $d$ as a linear function of the past daily, weekly and monthly averages of RV with error term $\varepsilon_d$ and parameters  $(\beta_0, \beta_D, \beta_W, \beta_M)$ that are estimated by ordinary least squares.

For each combination of asset, sampling scheme, and sampling frequency, \edit{and for the tick-level pre-averaging RV estimator}, we use the HAR model in \eqref{eq:HARmodel} to generate one-step-ahead forecasts by estimating the parameters in \eqref{eq:HARmodel} with a rolling window consisting of 803 trading days for model estimation starting on January 1, 2012.
This results in an evaluation period of 1000 trading days ranging from March 28, 2015 to March 29, 2019.
We evaluate the resulting forecasts with the MSE and QLIKE loss functions.
\edit{As the associated estimation target, we use daily squared returns as in \cite{LiuPattonSheppard2015}, to have a fair evaluation target for all estimators.}

\begin{figure}[tb] 
	\centering
	\includegraphics[width=1\textwidth]{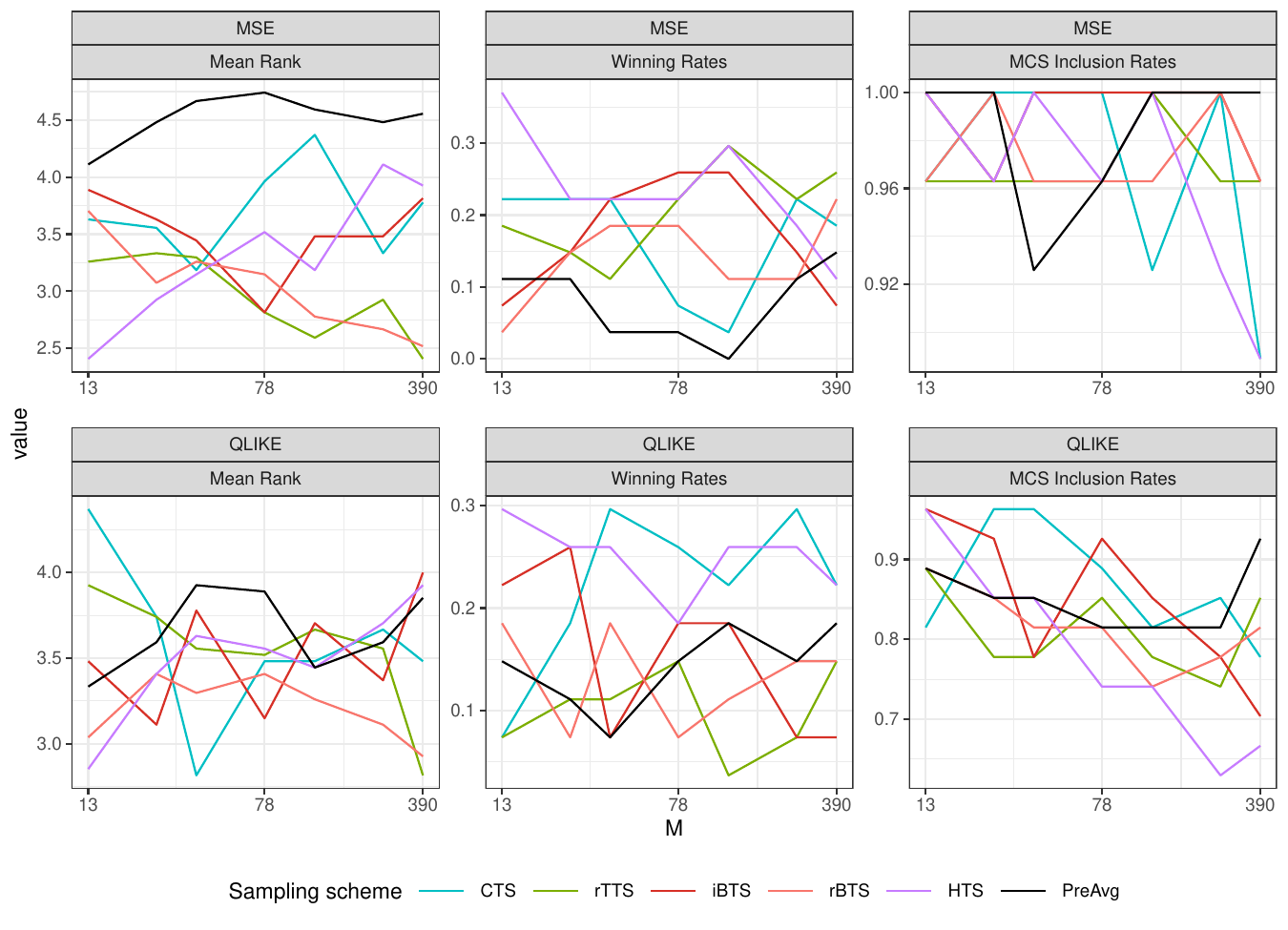}
	\caption{Average Ranks, winning rates, and MCS inclusion rates (at the $10\%$ level) of the MSE and QLIKE comparisons in the forecasting exercise plotted against the sampling frequency, individually for each considered sampling frequency (in color), and for the pre-averaging RV estimator.
    For the forecast evaluation, the daily squared return is used.}
	\label{fig:HARForecastResults_SQReturn}
\end{figure}

\edit{Figure~\ref{fig:HARForecastResults_SQReturn} reports results aggregated over time and across assets for each sampling scheme and frequency individually.
For both the MSE and QLIKE loss function, we report the average ranks of the respective sampling schemes, the proportion of comparisons where each sampling scheme is considered best, and the inclusion rates of the model confidence set (MCS) of \citet{hansen2011model} using the implementation of \citet{bernardi2018model}.}

\edit{We find that HTS performs best at very low sampling frequencies (below $M = 78$), achieving the lowest average ranks and highest winning rates.
The MCS inclusion rates are high across all sampling schemes and frequencies, which is unsurprising given the procedure’s low power, making the differences difficult to interpret.
For the higher frequencies between $M=78$ and $M=390$, no sampling scheme consistently outperforms the others, which can be explained by the substantial ``empirical noise'' that is added in such a forecasting exercise, compared to the estimation results from Section~\ref{sec:appl_Estimation}.
}

\section{Conclusions}
\label{sec:Conclusions}

In this paper we provide finite-sample theory as well as empirical results for the statistical quality of the classical RV estimator when the intraday returns are sampled in intrinsic time.
This approach accounts for intraday trading (transaction time sampling -- TTS), volatility patterns (business time sampling -- BTS), or absolute price changes (hitting time sampling -- HTS).
For BTS, we propose the novel \emph{realized} BTS variant that samples according to a combination of the observed transactions and the estimated tick variance.
The intrinsic time scales leverage the rich information content of high-frequency data by adopting a perspective that differs from traditional equidistant clock-time sampling, reflecting the irregular evolution of market activity and risk.

We find that, in the absence of market microstructure noise, the HTS scheme theoretically provides the most efficient RV estimates in finite samples.
However, the rBTS scheme emerges as most efficient in a restricted setting where sampling must occur independent of the observed intraday prices.
This restricted setting and consequently the rBTS scheme is motivated through the increased sensitivity of the HTS scheme to market microstructure noise, which we find empirically causes its performance to deteriorate rapidly when \edit{(intrinsic)} sampling frequencies exceed five minutes.
In contrast, the rBTS scheme is an attractive and robust alternative at all sampling frequencies.

The theoretical framework for our analysis builds on a joint model for the ticks (transaction or quote times) and prices, which we call the \emph{tick-time stochastic volatility} (TTSV) model:
The prices follow a continuous-time diffusion that is time-changed by a jump process that explicitly models the ticks.
As a result, prices form a pure jump process with time-varying and stochastic jump intensity capturing the empirical fact that price observations arrive randomly and at irregular intervals throughout the day.
Furthermore, the model includes a stochastic tick variance process---representing the variance of price jumps between adjacent ticks---that also varies over time and displays a mirrored intraday pattern relative to the trading intensity.

The TTSV model is particularly useful for theoretically disentangling the effects of intrinsic time sampling for several reasons.
First, it captures the natural spot variance decomposition into trading intensity and tick variance that is especially informative when comparing business and tick time sampling variants.
Second, it enables the derivation of theoretical finite-sample results in contrast to, for example, \citet{Barndorff2011SubsamplingRK} who provide asymptotic arguments in favor of the intensity version of BTS.
Third, by explicitly modeling the observed ticks through a jump process, the TTSV model naturally encompasses the novel realized BTS scheme, which performs well in our empirical application, demonstrating that its effectiveness reflects genuine practical improvements beyond the TTSV framework.

An interesting theoretical alternative is to accommodate the tick arrivals through \emph{discretization} instead of a \emph{time-change}, as recently proposed by \citet{Jacod2017, Jacod2019, DaXiu2021, Li2022remedi} among others.
\edit{While the TTSV framework enables convenient finite-sample derivations, we conjecture that the corresponding asymptotic analysis tends to be more complex and demands stronger assumptions compared to the discretization approaches in \citet{Jacod2017, Jacod2019, DaXiu2021, Li2022remedi}.}
Furthermore, advancing the theoretical analysis of noise-robust estimators such as subsampling, realized kernel, or pre-averaging RV, particularly in combination with rBTS and HTS sampling, offers promising paths for future research.

\section*{Replication Material}
Replication material is available  under \href{https://github.com/TimoDimi/replication_RVTTSV}{https://github.com/TimoDimi/replication\_RVTTSV}. \\ 
While the simulations can be fully replicated, we have to exclude the data files for the empirical application as these cannot be made publicly available.

\section*{Acknowledgements}

We would like to thank the editor, the associated editor and the two referees for very valuable and constructive comments that have substantially improved the results of the paper.
We are further thankful to Dobrislav Dobrev, Christian Gouri{\'e}roux, Andrew Patton, Davide Pirino, Winfried Pohlmeier, Angelo Ranaldo, Roberto Ren{\`o}, Richard Olsen, Philipp Sibbertsen, George Tauchen and the participants at the SoFiE Conference 2019, QFFE Conferences 2019 and 2022, and the Conference on Intrinsic Time in Finance 2022 for helpful comments.  
All remaining errors are ours.
We thank Sebastian Bayer and Christian Mücher for help in preparing the TAQ data.
T.~Dimitriadis gratefully acknowledges financial support from the German Research Foundation (DFG) through grant number 502572912.
R.~Halbleib gratefully acknowledges financial support from the DFG through the grant number 8672/1.

\addcontentsline{toc}{section}{References}

\singlespacing
\setlength{\bibsep}{5pt}

\bibliographystyle{apalike}
\bibliography{bib_RVTTS_v9}

\begin{thebibliography}{}

\bibitem[Admati and Pfleiderer, 1988]{Admati1988}
Admati, A.~R. and Pfleiderer, P. (1988).
\newblock A theory of intraday patterns: Volume and price variability.
\newblock {\em The Review of Financial Studies}, 1(1):3--40.

\bibitem[A{\"\i}t-Sahalia and Jacod, 2014]{ait2014high}
A{\"\i}t-Sahalia, Y. and Jacod, J. (2014).
\newblock {\em High-Frequency Financial Econometrics}.
\newblock Princeton University Press.

\bibitem[A{\"\i}t-Sahalia et~al., 2011]{aitsahalia2011}
A{\"\i}t-Sahalia, Y., Mykland, P., and Zhang, L. (2011).
\newblock Ultra high frequency volatility estimation with dependent
  microstructure noise.
\newblock {\em Journal of Econometrics}, 160(1):160--175.

\bibitem[Alfelt et~al., 2023]{alfelt2023singular}
Alfelt, G., Bodnar, T., Javed, F., and Tyrcha, J. (2023).
\newblock Singular conditional autoregressive wishart model for realized
  covariance matrices.
\newblock {\em Journal of Business \& Economic Statistics}, 41(3):833--845.

\bibitem[Andersen and Bollerslev, 1997]{andersen1997}
Andersen, T.~G. and Bollerslev, T. (1997).
\newblock Intraday periodicity and volatility persistence in financial markets.
\newblock {\em Journal of Empirical Finance}, 4(2-3):115--158.

\bibitem[Andersen and Bollerslev, 1998]{andersen1998a}
Andersen, T.~G. and Bollerslev, T. (1998).
\newblock Answering the skeptics: Yes, standard volatility models do provide
  accurate forecasts.
\newblock {\em International Economic Review}, 39(4):885--905.

\bibitem[Andersen et~al., 2001a]{andersen2001a}
Andersen, T.~G., Bollerslev, T., Diebold, F.~X., and Ebens, H. (2001a).
\newblock The distribution of realized stock return volatility.
\newblock {\em Journal of Financial Economics}, 61(1):43--76.

\bibitem[Andersen et~al., 2001b]{andersen2001b}
Andersen, T.~G., Bollerslev, T., Diebold, F.~X., and Labys, P. (2001b).
\newblock The distribution of realized exchange rate volatility.
\newblock {\em Journal of the American Statistical Association},
  96(453):42--55.

\bibitem[Andersen et~al., 2003]{andersen2003}
Andersen, T.~G., Bollerslev, T., Diebold, F.~X., and Labys, P. (2003).
\newblock Modeling and forecasting realized volatility.
\newblock {\em Econometrica}, 71(2):579--625.

\bibitem[Andersen et~al., 2007]{Andersenal2007}
Andersen, T.~G., Bollerslev, T., and Dobrev, D. (2007).
\newblock No-arbitrage semi-martingale restrictions for continuous-time
  volatility models subject to leverage effects, jumps and i.i.d. noise: Theory
  and testable distributional implications.
\newblock {\em Journal of Econometrics}, 138(1):125--180.

\bibitem[Andersen et~al., 2010]{Andersenal2010}
Andersen, T.~G., Bollerslev, T., Frederiksen, P., and Nielsen, M.~{\O}. (2010).
\newblock Continuous time models, realized volatilities, and testable
  distributional implications for daily stock returns.
\newblock {\em Journal of Applied Econometrics}, 25(2):233--261.

\bibitem[Andersen et~al., 2009]{andersen2006}
Andersen, T.~G., Davis, R.~A., Krei{\ss}, J.-P., and Mikosch, T.~V. (2009).
\newblock {\em Handbook of Financial Time Series}.
\newblock Springer Science \& Business Media.

\bibitem[Andersen et~al., 2012]{Andersen2012}
Andersen, T.~G., Dobrev, D., and Schaumburg, E. (2012).
\newblock Jump-robust volatility estimation using nearest neighbor truncation.
\newblock {\em Journal of Econometrics}, 169(1):75--93.

\bibitem[An{\'e} and Geman, 2000]{ane2000}
An{\'e}, T. and Geman, H. (2000).
\newblock Order flow, transaction clock, and normality of asset returns.
\newblock {\em Journal of Finance}, 55(5):2259--2284.

\bibitem[Bandi and Russell, 2008]{bandi2008}
Bandi, F.~M. and Russell, J.~R. (2008).
\newblock Microstructure noise, realized variance, and optimal sampling.
\newblock {\em Review of Economic Studies}, 75(2):339--369.

\bibitem[Barndorff-Nielsen et~al., 2008]{barndorff2008}
Barndorff-Nielsen, O.~E., Hansen, P.~R., Lunde, A., and Shephard, N. (2008).
\newblock Designing realized kernels to measure the ex post variation of equity
  prices in the presence of noise.
\newblock {\em Econometrica}, 76(6):1481--1536.

\bibitem[Barndorff-Nielsen et~al., 2009]{barndorff2009}
Barndorff-Nielsen, O.~E., Hansen, P.~R., Lunde, A., and Shephard, N. (2009).
\newblock Realized kernels in practice: Trades and quotes.
\newblock {\em The Econometrics Journal}, 12(3):C1--C32.

\bibitem[Barndorff-Nielsen et~al., 2011]{Barndorff2011SubsamplingRK}
Barndorff-Nielsen, O.~E., Hansen, P.~R., Lunde, A., and Shephard, N. (2011).
\newblock Subsampling realised kernels.
\newblock {\em Journal of Econometrics}, 160(1):204--219.

\bibitem[Barndorff-Nielsen and Shephard, 2002]{barndorff2002a}
Barndorff-Nielsen, O.~E. and Shephard, N. (2002).
\newblock Econometric analysis of realized volatility and its use in estimating
  stochastic volatility models.
\newblock {\em Journal of the Royal Statistical Society: Series B (Statistical
  Methodology)}, 64(2):253--280.

\bibitem[Bates, 2019]{bates2019crashes}
Bates, D.~S. (2019).
\newblock How crashes develop: intradaily volatility and crash evolution.
\newblock {\em The Journal of Finance}, 74(1):193--238.

\bibitem[Bauwens and Giot, 2001]{Bauwens2001}
Bauwens, L. and Giot, P. (2001).
\newblock {\em Econometric Modelling of Stock Market Intraday Activity,},
  volume~38 of {\em Advanced Studies in Theoretical and Applied Econometrics}.
\newblock Kluwer.

\bibitem[Bauwens and Hautsch, 2009]{bauwens2009}
Bauwens, L. and Hautsch, N. (2009).
\newblock Modelling financial high frequency data using point processes.
\newblock In {\em Handbook of financial time series}, pages 953--979. Springer.

\bibitem[Bernardi and Catania, 2018]{bernardi2018model}
Bernardi, M. and Catania, L. (2018).
\newblock The model confidence set package for {R}.
\newblock {\em International Journal of Computational Economics and
  Econometrics}, 8(2):144--158.

\bibitem[Bollerslev et~al., 2018]{bollerslev2018risk}
Bollerslev, T., Hood, B., Huss, J., and Pedersen, L.~H. (2018).
\newblock Risk everywhere: Modeling and managing volatility.
\newblock {\em The Review of Financial Studies}, 31(7):2729--2773.

\bibitem[Bollerslev et~al., 2020]{bollerslev2020good}
Bollerslev, T., Li, S.~Z., and Zhao, B. (2020).
\newblock Good volatility, bad volatility, and the cross section of stock
  returns.
\newblock {\em Journal of Financial and Quantitative Analysis}, 55(3):751--781.

\bibitem[Bollerslev et~al., 2022]{bollerslev2022zero}
Bollerslev, T., Medeiros, M.~C., Patton, A.~J., and Quaedvlieg, R. (2022).
\newblock From zero to hero: Realized partial (co) variances.
\newblock {\em Journal of Econometrics}, 231(2):348--360.

\bibitem[Br{é}maud, 1981]{bremaud1981}
Br{é}maud, P. (1981).
\newblock {\em Point Processes and Queues: Martingale Dynamics}.
\newblock Springer Series in Statistics.

\bibitem[Bucci, 2020]{bucci2020realized}
Bucci, A. (2020).
\newblock Realized volatility forecasting with neural networks.
\newblock {\em Journal of Financial Econometrics}, 18(3):502--531.

\bibitem[Carr and Wu, 2004]{carr2004}
Carr, P. and Wu, L. (2004).
\newblock Time-changed levy processes and option pricing.
\newblock {\em Journal of Financial Economics}, 71(1):113--114.

\bibitem[Clark, 1973]{clark1973}
Clark, P.~K. (1973).
\newblock A subordinated stochastic process model with finite variance for
  speculative prices.
\newblock {\em Econometrica}, pages 135--155.

\bibitem[Corsi, 2009]{Corsi2009}
Corsi, F. (2009).
\newblock {A simple approximate long-memory model of Realized Volatility}.
\newblock {\em Journal of Financial Econometrics}, 7(2):174--196.

\bibitem[Da and Xiu, 2021]{DaXiu2021}
Da, R. and Xiu, D. (2021).
\newblock When moving-average models meet high-frequency data: Uniform
  inference on volatility.
\newblock {\em Econometrica}, 89(6):2787--2825.

\bibitem[Dahlhaus and Neddermeyer, 2014]{Dahlhaus2014}
Dahlhaus, R. and Neddermeyer, J. (2014).
\newblock Online spot volatility-estimation and decomposition with nonlinear
  market microstructure noise models.
\newblock {\em Journal of Financial Econometrics}, 12(1):174--212.

\bibitem[Dahlhaus and Tunyavetchakit, 2016]{dahlhaus2016}
Dahlhaus, R. and Tunyavetchakit, S. (2016).
\newblock Volatility decomposition and estimation in time-changed price models.
\newblock {\em Preprint}.
\newblock
  \href{https://arxiv.org/abs/1605.02205}{https://arxiv.org/abs/1605.02205}.

\bibitem[Dassios and Zhao, 2013]{dassios2013exact}
Dassios, A. and Zhao, H. (2013).
\newblock Exact simulation of {H}awkes process with exponentially decaying
  intensity.
\newblock {\em Electronic Communications in Probability}, 18(62):1--13.

\bibitem[Delbaen and Schachermayer, 1994]{DelbaenSchachermayer1994}
Delbaen, F. and Schachermayer, W. (1994).
\newblock A general version of the fundamental theorem of asset pricing.
\newblock {\em Mathematische Annalen}, 300(3):463--520.

\bibitem[Diebold and Mariano, 1995]{DieboldMariano1995}
Diebold, F.~X. and Mariano, R.~S. (1995).
\newblock Comparing predictive accuracy.
\newblock {\em Journal of Business \& Economic Statististics}, 13(3):253--263.

\bibitem[Diggle and Marron, 1988]{DiggleMarron1988}
Diggle, P. and Marron, J.~S. (1988).
\newblock Equivalence of smoothing parameter selectors in density and intensity
  estimation.
\newblock {\em Journal of the American Statistical Association},
  83(403):793--800.

\bibitem[Dong and Tse, 2017]{dong2014}
Dong, Y. and Tse, Y.-K. (2017).
\newblock Business time sampling scheme with applications to testing
  semi-martingale hypothesis and estimating integrated volatility.
\newblock {\em Econometrics}, 5(51):1--19.

\bibitem[Engle and Russell, 2005]{engle2005}
Engle, R.~F. and Russell, J.~R. (2005).
\newblock A discrete-state continuous-time model of financial transaction
  prices and times.
\newblock {\em Journal of Business \& Economic Statistics}, 23(2):166--180.

\bibitem[Fukasawa, 2010]{Fukasawa2010RV}
Fukasawa, M. (2010).
\newblock Realized volatility with stochastic sampling.
\newblock {\em Stochastic Processes and their Applications}, 120(6):829--852.

\bibitem[Fukasawa and Rosenbaum, 2012]{FukasawaRosenbaum2012}
Fukasawa, M. and Rosenbaum, M. (2012).
\newblock Central limit theorems for realized volatility under hitting times of
  an irregular grid.
\newblock {\em Stochastic Processes and their Applications},
  122(12):3901--3920.

\bibitem[Gabaix et~al., 2003]{gabaix2003}
Gabaix, X., Gopikrishnan, P., Plerou, V., and Stanley, H.~E. (2003).
\newblock A theory of power-law distributions in financial market fluctuations.
\newblock {\em Nature}, 423(6937):267--270.

\bibitem[Griffin and Oomen, 2008]{griffin2008}
Griffin, J.~E. and Oomen, R. C.~A. (2008).
\newblock Sampling returns for realized variance calculations: Tick time or
  transaction time?
\newblock {\em Econometric Reviews}, 27(1-3):230--253.

\bibitem[Hamilton and Jorda, 2002]{hamilton2002}
Hamilton, J.~D. and Jorda, O. (2002).
\newblock A model of the federal funds rate target.
\newblock {\em Journal of Political Economy}, 110(5):1135--1167.

\bibitem[Hansen and Lunde, 2006]{hansen2006}
Hansen, P.~R. and Lunde, A. (2006).
\newblock Realized variance and market microstructure noise.
\newblock {\em Journal of Business \& Economic Statistics}, 24(2):127--161.

\bibitem[Hansen et~al., 2011]{hansen2011model}
Hansen, P.~R., Lunde, A., and Nason, J.~M. (2011).
\newblock The model confidence set.
\newblock {\em Econometrica}, 79(2):453--497.

\bibitem[Harris, 1986]{Harris1986}
Harris, L. (1986).
\newblock A transaction data study of weekly and intradaily patterns in stock
  returns.
\newblock {\em Journal of Financial Economics}, 16(1):99--117.

\bibitem[Hawkes, 1971]{hawkes1971}
Hawkes, A.~G. (1971).
\newblock Spectra of some self-exciting and mutually exciting point processes.
\newblock {\em Biometrika}, 58(1):83--90.

\bibitem[Hawkes, 2018]{hawkes2018hawkes}
Hawkes, A.~G. (2018).
\newblock Hawkes processes and their applications to finance: a review.
\newblock {\em Quantitative Finance}, 18(2):193--198.

\bibitem[Hoga and Dimitriadis, 2023]{HogaDimi2022}
Hoga, Y. and Dimitriadis, T. (2023).
\newblock On testing equal conditional predictive ability under measurement
  error.
\newblock {\em Journal of Business \& Economic Statistics}, 41(2):364--376.

\bibitem[Hussain et~al., 2023]{Hussain2023}
Hussain, S.~M., Ahmad, N., and Ahmed, S. (2023).
\newblock Applications of high-frequency data in finance: a bibliometric
  literature review.
\newblock {\em International Review of Financial Analysis}, 102790.

\bibitem[Jacod, 2018]{Jacod2018}
Jacod, J. (2018).
\newblock Limit of random measures associated with the increments of a brownian
  semimartingale.
\newblock {\em Journal of Financial Econometrics}, 16(4):526--569.

\bibitem[Jacod et~al., 2009]{Jacod2009}
Jacod, J., Li, Y., Mykland, P.~A., Podolskij, M., and Vetter, M. (2009).
\newblock Microstructure noise in the continuous case: the pre-averaging
  approach.
\newblock {\em Stochastic Processes and their Applications}, 119(7):2249--2276.

\bibitem[Jacod et~al., 2017]{Jacod2017}
Jacod, J., Li, Y., and Zheng, X. (2017).
\newblock Statistical properties of microstructure noise.
\newblock {\em Econometrica}, 85(4):1133--1174.

\bibitem[Jacod et~al., 2019]{Jacod2019}
Jacod, J., Li, Y., and Zheng, X. (2019).
\newblock Estimating the integrated volatility with tick observations.
\newblock {\em Journal of Econometrics}, 208(1):80--100.

\bibitem[Jacod and Shiryaev, 2003]{ref:Jacod2003:Limit}
Jacod, J. and Shiryaev, A. (2003).
\newblock {\em Limit theorems for stochastic processes}, volume 288.
\newblock Springer Science \& Business Media.

\bibitem[Jones et~al., 1994]{jones1994}
Jones, C.~M., Kaul, G., and Lipson, M.~L. (1994).
\newblock Transactions, volume, and volatility.
\newblock {\em Review of Financial Studies}, 7(4):631--651.

\bibitem[Kalnina and Linton, 2008]{Kalnina2008}
Kalnina, I. and Linton, O. (2008).
\newblock Estimating quadratic variation consistently in the presence of
  endogenous and diurnal measurement error.
\newblock {\em Journal of Econometrics}, 147(1):47--59.

\bibitem[Laub et~al., 2021]{laub2021elements}
Laub, P.~J., Lee, Y., and Taimre, T. (2021).
\newblock {\em The elements of Hawkes processes}.
\newblock Springer.

\bibitem[Li and Linton, 2022]{Li2022remedi}
Li, Z.~M. and Linton, O. (2022).
\newblock A remedi for microstructure noise.
\newblock {\em Econometrica}, 90(1):367--389.

\bibitem[Liesenfeld et~al., 2006]{Liesenfeld2006}
Liesenfeld, R., Nolte, I., and Pohlmeier, W. (2006).
\newblock Modelling financial transaction price movements: A dynamic integer
  count model.
\newblock {\em Empirical Economics}, 30(4):795--825.

\bibitem[Liptser and Shiryayev, 2012]{Liptser2012}
Liptser, R. and Shiryayev, A.~N. (2012).
\newblock {\em Theory of Martingales}, volume~49 of {\em Mathematics and its
  Applications}.
\newblock Springer Science \& Business Media.

\bibitem[Liu et~al., 2015]{LiuPattonSheppard2015}
Liu, L.~Y., Patton, A.~J., and Sheppard, K. (2015).
\newblock {Does anything beat 5-minute RV? A comparison of realized measures
  across multiple asset classes}.
\newblock {\em Journal of Econometrics}, 187(1):293--311.

\bibitem[Meddahi, 2002]{meddahi2002}
Meddahi, N. (2002).
\newblock A theoretical comparison between integrated and realized
  volatilities.
\newblock {\em Journal of Applied Econometrics}, 17(5):479--508.

\bibitem[Monroe, 1978]{Monroe1987}
Monroe, I. (1978).
\newblock Processes that can be embedded in brownian motion.
\newblock {\em Annals of Probability}, 6(1):42--56.

\bibitem[Oomen, 2005]{oomen2005}
Oomen, R. C.~A. (2005).
\newblock Properties of bias-corrected realized variance under alternative
  sampling schemes.
\newblock {\em Journal of Financial Econometrics}, 3(4):555--577.

\bibitem[Oomen, 2006]{oomen2006}
Oomen, R. C.~A. (2006).
\newblock Properties of realized variance under alternative sampling schemes.
\newblock {\em Journal of Business \& Economic Statistics}, 24(2):219--237.

\bibitem[Patton, 2011a]{Patton2011RV}
Patton, A.~J. (2011a).
\newblock Data-based ranking of realised volatility estimators.
\newblock {\em Journal of Econometrics}, 161(2):284--303.

\bibitem[Patton, 2011b]{Patton2011}
Patton, A.~J. (2011b).
\newblock Volatility forecast comparison using imperfect volatility proxies.
\newblock {\em Journal of Econometrics}, 160(1):246--256.

\bibitem[Patton and Zhang, 2023]{Patton2023Bespoke}
Patton, A.~J. and Zhang, H. (2023).
\newblock Bespoke realized volatility: Tailored measures of risk for volatility
  prediction.
\newblock {\em Preprint}.
\newblock
  \href{https://dx.doi.org/10.2139/ssrn.4315106}{https://dx.doi.org/10.2139/ssrn.4315106}.

\bibitem[Plerou et~al., 2001]{plerou2001}
Plerou, V., Gopikrishnan, P., Gabaix, X., Nunes~Amaral, L., and Stanley, H.~E.
  (2001).
\newblock Price fluctuations, market activity and trading volume.
\newblock {\em Quantitative Finance}, 1(2):262--269.

\bibitem[Podolskij and Vetter, 2009]{PodolskijVetter2009}
Podolskij, M. and Vetter, M. (2009).
\newblock {Estimation of volatility functionals in the simultaneous presence of
  microstructure noise and jumps}.
\newblock {\em Bernoulli}, 15(3):634--658.

\bibitem[Politis and Romano, 1994]{PolitisRomano1994}
Politis, D.~N. and Romano, J.~P. (1994).
\newblock The stationary bootstrap.
\newblock {\em Journal of the American Statistical Association},
  89(428):1303--1313.

\bibitem[Press, 1967]{press1967}
Press, S.~J. (1967).
\newblock A compound events model for security prices.
\newblock {\em Journal of Business}, 40(2):317--335.

\bibitem[Reisenhofer et~al., 2022]{reisenhofer2022harnet}
Reisenhofer, R., Bayer, X., and Hautsch, N. (2022).
\newblock {HARNet: A convolutional neural network for realized volatility
  forecasting}.
\newblock {\em Preprint}.
\newblock
  \href{https://arxiv.org/abs/2205.07719}{https://arxiv.org/abs/2205.07719}.

\bibitem[Robert and Rosenbaum, 2012]{robert2012volatility}
Robert, C.~Y. and Rosenbaum, M. (2012).
\newblock Volatility and covariation estimation when microstructure noise and
  trading times are endogenous.
\newblock {\em Mathematical Finance: An International Journal of Mathematics,
  Statistics and Financial Economics}, 22(1):133--164.

\bibitem[Shephard and Yang, 2017]{ShephardYang2017}
Shephard, N. and Yang, J.~J. (2017).
\newblock Continuous time analysis of fleeting discrete price moves.
\newblock {\em Journal of the American Statistical Association},
  112(519):1090--1106.

\bibitem[Vetter and Zwingmann, 2017]{Vetter2017}
Vetter, M. and Zwingmann, T. (2017).
\newblock {A note on central limit theorems for quadratic variation in case of
  endogenous observation times}.
\newblock {\em Electronic Journal of Statistics}, 11(1):963 -- 980.

\bibitem[Wood et~al., 1985]{Wood1985}
Wood, R.~A., McInish, T.~H., and Ord, J.~K. (1985).
\newblock An investigation of transactions data for {NYSE} stocks.
\newblock {\em The Journal of Finance}, 40(3):723--739.

\bibitem[Zhang et~al., 2005]{zhang2005}
Zhang, L., Mykland, P.~A., and A{\"\i}t-Sahalia, Y. (2005).
\newblock A tale of two time scales: Determining integrated volatility with
  noisy high-frequency data.
\newblock {\em Journal of the American Statistical Association},
  100(472):1394--1411.

\end{thebibliography}

\vspace{1cm}

\appendix
\huge
\noindent \textbf{Appendix}
\normalsize 
\onehalfspacing

\section{\edit{Main Proofs}}
\label{sec:ProofsMain}

\edit{This appendix contains the proofs of the main results from Sections~\ref{sec:TTSV_model} and \ref{sec:EfficiencyFiniteSample}. The supporting lemmas, along with their proofs and the remaining proofs for the paper, are provided in Appendix~\ref{app:proofs} of the Supplementary Material.}

\begin{proof}[\bf \hypertarget{proof:unbiasedness}{Proof of Theorem \ref{thm:unbiasedness}}]
	By Proposition \ref{prop:martingale_price}, the price process $P$ is martingale. In the proof of Proposition \ref{prop:martingale_price} we show that $P$ is square-integrable, which implies that the process $Q = P^2-[P]$ is a martingale as well. For any pair of stopping times from the sampling scheme $\tau_{j-1}$ and $\tau_j$ we can apply the Optional Stopping Theorem because the stopping times are surely bounded by $T$ and we have that
	\begin{align}
		\E [r^2(\tau_{j-1}, \tau_j) | \F_{\tau_{j-1}}] & = \E[P_{\tau_{j}}^2 - P_{\tau_{j-1}}^2 | \F_{\tau_{j-1}}] = \E[[P]_{\tau_{j}} - [P]_{\tau_{j-1}} | \F_{\tau_{j-1}} ].
	\end{align}
	We then obtain that the realized variance estimator is unbiased for $\E[[P]_T]$ as
	\begin{align}
		\E [\RV(\boldsymbol{\tau})] &= \E \left[ \sum_{j=1}^Mr^2(\tau_{j-1}, \tau_{j}) \right]  =\E\left[\sum_{j=1}^M \E[ r^2(\tau_{j-1}, \tau_{j}) | \F_{\tau_{j-1}} ] \right]\\
		&=\E\left[\sum_{j=1}^M  \E[ [P]_{\tau_j} - [P]_{\tau_{j-1}} | \F_{\tau_{j-1}} ] \right] =\E[[P]_T],
	\end{align}
	where we use that $M$ is a.s. finite together with Lemma \ref{lem:random_sum_cond_exp} applied to the nonnegative squared returns and increments of the quadratic variation.
	
	To show that the realized variance estimator is unbiased for expected realized $\IV$, we use the assumption about the conditional distribution of the price increments in Assumption \ref{ass:filtration}. Denote the (stochastic) jump times in the interval by $t_1,t_2,...$ with $0\leq t_1<t_2<\dots\leq T$. By Lemma \ref{lem:random_sum_cond_exp} and the non-negativity of the squared-price increments we have that
	\begin{align}
		\E[[P]_T] & = \E\left[\sum_{0\le t_i \le T} (\Delta P_{t_i})^2 \right]  = \E\left[ \sum_{0\leq t_i\leq T} \E \left[\varsigma^2(t_i)U_i^2 \bigg|\mathcal{F}_{t_{i}-} \right]\right] \\ 
		& = \E \left[ \sum_{0\leq t_i\leq T}\varsigma^2(t_i)\right]  = \E\left[ \rIV(0,T) \right].
	\end{align}
	It remains only to show $\E [\rIV(0,T)] = \E[\IV(0,T)]$. This equality follows from the non-negativity and $\Fil$-predictability of $\varsigma$ and the characterization of the compensator (see \cite{ref:Jacod2003:Limit}[Theorem 3.17])).
\end{proof}

\begin{proof}[\bf \hypertarget{proof:MSE_IV}{Proof of Theorem \ref{thm:MSE_IV}}]
	We begin by writing the difference between the estimator and the estimation target at any time $t\in[0,T]$ as the sum of three martingales at time $t$:
	\begin{align}
		RV(\boldsymbol{\tau}, t) - \IV(0,t) & = A(t) + B(t) + C(t),
	\end{align}
	where $RV(\boldsymbol{\tau}, t) = \sum_{j=1}^M r^2(\tau_{j-1}\wedge t, \tau_j \wedge t)$, $A := RV(\boldsymbol{\tau}, \cdot) - [P]$, $B:=[P] - \rIV(0,\cdot)$ and $C = \rIV(0,\cdot) -  \IV(0,\cdot)$. That $A$ is a martingale follows almost directly from the result that $P^2-[P]$ is a martingale, which we show in the proof of Theorem \ref{thm:unbiasedness}. Showing that $B$ is a martingale follows by similar arguments as in the proof of Proposition \ref{prop:martingale_price}, showing that $P$ is a martingale. For a reference that $C$ is a martingale, see the proof of Theorem \ref{thm:unbiasedness}.\\
	
	We can write the first martingale $A$ in a more convenient form by noting that for any pair of stopping times $0\le \sigma \le \tau\le T$ we have
	\begin{align}
		(P_{\tau}-P_{\sigma})^2 - ([P]_{\tau} - [P]_{\sigma}) &  = 2 \sum_{\sigma < t_i \le \tau} (P_{t_i-} - P_{\sigma}) \Delta P_{t_i},
	\end{align}
	where the $t_i$ denote the (stochastic) jump times in the stochastic interval $(\sigma,\tau]$ and $P_{t-} := \lim_{s \uparrow t} P_s$. Hence we have that 
	\begin{equation}
		A(t) = \sum_{j=1}^M \sum_{\tau_{j-1}\wedge t < t_i \le \tau_j \wedge t} A_i(\tau_{j-1}),
	\end{equation}
	where $A_i(\tau_{j-1})=2 (P_{t_i-} -P_{\tau_{j-1}}) \varsigma(t_i)U_i$. Here we use the notation from Assumption \ref{ass:filtration} such that $\Delta P_{t_i} = \varsigma(t_i)U_i$ where $U_i$ is a random variable with a standard normal distribution conditional on the $\sigma$-algebra $\mathcal{F}_{t_i-}$. Similarly we have for $B$ and $C$ that
	\begin{align}
		B(t) & = \sum_{0\le t_i \le t} B_i\\
		C(t) & = \sum_{0\le t_i \le t} C_i,
	\end{align}
	where $B_i= \varsigma^2(t_i)U_i^2 - \varsigma^2(t_i)$ and $C_i=\varsigma^2(t_i) - \int_{t_{i-1}}^{t_i} \varsigma^2(r)\lambda(r)dr$. We have that 
	\begin{equation}
		\E \left[ \left(\text{RV}(\boldsymbol{\tau}) - \IV(0,T) \right)^2 \right] = \E[[\text{RV}(\boldsymbol{\tau}, \cdot) - \IV(0,\cdot) ]_T] = \E[[A+B+C]_T] \label{eq:expected_quadratic_variation_MSE}
	\end{equation}
	as long as $A+B+C$ is square integrable, which holds if $\E[[A]_T]$, $\E[[B]_T]$ and $\E[[C]_T]$ are finite. \\
	
	To compute the expected quadratic variation in \eqref{eq:expected_quadratic_variation_MSE} and show the appropriate bounds, we note the following properties of the increments of the processes $A$, $B$ and $C$, which follow immediately from Assumption \ref{ass:filtration} for any jump time $t_i$ such that $\tau_{j-1}<t_i\le \tau_j$:
	\begin{itemize}
		\item $\E[A_i(\tau_{j-1})|\mathcal{F}_{t_i-}] = 0 $;
		\item $\E[B_i|\mathcal{F}_{t_i-}] = 0$;
		\item $C_i$ is $\mathcal{F}_{t_i-}$-measurable and we have $\E[C_i|\mathcal{F}_{t_{i-1}}] = 0$, since $C_i = \int_{t_{i-1}}^{t_i}\varsigma^2(r) d\tilde{N}(r)$, where $\tilde{N}(t) = N(t) - \int_0^t \lambda(r) dr$;
		\item $\E[A_i(\tau_{j-1}) B_i | \mathcal{F}_{t_i-}] = 0$, since $\E[U_i(U_i^2 -1)|\mathcal{F}_{t_i-}]=0$;
		\item $\E[A_i(\tau_{j-1})C_i|\mathcal{F}_{t_i-}]=E[B_iC_i|\mathcal{F}_{t_i-}]=0$, since $C_i$ is $\mathcal{F}_{t_i-}$-measurable;
		\item $\E[B_i^2|\mathcal{F}_{t_i-}] = 2\varsigma^4(t_i)$, since $\E[(U_i^2-1)^2|\mathcal{F}_{t_i-}]=2$;
		\item $\E[C_i^2|\mathcal{F}_{t_{i-1}}] = \E \left[ \int_{t_{i-1}}^{t_i} \varsigma^4(r) dN(r)|\mathcal{F}_{t_{i-1}} \right] = \E \left[ \varsigma^4(t_i) |\mathcal{F}_{t_{i-1}} \right]$ by the Ito isometry of the stochastic integral.
	\end{itemize}
	We first derive the MSE result by computing the expected quadratic variation in \eqref{eq:expected_quadratic_variation_MSE} and then show the appropriate moment bounds. We apply Lemma \ref{lem:random_sum_cond_exp} multiple times to various sums and $\sigma$-algebras such that we can use the properties above and that $ \E \left[ \sum_{\tau_{j-1}<t_i\le\tau_j} A_i^2(\tau_{j-1}) \Big| \mathcal{F}_{\tau_{j-1}} \right] = \frac{2}{3} \E\left[ r^4(\tau_{j-1}, \tau_j) \big| \mathcal{F}_{\tau_{j-1}} \right]- 2\E\left[\IQ(\tau_{j-1},\tau_j) \big| \mathcal{F}_{\tau_{j-1}} \right]$ by Lemma \ref{lem:square_incr_RV_QV} and find the MSE result:
	\begin{align}
		&\E[(\text{RV}(\tau)-\IV(0,T))^2] = \E[[A+B+C]_T]\\
		&\quad=\E\left[\sum_{j=1}^M \sum_{\tau_{j-1}<t_i\le\tau_j} \{A_i^2(\tau_{j-1}) + 2 A_i(\tau_{j-1})B_i + 2 A_i(\tau_{j-1})C_i \} + \sum_{0\le t_i \le T}\{B_i^2 + C_i^2 + 2 B_i C_i\}  \right]\\
		& \quad = \E\left[\sum_{j=1}^M \E \left[ \sum_{\tau_{j-1}<t_i\le\tau_j} A_i^2(\tau_{j-1}) \Bigg| \mathcal{F}_{\tau_{j-1}} \right] \right] + \E\left[ \sum_{0\le t_i \le T} \E [ B_i^2 | \mathcal{F}_{t_i-}]  \right] + \E\left[ \sum_{0\le t_i \le T} \E[C_i^2 | \mathcal{F}_{t_{i-1}} ]  \label{eq:MSE_computation}\right]\\
		& \quad = \frac{2}{3} \E\left[\sum_{j=1}^M r^4(\tau_{j-1}, \tau_j) \right] +  \E\left[ \IQ(0,T) \right].
	\end{align}
	It remains to show that $\E[[A]_T]$, $\E[[B]_T]$ and $\E[[C]_T]$ are finite such that the equality in \eqref{eq:expected_quadratic_variation_MSE} holds and that the we can apply Lemma \ref{lem:random_sum_cond_exp} in equation \eqref{eq:MSE_computation}. By applying Lemma \ref{lem:random_sum_cond_exp} to the positive increments of the quadratic variations and by using the Burkholder-Davis-Gundy inequalities we have that
	\begin{align}
		\E[[A]_T] & = \frac{2}{3} \E\left[ \sum_{j=1}^M \E \left[ r^4(\tau_{j-1}, \tau_j) \big| \mathcal{F}_{\tau_{j-1}}\right] \right] - 2\E[\IQ(0,T)] \\ & \le C\E \left[ \sum_{j=1}^M \E \left[ \left([P]_{\tau_{j}}-[P]_{\tau_{j-1}}\right)^2 \bigg| \mathcal{F}_{\tau_{j-1}} \right] \right]\le C \E\left[ [P]_T^2 \right]
	\end{align}
	for some constant $C>0$, $\E[[B]_T] = 2\E[\IQ(0,T)]$ and $\E[[C]_T] = \E[\IQ(0,T)]$. By Assumption \ref{ass:filtration} $\E\left[ [P]_T^2 \right]<\infty$ and we have
	\begin{align}
		3\E\left[\IQ(0,T) \right] = 3\E\left[ \int_0^T\varsigma^4(r)dN(r) \right] = \E\left[\sum_{t_i\le T} (\Delta P_{t_i})^4\right] \le \E[[P]_T^2].
	\end{align}
\end{proof}

\begin{proof}[\bf \hypertarget{proof:EfficientSampling}{Proof of Theorem \ref{thm:EfficientSampling}}]
	The results follow from applying Lemma \ref{lem:min_squared_sum} to the MSE results in Theorem \ref{thm:MSE_IV} and Corollaries \ref{cor:MSE_jump_based} and \ref{cor:MSE_intensity_based}. Note that in the latter two cases the additional assumption on the independence between Brownian motion and the other processes implies that the remainder term $\E[R(\boldsymbol{\tau})]$ in the MSE results in \eqref{eq:MSE_realized_sampling} and \eqref{eq:MSE_intensity} is zero. Similarly, by additionally assuming that $N$ is a doubly stochastic Poisson process, the second remainder term $\E[\tilde{R}(\boldsymbol{\tau})]$ in \eqref{eq:MSE_intensity} is zero, since in that case $\E\left[ \int_{\tau_{j-1}}^{\tau_j}\varsigma^2(r)dN(r)\big| \mathcal{F}^{\lambda,\varsigma}_{\tau_j}\right] = \int_{\tau_{j-1}}^{\tau_j}\varsigma^2(r) \lambda(r) dr$.
\end{proof}

\pagebreak 
\huge
\noindent \textbf{Supplementary Material}
\normalsize 
\onehalfspacing

\bigskip 
\noindent 
This supplemental material contains a comparison of the TTSV model to discretized diffusions in Appendix \ref{sec:ComparisonDiscretization}, additional finite sample theory in a setting where sampling can use information from the end of the trading day in Appendix~\ref{sec:EfficientSamplingUptoT}, and a specific comparison to the results of \citet{oomen2006} in Appendix~\ref{sec:ComparisonOomen2006}.
\edit{All proofs, except those of the main results, are collected in Appendix~\ref{app:proofs}, while Appendix~\ref{sec:RemainderTerms} provides arguments regarding the remainder terms of Corollary~\ref{cor:MSE_jump_based}.
Finally, Appendix~\ref{sec:AdditionalResults} presents additional empirical results.}

\setcounter{figure}{0}
\setcounter{table}{0}
\setcounter{thm}{0}
\renewcommand\thefigure{\Alph{section}.\arabic{figure}}
\renewcommand\thetable{\Alph{section}.\arabic{table}}
\renewcommand\thethm{\Alph{section}.\arabic{thm}}

\section{A Comparison to Discretized Diffusions}
\label{sec:ComparisonDiscretization}

In this section, we compare the TTSV model to the ``discretized'' diffusion framework of \citet{Jacod2017, Jacod2019, DaXiu2021, Li2022remedi} as an alternative modeling choice that exhibits random observation times.

The proposed model for the underlying log-price process is a possibly discontinuous  It\^{o} semimartingale that can (under standard regularity assumptions for $b$, $\sigma$ and $\delta$) be written as\footnote{See \citet[Equation (2.2)]{Jacod2019} and the following assumptions for more details.} 
$$Q(t) = Q(0) + \int_0^tb(r)dr + \int_0^t \sigma(r) dB(r) + \int_{[0,t] \times E} \delta(r,z)p(dr,dz).$$
The crucial components that facilitate comparability to the TTSV model are the possibly random observation times of the log-price process.
Following \citet[p.3]{Jacod2019}, observations of the underlying log-price take place based on the (possibly irregularly spaced and random) observation times $0 = T(n,0) < T(n,1) < \dots$ for a triangular sequence $T(n,i)$ of finite times, where the ``stage $n$'' diverges in the asymptotic setting.
Further define 
$$N^n(t): = \sum_{i\geq1} \mathds{1}_{\{T(n,i)\leq t\}}, \qquad \text{ and } \qquad \Delta(n,i)=T(n,i)-T(n,i-1),$$  
such that $N^n(t)+1$ denotes the number of observations up to time $t$ and $\Delta(n,i)$ is the time between observation number $i-1$ and $i$.

Given the assumption that for all $i$, the $\Delta(n,i)$ are in an appropriate sense of the same order of magnitude as the deterministic and positive sequence $\Delta_n$ that converges to zero as $n$ diverges, the observations times $T(n,i)$ are such that for all $t$,
\begin{align}
	\label{eqn:ConvergenceModulatingProcess}
	\Delta_nN^n(t) \overset{\mathbb{P}}{\longrightarrow} \int_0^t\alpha(r)dr,
\end{align}
where $\alpha(t)$ is an appropriately regular and strictly positive It\^{o} semimartingale that, in a statistical sense, modulates the difference of the observation scheme from a regular equally-spaced (calendar time) grid. 
These conditions allow for flexible observation times such as equidistant sampling, (modulated) Poisson sampling schemes and time-changed regular sampling schemes \citep{Jacod2019}.

The log-prices $Q(t)$ can further be contaminated with (different specifications of) MMN as $\widetilde{Q}(T(n,i) ) = Q\big( T(n,i) \big) + \epsilon^n(i)$ for some noise term $\epsilon^n(i)$, resembling our specification in \eqref{eq:SimulateNoise}.
Therefore, similar to the TTSV model, the observed price is constant between observation points that are potentially irregularly spaced and random.

In comparison, the discretized diffusions and the TTSV model share the properties of having observed price paths that are constant between the random observation points with the technical difference that this is achieved by a time-change with a jump process in the TTSV model and by random observation times in the discretized diffusions.
This implies the conceptual difference that in the TTSV model, realized transactions drive price changes and in the discretized diffusion framework, transaction times are simply the observation times of the prices.

An important difference of the models arises in the interpretation of the observation times $T(n,i)$ in the discretized diffusions, where sparse sampling could be included as follows:
First, as  in \citet{Jacod2019}, the $T(n,i)$ can be interpreted as the observed transaction times.
To consider sparsely sampled returns (as is our main focus of interest), we would however require another layer of random times that represent the sampling schemes.

Second, one could directly consider the random times $T(n,i)$ as the sampling points. 
The current standard assumption on the sampling points (see e.g., Assumption (O)(ii) in \citet{Jacod2017}, Assumption O 2.(c) in \citet{Li2022remedi} and Assumption A on page 302 in the book \cite{ait2014high}) however imposes that the duration $\Delta(n,i)$ is conditionally independent from the entire filtration conditional on the information up to observation time $T(n,i-1)$.
This assumption rules out the consideration of the sampling schemes such as the \emph{realized} TTS and BTS variants and HTS, which we advocate in this paper.
Therefore, while the discretized diffusion literature imposes very weak assumptions on the price process and the noise distribution, relaxing the modeling assumption on the   observation times to account for ``realized'' sampling schemes would require additional work.
We mention that there is also literature such as \cite{Fukasawa2010RV} and \cite{robert2012volatility} that allow for more general dependence for the duration $\Delta(n,i)$ within the discretized diffusions framework, though this is under specific assumptions on the observation times (hitting times at the trading grid) and the noise process.

Hence, while both these modeling possibilities do not immediately show how the research question of finding optimal (sparse) sampling points could be analyzed within the setting of discretized diffusions, a derivation of similar results might in principle be feasible.
Furthermore, the discretization schemes might be promising alternatives for future research to e.g., robustify our findings to different (possibly weaker) modeling assumptions, \edit{or extensions to asymptotic results}.

We continue to examine in more detail how the discretization framework described above could produce similar results to ours reported in the main paper for the TTSV model.
For this, we consider the diffusion (that is later on discretized) 
\begin{align}
	\label{eqn:DiffusionComparisonTTSV}
 	Q(t) = Q(0) + \int_0^t \varsigma(r) \sqrt{\lambda(r)} dB(r),
\end{align}	
for some strictly positive It\^{o} processes $\varsigma(r)$ and $\lambda(r)$ that are also used for the corresponding specification of the TTSV model in \eqref{eq:TTSV_model}.
These models are related as both have a spot variance process of $\varsigma^2(r) \lambda(r)$.\footnote{A notable difference between the discretized diffusion in \eqref{eqn:DiffusionComparisonTTSV} and the TTSV model is that in the latter, the jump variance between two trading times at jump time $t_i$ is $\varsigma^2(t_i)$, whereas for the former, the price jump has a variance of $\int_{t_{i-1}}^{t_i}\varsigma^2(r)\lambda(r)dr$.}
Furthermore, if we discretize the diffusion in \eqref{eqn:DiffusionComparisonTTSV} with Poisson random times that follow a modulating process with $\alpha(t) = \lambda(t) \Delta_n$ in the sense of \eqref{eqn:ConvergenceModulatingProcess}, the count process of the discretization $N^n(t)$ resembles the jump process $N(t)$ of the TTSV model (for $n$ large enough in the sense of the asymptotic approximation in \eqref{eqn:ConvergenceModulatingProcess}).

If $P(\cdot)$ denotes the log-price of the TTSV model, under Assumptions \eqref{ass:filtration}--\eqref{ass:Moments}, we also get that the \emph{ex ante} (conditional on $\mathcal{F}_s$) conditional variance of the prices in the interval $[s,t]$ is the same for both processes as
\begin{align*}
	\E \left[ \big(P(t)-P(s) \big)^2 \;\middle|\;  \mathcal{F}_s \right]
	&= \E \left[ \int_s^t \varsigma^2(r) dN(r) \; \middle| \; \mathcal{F}_s \right]  
	= \E \left[ \int_s^t \varsigma^2(r) \lambda(r) dr \; \middle| \; \mathcal{F}_s \right] \\  
	&= \E \left[ \big(Q(t) - Q(s) \big)^2 \; \middle| \; \mathcal{F}_s \right].
\end{align*}
However, when considering the \emph{ex post} variance over the interval $[s,t]$ (i.e., conditioning on $\FintN_t$, thus implying knowledge of the intensities and the transaction/observation times) we get that
\begin{align}
	\label{eqn:TTSVVariance}
	\E \left[ \big( P(t) - P(s) \big)^2 \; \middle| \; \FintN_t \right] = \E \left[ \int_s^t \varsigma^2(r) dN(r) \; \middle| \; \FintN_t \right] 
	=\rIV(s,t)
\end{align}
under the TTSV model.
In the discretized diffusion, when defining the last observation time prior to time $t$ by $\tau(t):=\max\{s\leq t: \exists i\in \N : s = T(n,i)\}$, we however get that
\begin{align}
	\label{eqn:DiscretizedTTSVlikeVariance}
	\E \left[ \big( Q(t) - Q(s) \big)^2 \; \middle| \; \mathcal{F}^{\lambda,\varsigma,N^n}_t \right] 
	= \E \left[ \int_{\tau(s)}^{\tau(t)} \varsigma^2(r) \lambda(r) dr \; \middle| \; \mathcal{F}^{\lambda,\varsigma,N^n}_t \right] 
	= \IV \big( \tau(s),\tau(t) \big).
\end{align}
In this calculation, conditioning on $N^n(\cdot)$ corresponds to knowledge of the observation times, similar as conditioning on $N(\cdot)$ in the TTSV model.

While the right-hand side of \eqref{eqn:DiscretizedTTSVlikeVariance} equals the IV between the last observations before $s$ and $t$ respectively,  we obtain the \emph{realized} IV between $s$ and $t$ for the TTSV model under \eqref{eqn:TTSVVariance}. 
Hence, the comparison of \eqref{eqn:TTSVVariance} and \eqref{eqn:DiscretizedTTSVlikeVariance} illustrates that when employing jump-based sampling/observation schemes and by conditioning on $\FintN_t$, the \emph{realized} IV only arises under the TTSV model.
Consequently, with the choice of a discretized diffusion described in \eqref{eqn:DiffusionComparisonTTSV} and below, we would be unable to theoretically derive the \emph{realized} BTS scheme.
Notice that the \emph{realized} BTS scheme appears to be superior to the classical  \emph{intensity} BTS scheme in both, the estimation and forecasting setting of our empirical application in Section \ref{sec:Application} as can be seen in Table~\ref{tab:applSigPosNegValues} and Figure~\ref{fig:HARForecastResults_SQReturn}.
Since these results are obtained in the model-free empirical application, this illustrates that the TTSV model allows to develop theory for a new, efficient sampling scheme, which is practically relevant as it performs well in the empirical application.

\pagebreak 
\section{Efficient Sampling Using Information of the Entire Day}
\label{sec:EfficientSamplingUptoT}

In this section, we derive the conditional bias and MSE of the RV estimator based on general sampling schemes $\btau$ that are allowed to incorporate information up to the \emph{end of the trading day}, which are hence not necessarily stopping times. 
The use of the information of the entire day allows to fix the number of sampled returns of a sampling scheme to a deterministic number $M$ and corresponds to the empirical practice of computing RV at the end of the trading day, often with a fixed frequency (amount of samples) $M$.
In this way we can explicitly control the noise picked up be the RV estimator, when applying it to observed price data.

Since the sampling times considered here are no longer stopping times with respect to the filtration $\Fil$, we deviate from the setup in Section \ref{theory} and develop new theory in this section. 
We consider results pertaining to the bias and the MSE of RV, \emph{conditional} on the following information sets that are defined for all $t \in [0,T]$,
\begin{align*}
	\Fint_t &= 
	\sigma \big( \lambda(s), \varsigma(s); \quad 0 \le s \le t \big) \subset \mathcal{F}_t, \qquad \text{ and } \\
	\FintN_t &= 
	\sigma \big( \lambda(s), \varsigma(s), N(s); \quad 0 \le s \le t \big) \subset \mathcal{F}_t.
\end{align*} 
In a similar spirit, we distinguish between sampling schemes $\btau$ that are $\Fint_T$- and $\FintN_T$-measurable, where the latter ``realized'' or ``jump-based'' case allows for a dependence of the sampling times on the realized tick pattern of the particular day. Here, a sampling scheme is understood to be $\mathcal{G}$-measurable for some information set $\mathcal{G}$, if all the sampling times in $\btau$ are $\mathcal{G}$-measurable.
Opposed to the results of Theorem~\ref{thm:EfficientSampling} (a), the theory in this section cannot deal with sampling schemes that are allowed to use \emph{price} information in $\mathcal{F}_T$.

In order to get conditional results for the sampling schemes that use information of the entire day, we impose the following, additional assumptions:
\begin{ass}
	\label{ass:Independence}
	The process $\{B(n)\}_{n\geq0}$ is independent from $\{N(t)\}_{t\geq0}$ and from $\{\varsigma(t)\}_{t\geq0}$.
\end{ass}

\begin{ass}
	\label{ass:Moments}
	The expectations $\E \big[\int_t^T \varsigma^2(r) \lambda(r) dr \mid \mathcal{F}_t \big]$, $\E\left[\varsigma^4\left(t\right)\right]$ and $\E \big[\int_0^t\varsigma^4\left(r\right)\lambda\left(r\right)dr\big]$ exist and are finite for all $t \in [0,T]$.
\end{ass}

\begin{ass}
	\label{ass:intensity_restrictions}
	\begin{enumerate}[label=(\alph*)]
		\item The counting process $\{N(t)\}_{t\geq0}$ is a doubly stochastic Poisson process, adapted to $\mathcal{F}_t$, which has a positive, $\mathcal{F}_t$-measurable and continuous intensity $\{\lambda(t)\}_{t\geq0}$ such that $\int_0^t \lambda(s) ds < \infty$ a.s.\ for all $t\geq 0$;  see \citet[Chapter II.1]{bremaud1981} for details;
		\item The processes $\{N(t)\}_{t\geq0}$ and $\{\varsigma(t)\}_{t\geq0}$ are independent.
	\end{enumerate}
	\smallskip 
\end{ass}

\begin{thm} 
	\label{prop:bias}
	Let the sampling scheme $\btau$ be $ \FintN_T$-measurable.		
	\begin{enumerate}[label=(\alph*)]
		\item 
		\label{item:Bias_general}
		Under Assumptions \eqref{ass:filtration}--\eqref{ass:Moments}, it holds that 
		$\E \left[ \RV(\btau) \; \middle| \; \FintN_T \right]  = \rIV(0,T)$.
		
		\item 
		\label{item:Bias_Poisson}
		Under Assumptions \eqref{ass:filtration}--\eqref{ass:intensity_restrictions}, it holds that 
		$\E \left[ \RV(\btau)\; \middle| \; \Fint_T \right] = \IV(0,T)$.
	\end{enumerate}
\end{thm}

Part (b) of this theorem shows that for any $\FintN_T$-measurable sampling scheme, RV is an $\Fint_T$-conditionally unbiased estimator for IV under the TTSV model based on a doubly stochastic Poisson process $N(t)$ as specified in Assumption \eqref{ass:intensity_restrictions}.
When conditioning on $\FintN_T$ however, part (a) shows that for the general TTSV model, RV is conditionally unbiased for the realized IV, which can be interpreted as an $N(t)$-dependent refinement of IV.

While similar to Theorem~\ref{thm:unbiasedness}, there is no theoretical distinction between different sampling schemes $\btau$ in terms of a bias of the RV estimator (when either staying in setting (a) or (b) of Theorem \ref{prop:bias}), we continue by showing that similar to Theorem~\ref{thm:EfficientSampling}, the choice of $\btau$ entails a difference in the estimation efficiency.
For this, we derive a closed-form expression for the MSE of the RV estimator depending on the sampling grid $\btau$ with a finite amount of $M$ sampling points.

\begin{thm} $ $ \vspace{-0.2cm}
	\label{prop:MSE_tick} 
	\begin{enumerate}[label=(\alph*)]
		\item 
		Under Assumptions \eqref{ass:filtration}--\eqref{ass:Moments} and given that the sampling times $\btau$ are $\FintN_T$-measurable,  \\
		$\E \left[ \big( \RV({\btau}) - \IV(0,T) \big)^2 \;\middle|\; \FintN_T \right]  
		=  \big( \rIV(0,T) - \IV(0,T) \big)^2 + 2 \sum_{j=1}^{M} \rIV(\tau_{j-1},{\tau_j})^2$.
		
		\item 
		Under Assumptions  \eqref{ass:filtration}--\eqref{ass:intensity_restrictions} and given that the sampling times $\btau$ are $\Fint_T$-measurable, \\
		$\E\left[ \big( \RV({\btau}) - \IV(0,T) \big)^2 \;\middle|\; \Fint_T\right] 
		= 3 \IQ(0,T) + 2 \sum_{j=1}^{M} \IV(\tau_{j-1},{\tau_j})^2$,
		where \\
		$\IQ(s,t):=\int_{s}^{t}\varsigma^4(r)\lambda(r)dr$ denotes the Integrated Quarticity of the TTSV model.
	\end{enumerate}
\end{thm}

Part (a) of Theorem~\ref{prop:MSE_tick} provides the MSE result for $\FintN_T$-measurable sampling times
for general jump processes \emph{without} imposing the Poisson Assumption \eqref{ass:intensity_restrictions} such that it e.g., also applies to Hawkes processes.
In contrast, the Poisson restriction is required for part (b) as the proof relies on the zero-mean martingale property of the compensated jump process conditional on $\Fint_T$, which is only satisfied under Assumption \eqref{ass:intensity_restrictions}.

In both parts of Theorem \ref{prop:MSE_tick}, the MSE is bounded from below by the constant factors $\big( \rIV(0,T) - \IV(0,T) \big)^2$ and $3 \IQ(0,T)$, respectively. 
Most important for our purposes are however the terms $2\sum_{j=1}^{M(T)} \IV (\tau_{j-1},{\tau_j})^2$ and $2\sum_{j=1}^{M} \rIV (\tau_{j-1},{\tau_j})^2$, which depend on the sum of the squared intraday (realized) IVs according to the chosen sampling grid $\btau$.
Hence, the results of Theorem~\ref{prop:MSE_tick} align with Corollary~\ref{cor:MSE_jump_based} and Corollary~\ref{cor:MSE_intensity_based} and show that the sampling points should be chosen to homogenize the realized and classical IV, respectively.

As in Section \ref{sec:EfficiencyFiniteSample}, we see that by applying the Cauchy-Schwartz inequality, these terms are minimized by sampling times that are chosen such that the intraday returns become as homogeneous as possible in terms of their (realized) IV.
It is important to notice that Theorem \ref{prop:MSE_tick} is valid for any finite (and in practice user-chosen) value of sampling points $M$.
This allows the subsequent analysis of the finite sample efficiency of different sampling schemes through the terms  $2\sum_{j=1}^{M} \IV (\tau_{j-1},{\tau_j})^2$ and $2\sum_{j=1}^{M} \rIV (\tau_{j-1},{\tau_j})^2$, respectively.\footnote{\label{foot:rIVEfficiency}While choosing realized IV as the estimation target for part (a) would eliminate the first term $\big( \rIV(0,T) - \IV(0,T) \big)^2$, it would have leave the more important quantity $2\sum_{j=1}^{M} \rIV (\tau_{j-1},{\tau_j})^2$ unchanged, hence not affecting the relative finite sample efficiencies of different sampling schemes; see Appendix \ref{sec:ComparisonOomen2006} and in particular Table \ref{tab:MSEresults} for details.}

We continue to investigate the MSE for the specific (theoretical) sampling schemes introduced above.
The two MSE expressions in Theorem \ref{prop:MSE_tick} can be further simplified under the iBTS  and rBTS schemes as 
\begin{align}
	\label{eqn:HomogenousReturns}
	\sum_{j=1}^{M} \IV(\tau^{\mbox{\tiny iBTS}}_{j-1},{\tau^{\mbox{\tiny iBTS}}_j})^2 = \frac{\IV(0,T)^2}{M}
	\quad \text{ and } \quad
	\sum_{j=1}^{M} \rIV(\tau^{\mbox{\tiny rBTS}}_{j-1},{\tau^{\mbox{\tiny rBTS}}_j})^2 =  \frac{\rIV(0,T)^2}{M}.
\end{align}
This implies that the iBTS and rBTS schemes respectively make the distribution of the sampled intraday returns as homogeneous as possible, which we formalize in the following Corollary that follows directly from Theorem \ref{prop:MSE_tick}, equation \eqref{eqn:HomogenousReturns} and the Cauchy-Schwartz inequality.

\begin{cor}
	\label{cor:EfficientSamplingAppendix}
	$ $ \vspace{-0.2cm}
	\begin{enumerate}[label=(\alph*)]
		\item 
		Under Assumptions \eqref{ass:filtration}--\eqref{ass:Moments} and given that the sampling times $\btau$ are $\FintN_T$-measurable, \\
		$\E \left[ \big( \RV({\btau}) - \IV(0,T) \big)^2 \;\middle|\; \FintN_T \right]  \ge  \E \left[ \big( \RV({\btau^{\mbox{\tiny rBTS}}}) - \IV(0,T) \big)^2 \;\middle|\; \FintN_T \right]$, with equality if and only if $\btau \equiv \btau^{\mbox{\tiny rBTS}}$.

		\item 
		Under Assumptions  \eqref{ass:filtration}--\eqref{ass:intensity_restrictions} and given that the sampling times $\btau$ are $\Fint_T$-measurable, \\
		$\E \left[ \big( \RV({\btau}) - \IV(0,T) \big)^2 \;\middle|\; \Fint_T \right]  \ge  \E \left[ \big( \RV({\btau^{\mbox{\tiny iBTS}}}) - \IV(0,T) \big)^2 \;\middle|\; \Fint_T \right]$, with equality if and only if  $\btau \equiv \btau^{\mbox{\tiny iBTS}}$.
	\end{enumerate}
\end{cor}

This implies that for a fixed value of $M$, the rBTS scheme provides the smallest MSE among all possible $\FintN_T$-measurable sampling schemes.
Equivalently, if we only consider $\Fint_T$-measurable sampling, the iBTS scheme achieves the lowest MSE.
The proof techniques used in this section unfortunately do not allow for the consideration of the most general class of $\mathcal{F}_T$-measurable sampling, such that an ``end of the day variant'' of HTS cannot be considered here.

\section[A Comparison with the Results of Oomen (2006)]{A Comparison with the Results of \citet{oomen2006}}
\label{sec:ComparisonOomen2006}

In this section, we thoroughly relate the theory results of Appendix~\ref{sec:EfficientSamplingUptoT} to the results of \citet{oomen2006}, who uses a simplified version of the TTSV model with a constant tick variance process $\varsigma(t) = \varsigma_c$ and a non-homogeneous Poisson process $N(t)$.
He derives MSE expressions in his equations (9)--(10), which are in the spirit of our Theorem \ref{prop:MSE_tick} and Corollary \ref{cor:EfficientSamplingAppendix}.\footnote{
	The past literature on sampling schemes often uses inconsistent terminologies, which requires special care when comparing the results among different papers.
	E.g., \cite{oomen2006} refers to BTS as sampling with respect to the ``expected number of transactions'' and to TTS as sampling with respect to the ``realized number of transactions'', which matches our definitions of iTTS and rTTS.
	Furthermore, \cite{griffin2008} differentiate between the tick and transaction time sampling, where the former samples with respect to transactions with non-zero price changes.
}
This section illustrates how our results nest the ones of \citet{oomen2006} and additionally clarifies the specific settings under which the MSE results in \citet[Equations (9)--(10)]{oomen2006} can be derived.
For this, we impose Assumptions \eqref{ass:filtration}--\eqref{ass:intensity_restrictions} throughout this section.

In order to conduct a formal comparison with our results, we have to distinguish four settings with respect to the information set that is used for the sampling grids and the conditioning in the MSE (either $\Fint_T$  or $\FintN_T$), and with respect to the estimation target (either $\IV$ or $\rIV$), that we give in Table \ref{tab:combinations}. 
While settings (i) and (ii) allow for the comparison of $\FintN_T$-measurable sampling schemes, we should only compare $\Fint_T$-measurable sampling schemes in settings (iii) and (iv).
It is crucial to note that MSE comparisons between sampling schemes are only meaningful when carried out under the same setting.

\begin{table}[tbh]
	\centering
	\begin{tabular}{lcc}
		\toprule 
		Information Set $\backslash$ Target \qquad & $\rIV=\Stil$ & $\IV=\Sno$  \\
		\midrule
		$\FintN_T$  & (i) & (ii)\\
		$\Fint_T$  & (iii) & (iv)\\
		\bottomrule
	\end{tabular}
	\caption{Overview of the four considered settings in deriving MSE results.}
	\label{tab:combinations}
\end{table}

Table \ref{tab:MSEresults} reports the conditional MSE results, together with the efficient sampling schemes and their respective MSE for the four settings (i)--(iv). 
The upper Panel A gives results for the TTSV model (restricted to a doubly stochastic Poisson process $N(t)$), where the lower Panel B presents simplifications to the case $\varsigma(t) = \varsigma_c$, hence allowing for a direct comparison with the results of \citet{oomen2006}.
The MSE results under settings (ii) and (iv) are stated in our Theorem \ref{prop:MSE_tick}.
For the settings (i) and (iii), the results can be easily obtained from the proof of Theorem \ref{prop:MSE_tick}; in particular see the quadratic expansions in equations \eqref{eq:MSEdecomp} and \eqref{eq:MSEdecomp2}.
Further notice that the ranking of the sampling schemes is the same for settings (i) and (ii) and for settings (iii) and (iv), respectively, as the conditional MSEs only differ by a term that is invariant from the sampling scheme.

\begin{table}[tbh]
	\scriptsize 
	\centering
	\begin{tabular}{llcl}
		\toprule 
		\textbf{Setting} & \textbf{Conditional~MSE} & \textbf{Eff.~Sampl.} & \textbf{Cond.~MSE~of~Eff.~Sampl.}  \\
		\midrule 
		\\
		\multicolumn{4}{l}{\textbf{Panel A: TTSV model}} \\
		\midrule
		(i) & $2 \sum_{j=1}^{M} \rIV(\tau_{j-1}, \tau_j)^2$  & rBTS & $2 \rIV^2 / M$ \\
		(ii) & $2 \sum_{j=1}^{M} \rIV(\tau_{j-1}, \tau_j)^2 + (\rIV - \IV)^2$ & rBTS & $2 \rIV^2 / M + (\rIV - \IV)^2$ \\
		(iii) &  $2 \sum_{j=1}^{M} \IV(\tau_{j-1}, \tau_j)^2 + 2 \IQ$ & iBTS & $2 \IV^2/M + 2 \IQ$ \\
		(iv) & $2 \sum_{j=1}^{M} \IV(\tau_{j-1}, \tau_j)^2 + 3 \IQ $ & iBTS &  $2 \IV^2/M + 3 \IQ$ \\
		\midrule 
		\\
		\multicolumn{4}{l}{\textbf{Panel B: Model of \citet{oomen2006} with constant tick variance $\varsigma(t) = \varsigma_c$}} \\
		\midrule
		(i) & $2\varsigma_c^4 \sum_{j=1}^{M} \big( N(\tau_j) - N(\tau_{j-1}) \big)^2$ & rTTS = rBTS & $2 \varsigma_c^4 N(T)^2/M$ \\
		(ii) & $2\varsigma_c^4 \sum_{j=1}^{M}  \big( N(\tau_j) - N(\tau_{j-1}) \big)^2 + \varsigma_c^4 (N(T)-\Lambda(T))^2$ & rTTS = rBTS & $2 \varsigma_c^4 N(T)^2/M + \varsigma_c^4 (N(T)-\Lambda(T))^2$ \\
		(iii) & $2\varsigma_c^4\sum_{j=1}^{M} \big(\Lambda(\tau_j) - \Lambda(\tau_{j-1})\big)^2 + 2 \varsigma_c^4 \Lambda(T)$ & iTTS = iBTS & $2 \varsigma_c^4\Lambda(T)^2/M +2\varsigma_c^4\Lambda(T)$ \\
		(iv) & $2\varsigma_c^4 \sum_{j=1}^{M} \big(\Lambda(\tau_j) - \Lambda(\tau_{j-1})\big)^2 + 3 \varsigma_c^4 \Lambda(T)$ & iTTS = iBTS & $2 \varsigma_c^4 \Lambda(T)^2/M + 3 \varsigma_c^4 \Lambda(T)$ \\
		\bottomrule 
	\end{tabular}
	\caption{MSE results and efficient sampling schemes under the settings (i)--(iv) described in Table \ref{tab:combinations} for the general TTSV model in Panel A and for the simplified version of \citet{oomen2006} in Panel B. The table is expressed in terms of our notation, where we use the shorthands $\IV := \IV(0,T)$, $\rIV := \rIV(0,T)$, $\IQ := \IQ(0,T)$ and $\Lambda(t) := \int_0^t \lambda(s) \mathrm{d}s$ for $t \in [0, T]$.
	The efficient sampling schemes in settings (iii) and (iv) are taken among the $\Fint_T$-measurable sampling schemes (that are in particular not based on the realizations of the process $N(t)$).}
	\label{tab:MSEresults}
\end{table}

The results of Panel B of Table \ref{tab:MSEresults} are obtained as under the simplifications of \citet{oomen2006}, we get  $\rIV(\tau_{j-1}, \tau_j) = \varsigma_c^2 \cdot \big( N(\tau_j) - N(\tau_{j-1}) \big)$, $\IV(\tau_{j-1}, \tau_j) = \varsigma_c^2 \cdot \big( \Lambda(\tau_j) - \Lambda(\tau_{j-1})\big)$, and $\IQ(0,T) = \varsigma_c^4 \Lambda(T)$, where $\Lambda(t) = \int_0^t \lambda(s) \mathrm{d}s$ for $t \in [0, T]$.
The MSE result of \citet[Equation (9)]{oomen2006} for iTTS (denoted BTS in his paper) corresponds to the result derived in our setting (iv), whereas the MSE result for rTTS (denoted TTS in his notation) in his equation (10) corresponds to setting (i), hence rendering these conditional MSEs not directly comparable.
(Notice here that the notation $\Sigma$ in \citet{oomen2006} is unfortunately used for both, IV in his equation (9) and rIV in his equation (10).)
However, the conclusion that rTTS is more efficient than iTTS in his setting still holds true, but should formally be concluded from the MSE calculations under setting (ii) as \citet{oomen2006} allows for $\FintN_T$-measurable, jump-based sampling schemes and considers IV as the estimation target.

\pagebreak
\section{Proofs} 
\label{app:proofs}
We structure the proofs as follows:
Subsection~\ref{subsec:ProofMainPaper} contains the proofs for the results in the Sections~\ref{sec:TTSV_model} and \ref{sec:EfficiencyFiniteSample} \edit{apart from the proofs for the main results, which are contained in Appendix \ref{sec:ProofsMain}}. We give proofs for our results on sampling efficiency using information of the entire day in Appendix~\ref{sec:EfficientSamplingUptoT} in Subsection~\ref{subsec:ProofsEfficientSamplingUptoT}.

\subsection{\edit{Remaining} Proofs for the Results in the Sections~\ref{sec:TTSV_model} and \ref{sec:EfficiencyFiniteSample}}
\label{subsec:ProofMainPaper}

\begin{proof}[\bf \hypertarget{proof:martingale_price}{Proof of Proposition \ref{prop:martingale_price}}] 
	By Assumption \ref{ass:filtration}, the jump process has finite activity, such that $N_t-N_s<\infty$ a.s. for any $0\le s\le t\le T$. For each $n\in \mathbb{N}$ we define the stopping time $\rho_n:= \sup\{t\in[0,T]: N(t) < n\}$, which equals the $n$-th jump time or the final time $T$, if the process jumps has fewer than $n$ jumps. In particular note that $\mathbb{P}(\rho_n\to T)=1$, because $N$ is of finite activity. The stopped process $P^{\rho_n}=\{P_{\rho_n\wedge t}\}_{t\in [0,T]}$ is a martingale for each $n\in \mathbb{N}$, since we can condition on the $\sigma$-algebras just before the jump times $\mathcal{F}_{t_{i-}}$ and use Lemma \ref{lem:random_sum_cond_exp} such that
	\begin{align}
		\E\left[P^{\rho_n}_t-P^{\rho_n}_s|\mathcal{F}_s\right]&=\E \left[ \sum_{i=N_s+1\wedge n}^{N_t\wedge n}\varsigma(t_i)U_i \Bigg|\mathcal{F}_s \right]
		=\E \left[ \sum_{i=N_s+1\wedge n}^{N_t\wedge n}\varsigma(t_i)\E[U_i|\mathcal{F}_{t_i-}] \Bigg|\mathcal{F}_s \right]
		=0\label{eq:martingale},
	\end{align}
	which implies that $P$ is a local martingale. In particular in \eqref{eq:martingale}, we use the $\mathcal{F}_{t_{i-}}$-measurability of the tick volatility and the conditional distribution of $U_i$
	\begin{equation}
		\E[U_i|\mathcal{F}_{t_i-}]=\E[B(N(t_i))-B(N(t_{i-1}))|\mathcal{F}_{t_i-}]=0.
	\end{equation}
	The bound required for Lemma~\ref{lem:random_sum_cond_exp} holds, because
	\begin{align}
		\E[|P^{\rho_n}_t-P^{\rho_n}_s|]\le \sqrt{n} \sqrt{\E[[P]_T]}
	\end{align}
	by the Cauchy-Schwarz inequality and we assume $\E[[P]_T^2]<\infty$ in Assumption \ref{ass:filtration} which implies $\E[[P]_T]<\infty$ by Jensen's inequality. Since $\E[[P]_T]<\infty$, $P$ is square-integrable and this implies that $P$ is a true martingale.
\end{proof}

\phantomsection
\hypertarget{lem:random_sum_cond_exp}{}
\begin{lem} 
	\label{lem:random_sum_cond_exp}
	Consider a sequence of integrable random variables $A_1, A_2,...$, a sequence of $\sigma$-algebras $\mathcal{G}_1\subseteq \mathcal{G}_2,... \in \mathbb{F}$ and an almost surely finite integer-valued random variable $M$. Assume for each $j\in \mathbb{N}$ that $\{j \le M\} \in \mathcal{G}_j$.\footnote{Interpretation: the sequence of $\sigma$-algebras forms a discrete-time filtration $(\mathcal{G}_j)_{j=1,...}$ and $M+1$ is required to be a stopping time with respect to that filtration, i.e. $\{M=j-1\}\in \mathcal{G}_{j}$.} If  $A_j\ge0$ for each $j\in \mathbb{N}$ or there exist random variables $\bar{A}$ and $\tilde{A}$ such that $ \big| \sum_{j=1}^M A_j \big| \le \bar{A}$, $\E[\bar{A}]<\infty$, $\big| \sum_{j=1}^M \E[A_j|\mathcal{G}_j] \big| \le \tilde{A}$ and $\E[\tilde{A}]<\infty$, then we have for any $\sigma$-algebra $\mathcal{G} \subseteq \mathcal{G}_1$ that
	\begin{equation}
		\E\left[ \sum_{j=1}^M A_j \Bigg| \mathcal{G} \right] = \E\left[ \sum_{j=1}^M \E[A_j | \mathcal{G}_j] \Bigg| \mathcal{G} \right].
	\end{equation}
\end{lem}

\begin{proof} [\bf \hypertarget{proof:random_sum_cond_exp}{Proof of Lemma \ref{lem:random_sum_cond_exp}}]
	Because $M$ is almost surely finite, we have the almost sure convergence $\lim_{n\to \infty} M \wedge n = M$. The result now follows, as    
	\begin{align}
		\E\left[ \sum_{j=1}^M A_j  \Bigg| \mathcal{G} \right] & = \lim_{n\to \infty} \E\left[ \sum_{j=1}^{M\wedge n} A_j  \Bigg| \mathcal{G} \right]  = \lim_{n\to \infty} \E\left[ \sum_{j=1}^{n} \mathbf{1}_{\{j \le M \} } A_j  \Bigg| \mathcal{G}  \right]  = \lim_{n\to \infty} \sum_{j=1}^{n}  \E\left[ \mathbf{1}_{\{j \le M \} } A_j  \big| \mathcal{G}  \right] \\
		& = \lim_{n\to \infty} \sum_{j=1}^{n}  \E \left[\E \left[ \mathbf{1}_{\{j \le M \} } A_j |\mathcal{G}_j\right]  \big| \mathcal{G} \right]  = \lim_{n\to \infty} \sum_{j=1}^{n}  \E\left[ \mathbf{1}_{\{j \le M \} }\E\left[  A_j |\mathcal{G}_j\right]  \big| \mathcal{G} \right] \\
		&  = \lim_{n\to \infty} \E\left[ \sum_{j=1}^{M\wedge n}  \E\left[  A_j |\mathcal{G}_j\right]  \Bigg| \mathcal{G} \right]  = \E\left[ \sum_{j=1}^{M}  \E\left[  A_j |\mathcal{G}_j\right] \Bigg| \mathcal{G}  \right],
	\end{align}
	where the first and last equality follow from the Monotone Convergence Theorem in the case that $A_j\ge 0$ for each $j\in\mathbb{N}$ and from the Dominated Convergence Theorem for conditional expectations under the assumption of the existence of integrable bounding random variables.
\end{proof}

\begin{proof}[\bf \hypertarget{proof:vola_decomposition}{Proof of Proposition \ref{prop:vola_decomposition}}]
	\begin{align}
		\lim_{\delta \downarrow 0} \frac{1}{\delta} \E\left[ (P_{t+\delta} - P_t)^2 \big| \mathcal{F}_t \right] & = \lim_{\delta \downarrow 0} \frac{1}{\delta} \E\left[ [P]_{t+\delta} - [P]_t \big| \mathcal{F}_t \right] \\
		& = \lim_{\delta \downarrow 0} \E\left[ \frac{\IV(t, t+\delta)}{\delta} \Bigg| \mathcal{F}_t \right] \\
		& =  \E\left[ \lim_{\delta \downarrow 0} \frac{\IV(t, t+\delta)}{\delta} \Bigg| \mathcal{F}_t \right] \\
		& =  \varsigma^2(t+) \lambda(t+),
	\end{align}
	where we use the Dominated Convergence Theorem in the third step to exchange the limit and the expectation with the bound coming from the integrable random variable $Z(t)$ and we apply the Fundamental Theorem of Calculus in the last step to the right-sided derivative of $\IV(0,\cdot)$ at $t$ and use the right-continuity of the filtration.
\end{proof}

\begin{proof}[\bf \hypertarget{proof:IVvola}{Proof of Proposition \ref{prop:IVvola}}]
	This result is a special case of Theorem \ref{thm:unbiasedness} that appears in Section \ref{sec:EfficiencyFiniteSample} by choosing the trivial sampling scheme $\btau=\{0,T\}$.
\end{proof}

\begin{lem} \label{lem:square_incr_RV_QV}
	Under Assumption \ref{ass:filtration} for any pair of stopping times $0\le\sigma \le \tau \le T$
	\begin{align}
		\E \left[ ( (P_{\tau}-P_{\sigma})^2 - ([P]_{\tau} - [P]_{\sigma}) )^2| \mathcal{F}_{\sigma} \right] &= 
		\E \left[ \sum_{\sigma<t_i\le\tau} A_i^2(\sigma) \Bigg| \mathcal{F}_{\sigma} \right] \\
		&= \frac{2}{3} \E\left[ r^4(\sigma, \tau)| \mathcal{F}_{\sigma} \right] - 2\E \left[\IQ(\sigma, \tau)| \mathcal{F}_{\sigma} \right]\\
		&=2 \E\left[\rIV(\sigma,\tau)(2(P_\tau - P_\sigma)^2 - \rIV(\sigma,\tau)) | \mathcal{F}_{\sigma} \right] \\
		& \quad - 2 \E \left[ \IQ(\sigma, \tau)| \mathcal{F}_{\sigma} \right],
	\end{align}
	where $A_i(\sigma)=2 (P_{t_i-}-P_{\sigma})\varsigma(t_i)U_i$.
\end{lem}

\begin{proof}[\bf \hypertarget{proof:square_incr_RV_QV}{Proof of Lemma \ref{lem:square_incr_RV_QV}}]
	As noted in the proof of Theorem \ref{thm:MSE_IV}, we can write 
	\begin{equation}
		(P_{\tau}-P_{\sigma})^2 - ([P]_{\tau} - [P]_{\sigma}) = 2 \int_\sigma^\tau  (P_{r-} - P_{\sigma}) dP_r,
	\end{equation}
	as a stochastic integral with respect to the price process $P$. Using the It\^{o} isometry for the stochastic integral and the Optional Stopping Theorem it follows that
	
	\begin{align}
		\E \left[((P_{\tau}-P_{\sigma})^2 - ([P]_{\tau} - [P]_{\sigma}))^2 | \F_{\sigma} \right] & = \E \left[4\int_{\sigma}^{\tau} (P_{r-} - P_{\sigma})^2 d [P]_{r} \Bigg| \F_{\sigma} \right]\\
		& = \E \left[4 \sum_{\sigma < t_i \le \tau} (P_{t_i-} - P_{\sigma})^2 (\Delta P_{t_i})^2 \Bigg| \F_{\sigma} \right]\label{eq:variance_of_variance_price_increment}\\
		& = \E \left[ \sum_{\sigma<t_i\le\tau} A_i^2(\sigma) \Bigg| \mathcal{F}_{\sigma} \right].
	\end{align}
	By iteratively using the binomial formula, the fourth power of the intraday return $r(\sigma, \tau)$ can be written as
	\begin{align}
		(P_{\tau}-P_{\sigma})^4 & = \left(\sum_{\sigma < t_i \le \tau} \Delta P_{t_i} \right) ^ 4\\    
		& =  6 \sum_{\sigma < t_i \le \tau} (P_{t_{i}-} - P_{\sigma})^2 (\Delta P_{t_{i}})^2 + \sum_{\sigma < t_i \le \tau} (\Delta P_{t_{i}})^4 +  Q^{\sigma}( \tau - \sigma) \label{eq:fourth_power_price_increment}
	\end{align}
	where $Q^{\sigma}$ is a process defined by
	
	\begin{equation}
		Q^{\sigma}(t) = 4 \sum_{\sigma \le t_i \le \sigma + t \wedge T} (P_{t_{i}-} - P_{\sigma})^3 \Delta P_{t_{i}} + 4 \sum_{\sigma \le t_i \le \sigma + t \wedge T} (P_{t_{i}-} - P_{\sigma}) (\Delta P_{t_{i}})^3
	\end{equation} 
	for $t\ge 0$. If we show that $\E[Q^{\sigma}( \tau - \sigma) | \mathcal{F}_{\sigma}] = 0$ we can conclude from \eqref{eq:variance_of_variance_price_increment} and \eqref{eq:fourth_power_price_increment} that we have that 
	
	\begin{align}
		\E \left[((P_{\tau}-P_{\sigma})^2 - ([P]_{\tau} - [P]_{\sigma}))^2 \big| \F_{\sigma} \right]  & = \frac{2}{3} \E \left[(P_{\tau}-P_{\sigma})^4 - \sum_{\sigma < t_i \le \tau} (\Delta P_{t_{i}})^4   \Bigg| \F_{\sigma} \right],
	\end{align}
	which implies the result, since we can apply Lemma~\ref{lem:random_sum_cond_exp} to show that
	\begin{align}
		\E\left[\sum_{\sigma<t_i\le\tau} (\Delta P_{t_i})^4 \Bigg| \mathcal{F}_\sigma \right] = \E\left[\sum_{\sigma<t_i\le\tau} \E[(\Delta P_{t_i})^4 |\mathcal{F}_{t_i-}] \Bigg|\mathcal{F}_\sigma \right] =\E\left[\sum_{\sigma<t_i\le\tau} 3\varsigma(t_i)^4 \Bigg|\mathcal{F}_\sigma \right] = 3\E\left[ \IQ(\sigma,\tau) |\mathcal{F}_\sigma\right].
	\end{align}
	To show that $\E[Q^{\sigma}( \tau - \sigma) | \mathcal{F}_{\sigma}] = 0 $, we use the Optional Stopping Theorem. To this end, we begin by showing that $Q^{\sigma}$ is a martingale with respect to the filtration $\{\mathcal{F}_{\sigma+t \wedge T}\}_{t \ge 0}$. Clearly, $Q^{\sigma}$ is adapted to the filtration $\{\mathcal{F}_{\sigma+t \wedge T}\}_{t \ge 0}$ and $Q^{\sigma}(0)=0$ is integrable. The martingale property follows as
	
	\begin{align}
		&\E[ Q^{\sigma}(t)-Q^{\sigma}(s)|\mathcal{F}_{\sigma+s \wedge T}] \\
		&\quad=\E\left[4 \sum_{\sigma + s \wedge T < t_i \le \sigma + t \wedge T} (P_{t_{i}-} - P_{\sigma})^3 \Delta P_{t_{i}} \Bigg|\mathcal{F}_{\sigma+s \wedge T} \right]\\
		&\qquad+\E\left[4 \sum_{\sigma + s \wedge T < t_i \le \sigma + t \wedge T} (P_{t_{i}-} - P_{\sigma}) (\Delta P_{t_{i}})^3\Bigg|\mathcal{F}_{\sigma+s \wedge T} \right]\\
		\label{eq:Qeq3}
		&\quad=\E\left[4 \sum_{\sigma + s \wedge T < t_i \le \sigma + t \wedge T} (P_{t_{i}-} - P_{\sigma})^3 \E[\Delta P_{t_{i}} | \mathcal{F}_{t_i-}] \Bigg|\mathcal{F}_{\sigma+s \wedge T} \right]\\
		&\qquad+\E\left[4 \sum_{\sigma + s \wedge T < t_i \le \sigma + t \wedge T} (P_{t_{i}-} - P_{\sigma}) \E[(\Delta P_{t_{i}})^3 | \mathcal{F}_{t_i-}]\Bigg|\mathcal{F}_{\sigma+s \wedge T} \right]\\
		&\quad=\E\left[4 \sum_{\sigma + s \wedge T < t_i \le \sigma + t \wedge T} (P_{t_{i}-} - P_{\sigma})^3 \varsigma(t_i)\E[U_i | \mathcal{F}_{t_i-}] \Bigg|\mathcal{F}_{\sigma+s \wedge T} \right]\\
		&\qquad+\E\left[4 \sum_{\sigma + s \wedge T < t_i \le \sigma + t \wedge T} (P_{t_{i}-} - P_{\sigma}) \varsigma(t_i)^3\E[U_i^3 | \mathcal{F}_{t_i-}]\Bigg|\mathcal{F}_{\sigma+s \wedge T} \right]\\
		&\quad=0,
	\end{align}
	where in the last step we use that at each jump time the Brownian increments $U_i|\mathcal{F}_{t_i-}\sim\mathcal{N}(0,1)$ under Assumption~\ref{ass:filtration} such that $\E[U_i|\mathcal{F}_{t_i-}]=\E[U_i^3|\mathcal{F}_{t_i-}]=0$. In the second step, we apply Lemma \ref{lem:random_sum_cond_exp} in a similar way as in the proof of Proposition \ref{prop:martingale_price}. Note that $Q^\sigma(t)$ can be written in terms of $(P_t-P_\sigma)^4$, $\sum_{t_i\le t} (P_{t-}-P_\sigma)^2 \Delta P_{t_i}^2$ and $\sum_{t_i\le t} \Delta P_{t_i}^4 $ and it is possible to bound these latter terms in expectation by $\E[[P]_T^2]$ by applying the Burkholder-Davis-Gundy inequalities. The Optional Stopping Theorem now gives the desired result that 
	\begin{equation}
		\E[Q^{\sigma}(\tau-\sigma)|\mathcal{F}_{\sigma}]=\E[Q^{\sigma}(0)|\mathcal{F}_{\sigma}]=0.
	\end{equation}
	To show the last equation in the statement of the Lemma, we use the integration by parts formula for the stochastic integral and work out the resulting expectation by using the conditioning as in Lemma \ref{lem:random_sum_cond_exp}:
	\begin{align}
		&\E \left[4 \sum_{\sigma < t_i \le \tau} (P_{t_i-} - P_{\sigma})^2 (\Delta P_{t_i})^2 \Bigg| \F_{\sigma} \right] = 4\E \left[ \int_\sigma^\tau (P_{r-} - P_{\sigma})^2 d \rIV(\sigma,r) \Bigg| \F_{\sigma}\right]\\ 
		&\quad = 4 \E \Bigg[ (P_{\tau} - P_{\sigma})^2 \rIV(\sigma,\tau) - \int_\sigma^\tau \rIV(\sigma,r-) d( (P_{\cdot} - P_{\sigma})^2)_r -\left[ \rIV(\sigma,\cdot), (P_{\cdot} - P_{\sigma})^2\right]_\tau  \Bigg| \F_{\sigma}\Bigg] \\
		&\quad = 4 \E \Bigg[ (P_{\tau} - P_{\sigma})^2 \rIV(\sigma,\tau) - \int_\sigma^\tau \rIV(\sigma,r-) d[P]_r - \int_\sigma^\tau \varsigma^2(r) d[P]_r   \Bigg| \F_{\sigma}\Bigg] \\
		&\quad = 4 \E \Bigg[ (P_{\tau} - P_{\sigma})^2 \rIV(\sigma,\tau) - \frac{1}{2} \left( \rIV(\sigma, \tau)^2 - \int_\sigma^\tau \varsigma^4(r)dN(r) \right) - \int_\sigma^\tau \varsigma^4(r)dN(r) \Bigg| \F_{\sigma}\Bigg] \\
		&\quad = 2 \E\left[\rIV(\sigma,\tau)(2(P_\tau - P_\sigma)^2 - \rIV(\sigma,\tau)) | \mathcal{F}_{\sigma} \right] - 2 \E \left[ \IQ(\sigma, \tau)| \mathcal{F}_{\sigma} \right].
	\end{align}
	In the second equality we use that $\E\left[\int_\sigma^\tau \rIV(\sigma, r-) P_{r-} dP_r |\mathcal{F}_{\sigma} \right] = 0$, which follows from Lemma \ref{lem:random_sum_cond_exp}, and the integrability can be shown to follow from Assumption \ref{ass:filtration} in which we assume that $\E[[P]_T^2]$ and $\E[\rIV(0,T)^2]$ are finite. 
\end{proof}

\begin{proof}[\bf \hypertarget{proof:MSE_jump_based}{Proof of Corollary \ref{cor:MSE_jump_based}}]
	The contribution of the sampling scheme to the MSE is only through the first term in Equation \eqref{eq:MSE_computation}, which is computed in Lemma \ref{lem:square_incr_RV_QV}. Instead of conditioning on $\mathcal{F}_{\tau_{j-1}}$ in the first term in \eqref{eq:MSE_computation}, we now choose to condition on $\mathcal{F}^{\lambda,\varsigma,N}_{\tau_j}$ and we can still apply Lemma \ref{lem:random_sum_cond_exp}, because the sampling scheme $\tau$ is $\FilintN$-measurable such that we have that 
	\begin{align}
		&\E\left[\sum_{j=1}^M \E \left[ \sum_{\tau_{j-1}<t_i\le\tau_j} A_i^2(\tau_{j-1}) \Bigg| \mathcal{F}_{\tau_{j-1}} \right] \right] \\
		&\quad = 2\E\left[\sum_{j=1}^M \E \left[  \rIV(\tau_{j-1}, \tau_{j})\left(2\left(P_{\tau_{j}} - P_{\tau_{j-1}}\right)^2 - \rIV(\tau_{j}, \tau_{j-1})\right) \Big| \mathcal{F}^{\lambda, \varsigma, N}_{\tau_{j}} \right] \right]\\
		&\qquad - 2\E\left[\sum_{j=1}^M \E \left[  \IQ(\tau_{j}, \tau_{j-1}) \Big| \mathcal{F}^{\lambda, \varsigma, N}_{\tau_{j}} \right] \right]\\
		&\quad = 2\E\left[\sum_{j=1}^M \rIV(\tau_{j-1}, \tau_{j}) \left( 2 \E \left[  \left(P_{\tau_{j}} - P_{\tau_{j-1}}\right)^2 \Big| \mathcal{F}^{\lambda, \varsigma, N}_{\tau_{j}} \right] - \rIV(\tau_{j}, \tau_{j-1})\right)  \right] - 2\E\left[ \IQ(0, T) \right].
	\end{align}
	Under the assumption that $U_i^2$ for any $i=1,...,N(T)$ is independent of the paths of $\lambda$, $\varsigma$ and $N$, we have that
	
	\begin{align}
		\E \left[  \left(P(\tau_{j}) - P(\tau_{j-1})\right)^2 \Big| \mathcal{F}^{\lambda, \varsigma, N}_{\tau_{j}} \right] & = \rIV(\tau_{j-1}, \tau_j) + \E \left[  \left(P(\tau_{j}) - P(\tau_{j-1})\right)^2 - ([P]_{\tau_j} - [P]_{\tau_{j-1}})  \Big| \mathcal{F}^{\lambda, \varsigma, N}_{\tau_{j}} \right]
	\end{align}
	The result now follows by applying Lemma \ref{lem:random_sum_cond_exp} once more.
\end{proof}

\begin{proof}[\bf \hypertarget{proof:MSE_intensity_based}{Proof of Corollary \ref{cor:MSE_intensity_based}}]
	By the It\^{o} isometry for the stochastic integral we have that
	\begin{align}
		\E\left[ \rIV(\tau_{j-1}, \tau_j)^2 | \mathcal{F}_{\tau_{j-1}} \right] & = \E \left[ \IV(\tau_{j-1}, \tau_j)^2 + \IQ(\tau_{j-1}, \tau_j) + 2 \IV(\tau_{j-1}, \tau_j) \int_{\tau_{j-1}}^{\tau_j}\varsigma^2(r)d\tilde{N}(r) \Bigg| \mathcal{F}_{\tau_{j-1}} \right].
	\end{align}
	Using this result and applying Lemma \ref{lem:random_sum_cond_exp} to the MSE result in \eqref{eq:MSE_realized_sampling}, we find for an $\Filint$-adapted sampling scheme that
	\begin{align}
		2\E \left[\sum_{j=1}^M \rIV(\tau_{j-1}, \tau_{j})^2 \right] & = 2 \E \left[\sum_{j=1}^M \IV(\tau_{j-1}, \tau_{j})^2 \right] + 2 \E \left[\IQ(0,T) \right] + \E \left[\tilde{R}(\boldsymbol{\tau}) \right],
	\end{align}
	which implies the MSE result in \eqref{eq:MSE_intensity}.
\end{proof}

\begin{defn}  \label{def:M_of_A}
	For any random sequence $A = \{A_1,A_2,...\}$ taking values in $\mathbb{R}_{\ge0}^\infty$ such that eventually $A_j=0$ for large enough $j$, we define $M(A)=\min(m\in\mathbb{N}:A_j=0 \text{ for all } j > m)$.
\end{defn}

\begin{lem} \label{lem:min_squared_sum}
	Given two constants $\bar{M}\in \mathbb{N}$ and $Q\in\mathbb{R}_{>0}$, denote by $\mathcal{A}(\bar{M}, Q)$ the collection of all random sequences $A$ taking values in $\mathbb{R}_{\ge0}^\infty$ such that $M(A) > 0$ almost surely and $\E[M(A)]=\bar{M}$ and $\E\left[ \sum_{j=1}^{M(A)} A_j\right] = Q$, where $M(A)$ is defined in Definition \ref{def:M_of_A}. The minimization 
	
	\begin{equation}
		\min_{A\in \mathcal{A}(\bar{M}, Q)} \E\left[ \sum_{j=1}^{M(A)} A_j^2\right] \label{eq:minimization}
	\end{equation}
	is attained by the deterministic sequence $A^*$ such that $A^*_j = \frac{Q}{\bar{M}}$ for $j\le \bar{M}$ and $A^*_j=0$ for $j > \bar{M}$.
\end{lem}

\begin{proof}[\bf Proof of Lemma \ref{lem:min_squared_sum}]	
	A lower bound for the minimization objective in \eqref{eq:minimization} follows by applying the Cauchy-Schwarz inequality twice:
	\begin{align}
		\E\left[ \sum_{j=1}^{M(A)} A_j^2\right] & \ge \E\left[ \frac{\left(\sum_{j=1}^{M(A)} A_j\right)^2}{M(A)}\right]\\
		& \ge \frac{Q^2}{\E[M(A)]}\label{eq:lower_bound_cs}.
	\end{align}
	The first application of the Cauchy-Schwarz inequality is for the standard $l^2$ inner product for square-summable sequences and the second inequality is for the inner product for random variables given by $\E[XY]$, where we choose $X=\frac{\sum_{j=1}^{M(A)} A_j}{\sqrt{M(A)}}$ and $Y=\sqrt{M(A)}$. It is straightforward to show that $A^*$ satisfies the required conditions and that the lower bound in \eqref{eq:lower_bound_cs} is reached for $A^*$.
\end{proof}

\subsection{Proofs for Appendix~\ref{sec:EfficientSamplingUptoT}}
\label{subsec:ProofsEfficientSamplingUptoT}

\begin{proof}[\bf \hypertarget{proof:bias}{Proof of Theorem \ref{prop:bias}}]
	Let $\left\{t_i\right\}_{i=n}^{m}$ with $t_n<\ldots<t_m,  n,m\in\N$ and $n\leq m$ denote the sequence of arrival times in the interval $(\tau_{j-1},\tau_j]$. 
	Then, it holds that
	\begin{align}
		\begin{aligned}
			\label{eq:binom_sum}
			&\E\left[ r_{j}^2 \middle| \FintN_T \right]
			= \E\left[ \left(\int_{\tau_{j-1}}^{\tau_j} \varsigma\left(r\right)dB\left(N\left(r\right)\right)\right)^2 \middle| \FintN_T \right] 
			= \E\left[ \left(\sum_{\tau_{j-1}  <t_i\leq \tau_j } \varsigma\left(t_i\right) U_i\right)^2 \middle| \FintN_T \right] \\
			=\;&\E\left[ \left(\sum_{t_n\leq t_i\leq t_m} \varsigma \left(t_i\right)U_i\right)^2 \middle| \FintN_T\right] 
			= \E\left[ \left(\sum_{t_n\leq t_i\leq t_{m-1}} \varsigma\left(t_i\right) U_i+\varsigma\left(t_m\right)U_m\right)^2 \middle| \FintN_T \right] \\
			=\; &\E\left[\left(\sum_{t_n\leq t_i\leq t_{m-1}} \varsigma\left(t_i\right)U_i\right)^2  +\left(\varsigma\left(t_m\right)U_m\right)^2 +2\left(\sum_{t_n\leq t_i\leq t_{m-1}} \varsigma\left(t_i\right)U_i\right)\varsigma\left(t_m\right)U_m\middle|\FintN_T\right]. 
		\end{aligned}
	\end{align}
	From Assumption \eqref{ass:filtration} and the independence in Assumption \eqref{ass:Independence}, we obtain $U_i \mid \FintN_T \sim \mathcal{N}\left(0,1\right)$ and $U_i \mid \FintN_T\vee\mathcal{F}_{t_i-} \sim \mathcal{N}\left(0,1\right)$. Using the predictability of $\varsigma$ and the tower property, noting that $ \mathcal{F}^{\lambda, \varsigma}_T \subset \left( \FintN_T \vee \mathcal{F}_{t_m-} \right)$, it follows that 
	\begin{align*}
		&\E\left[\left(\sum_{t_n\leq t_i\leq t_{m-1}} \varsigma\left(t_i\right)U_i\right)\varsigma\left(t_m\right)U_m\middle|\FintN_T\right] \\ 		
		& = \E\left[ \E\left[ \left(\sum_{t_n\leq t_i\leq t_{m-1}} \varsigma\left(t_i\right)U_i\right)\varsigma\left(t_m\right)U_m \middle| \FintN_T \vee \mathcal{F}_{t_m-} \right]  \middle|\FintN_T\right] \\ 
		& = \E\left[ \left(\sum_{t_n\leq t_i\leq t_{m-1}} \varsigma\left(t_i\right)U_i\right) \varsigma\left(t_m\right) \E\left[ U_m \middle| \FintN_T \vee \mathcal{F}_{t_m-} \right]  \middle|\FintN_T\right] = 0,
	\end{align*}
	and thus, the third term in the last row of \eqref{eq:binom_sum} is zero. 
	Similarly,
	\begin{align*}
		\E\left[ \left(\varsigma\left(t_m\right)U_m\right)^2 \middle| \FintN_T \right]
		& = \E\left[ \E\left[ \left(\varsigma\left(t_m\right)U_m\right)^2 \middle| \FintN_T \vee \mathcal{F}_{t_m-} \right] \middle|\FintN_T\right] \\
		& = \E\left[ \varsigma^2\left(t_m\right) \E\left[ U_m^2 \middle| \FintN_T \vee \mathcal{F}_{t_m-} \right] \middle|\FintN_T\right] \\
		& = \E\left[ \varsigma^2\left(t_m\right) \middle|\FintN_T\right].
	\end{align*}
	Repeatedly splitting up the squared sum in \eqref{eq:binom_sum} hence yields
	\begin{align*}
		\E\left[ r_{j}^2\middle|\FintN_T\right] &=\E\left[\sum_{t_n\leq t_i\leq t_{m}} \varsigma^2\left(t_i\right)\middle|\FintN_T\right] 
		= \E\left[\sum_{\tau_{j-1}< t_i\leq \tau_j } \varsigma^2\left(t_i\right)\middle|\FintN_T\right]
		\\ &=\E\left[\int_{\tau_{j-1} }^{\tau_j }\varsigma^2\left(r\right)dN\left(r\right)\middle|\FintN_T\right]\\
		&=\int_{\tau_{j-1} }^{\tau_j }\varsigma^2\left(r\right)dN\left(r\right).
	\end{align*} 
	Summing up, we get that
	\begin{align*} \E\left[ \RV(\btau )\middle|\FintN_T\right] &=
		\E\left[ \sum_{j=1}^{M(T)} r_{j}^2\middle|\FintN_T\right] 
		= \int_{0}^{T}\varsigma^2\left(r\right)dN(r)  =  \rIV\left(0,T\right).
	\end{align*} 
	
	Given the additional Assumption \eqref{ass:intensity_restrictions} we use the Doob-Meyer decomposition of the jump process into the zero-mean martingale $\Tilde{N}(t)$ w.r.t $\mathcal{F}_t$ and the $\mathcal{F}_t$-predictable compensator $\int_0^t\lambda\left(r\right)dr$. For $\Tilde{N}\left(t\right)=N\left(t\right)-\int_0^t\lambda\left(r\right)dr$, \citet[Lemma L3, page 24]{bremaud1981} yields that $\int_0^T\varsigma^2(r) d\Tilde{N}(r)$ also has a zero-mean conditioning on $\Fint_T$.\footnote{With the more general jump process, the information set $\Fint_T$ could also contain the information of $N$ which would result in $\Tilde{N}$ being $\Fint_T$-measurable. The conditional expectation wouldn't be zero anymore.}
	Hence with the the tower property, we obtain:
	\begin{align*} 
		\E\left[ \RV(\btau )\middle|\Fint_T\right] &= \E\left[ \E\left[\RV(\btau )\middle|\FintN_T\right]\middle|\Fint_T\right]\\
		&= \E\left[\rIV(0,T)\middle|\Fint_T\right]= \E\left[\int_{0}^{T}\varsigma^2\left(r\right)dN(r)\middle|\Fint_T\right]\\
		&= \E\left[\int_0^T\varsigma^2\left(r\right)d\Tilde{N}\left(r\right)\middle|\Fint_T\right] + \E\left[\int_0^T\varsigma^2\left(r\right)\lambda\left(r\right)dr\middle|\Fint_T\right] \\
		&=\int_0^T\varsigma^2\left(r\right)\lambda\left(r\right)dr=\IV(0,T),
	\end{align*} 
	which finishes this proof.\\ \qedhere
\end{proof}

\begin{proof}[\bf \hypertarget{proof:MSE_tick}{Proof of Theorem \ref{prop:MSE_tick}}] 
	We begin by proving part (a):
	Given Assumptions \eqref{ass:filtration}, \eqref{ass:Independence} and \eqref{ass:Moments}, we get that
	\begin{align}
		\begin{aligned}
		& \E\left[\left(\RV(\btau )-\IV(0,T)\right)^2\middle|\FintN_T\right] \\
		&=\E\left[\left(\RV(\btau )-\rIV(0,T)+\rIV(0,T)-\IV(0,T)\right)^2\middle|\FintN_T\right]\\
		&=\E\left[\left(\RV(\btau )-\rIV(0,T)\right)^2\middle|\FintN_T\right]
		\\&\qquad+2\E\left[\left(\RV(\btau )-\rIV(0,T)\right)\left(\rIV(0,T)-\IV(0,T)\right)\middle|\FintN_T\right]
		\\&\qquad+\E\left[\left(\rIV(0,T)-\IV(0,T)\right)^2\middle|\FintN_T\right]\\
		&=\E\left[\left(\RV(\btau)-\rIV(0,T)\right)^2\middle|\FintN_T\right]+\left(\rIV(0,T)-\IV(0,T)\right)^2.
		\label{eq:MSEdecomp}
		\end{aligned}
	\end{align}
	The mixed term disappears since $\E\left[\RV(\btau)-\rIV(0,T)\middle|\FintN_T\right]=0$ and $\left(\rIV(0,T)-\IV(0,T)\right)$ is $\FintN_T$-measurable. 
	We proceed by calculating the first term. From the conditional unbiasedness in Theorem \ref{prop:bias}, it follows that
	\begin{align} 
		\begin{aligned} 
			\label{eq:IVdecomp}
			&\E\left[ \left(\RV(\btau)-\rIV\left(0,T\right)\right)^2\middle|\FintN_T\right] 
			\\ &~~~~~= \E\left[ \left(\RV(\btau)\right)^2-2\RV(\btau)\rIV\left(0,T\right)+\rIV\left(0,T\right)^2\middle|\FintN_T\right] 
			\\ &~~~~~= \E\left[ \left(\RV(\btau)\right)^2\middle|\FintN_T\right] -\rIV\left(0,T\right)^2.
		\end{aligned}
	\end{align}
	Applying the multinomial theorem, we get
	\begin{align}
		\label{eq:RV^2}
		\big( \RV(\btau) \big)^2  &= \left( \sum_{j=1}^{M(T)} r_{j}^2 \right)^2 
		=  \sum_{j=1}^{M(T)} r_{j}^4 +  \sum_{\substack{j,k = 1 \\ j \neq k}}^{M(T)} r_{j}^2 r_{k}^2. 
	\end{align}
	We now split the proof into three parts:
	
	\proofpart{1} We begin by analyzing the first term in \eqref{eq:RV^2}. 
	Let $\left\{t_i\right\}_{i=n}^{m}$ with $t_n<\ldots<t_m,  n,m\in\N$ and $n\leq m$ denote the series of jump times of the counting process $N$ in the interval $(\tau_{j-1},\tau_j]$. 
	By subsequently detaching the smallest term in the sums to the fourth power and applying the binomial theorem, we get for all $j=1,\ldots,M(T)$ that
	\begin{align} 
		\E\left[ r_{j}^4 \middle| \FintN_T\right] 
		&=  \E\left[ \left(\sum_{\tau_{j-1}< t_i\leq \tau_j} \varsigma\left(t_i\right)U_i\right)^4 \middle| \FintN_T\right] \nonumber \\
		&=  \E\left[ \left(\sum_{t_{n+1}\leq t_i\leq t_{m}} \varsigma\left(t_i\right)U_i\right)^4 + \varsigma^4\left(t_n\right)U^4_n\right. \nonumber \\
		&~~~~~~~\left.+4\left(\sum_{t_{n+1}\leq t_i\leq t_{m}} \varsigma\left(t_i\right)U_i\right)^3 \varsigma\left(t_n\right)U_n\right. \nonumber \\
		&~~~~~~~\left.+6\left(\sum_{t_{n+1}\leq t_i\leq t_{m}}\varsigma\left(t_i\right)U_i\right)^2\varsigma^2\left(t_n\right)U^2_n\right. \nonumber \\
		&~~~~~~~\left.+4\left(\sum_{t_{n+1}\leq t_i\leq t_{m}}\varsigma\left(t_i\right)U_i\right)\varsigma^3\left(t_n\right)U^3_n\middle|\FintN_T\right]\nonumber
		\\ &= \E\left[ 3\sum_{\tau_{j-1}< t_i\leq \tau_j}\varsigma^4\left(t_i\right)+6\sum_{\tau_{j-1}< t_i< \tau_j}\sum_{t_{i+1}\leq t_h\leq \tau_{j}}\varsigma^2\left(t_h\right)\varsigma^2\left(t_i\right)\middle|\FintN_T\right] \nonumber
		\\ &=\E\left[ 3\left(\sum_{\tau_{j-1}< t_i\leq \tau_j}\varsigma^2\left(t_i\right)\right)^2\middle|\FintN_T\right] \nonumber
		\\ &=3\rIV(\tau_{j-1},\tau_j)^2,         
		\label{eq:Erj4}
	\end{align}
	where we use Assumption \eqref{ass:Independence}, and especially, the moment structure of $U_i \mid \FintN_T \sim \mathcal{N}\left(0,1\right)$ resulting from Assumption \eqref{ass:filtration} and \eqref{ass:Independence}. 
	
	\proofpart{2} 
	We continue by simplifying the second term in \eqref{eq:RV^2}. 
	For the non-overlapping intervals $(\tau_{j-1},\tau_{j}]$ and $(\tau_{k-1},\tau_k]$ for $j\neq k$, it holds that
	\begin{align} 
		\begin{aligned} 
		\E\left[ r_{j}^2r_{k}^2\middle|\FintN_T\right]  &= \E\left[ \left(\sum_{\tau_{j-1}<t_i\leq \tau_j}\varsigma\left(t_i\right)U_i\right)^2\left(\sum_{\tau_{k-1}< t_i\leq \tau_{k}}\varsigma\left(t_i\right)U_i\right)^2\middle|\FintN_T\right]
		\\ &= \E\left[ \left(\sum_{\tau_{j-1}< t_i\leq \tau_j}\varsigma^2\left(t_i\right)\right)\left(\sum_{\tau_{k-1}< t_i\leq \tau_{k}}\varsigma^2\left(t_i\right)\right)\middle|\FintN_T\right] 
		\\ &= \E\left[ \left(\int_{\tau_{j-1}}^{\tau_j}\varsigma^2\left(r\right)dN\left(r\right)\right)\left(\int_{\tau_{k-1}}^{\tau_{k}}\varsigma^2\left(r\right)dN\left(r\right)\right)\middle|\FintN_T\right] 
		\\ &= \left(\int_{\tau_{j-1}}^{\tau_j}\varsigma^2\left(r\right)dN\left(r\right)\right)\left(\int_{\tau_{k-1}}^{\tau_{k}}\varsigma^2\left(r\right)dN\left(r\right)\right) 
		\\ &= \rIV(\tau_{j-1},\tau_j)\rIV(\tau_{k-1},\tau_k),
		\label{eq:Erj2rk2}
		\end{aligned}
	\end{align} 
	due to the independence of $\varsigma\left(t_i\right)$ and $U_i$. 

	\proofpart{3} 
	We proceed by inserting the results from \eqref{eq:Erj4} and \eqref{eq:Erj2rk2} into equation \eqref{eq:RV^2} and summing them up according to \eqref{eq:IVdecomp}. We get
	\begin{align*} 
		\E\left[ \left(\RV(\btau)-\rIV(0,T)\right)^2  \middle| \FintN_T\right]&= \E\left[ \left(\RV(\btau)\right)^2 \middle| \FintN_T\right] - \rIV\left(0,T\right)^2 \\
		= \, &3 \sum_{j=1}^{M(T)}  \rIV\left(\tau_{j-1},{\tau_j}\right)^2\\
		&~~~~+\sum_{\substack{j,k = 1 \\ j \neq k}}^{M}  \rIV\left(\tau_{j-1},{\tau_j}\right)\rIV\left({\tau_{k-1}},{\tau_{k}}\right) - \rIV\left(0,T\right)^2 \nonumber\\
		= \, &2 \sum_{j=1}^{M(T)}  \rIV\left(\tau_{j-1},{\tau_j}\right)^2+\rIV\left(0,T\right)^2 - \rIV\left(0,T\right)^2 \\
		= &2\sum_{j=1}^{M(T)}\rIV\left(\tau_{j-1},{\tau_j}\right)^2.
	\end{align*}
	Inserting this result into \eqref{eq:MSEdecomp} then yields the claim (a):
	\begin{align}
		\E\left[\left(\RV(\btau)-\IV(0,T)\right)^2\middle|\FintN_T\right]&=\left(\rIV(0,T)-\IV(0,T)\right)^2+2\sum_{j=1}^{M(T)}\rIV(\tau_{j-1},\tau_j)^2.
	\end{align}
	
	We proceed to show the claim (b): 
	\noindent
	Let Assumptions \eqref{ass:filtration}--\eqref{ass:intensity_restrictions} hold. We calculate the conditional MSE of $\RV(\btau)$ on $\Fint_T$ by taking the conditional expectation of the result in claim (a). With the tower property the following holds:
	\begin{align}
		\begin{aligned} 
		& \E\left[\left(\RV(\btau)-\IV(0,T)\right)^2\middle|\Fint_T\right] \\
		&=\E\left[\E\left[\left(\RV(\btau)-\IV(0,T)\right)^2\middle|\FintN_T\right]\Fint_T\right]\\
		&=\E\left[\left(\rIV(0,T)-\IV(0,T)\right)^2+2\sum_{j=1}^{M(T)}\rIV(\tau_{j-1},\tau_j)^2\middle|\Fint_T\right]\\
		&=\E\left[\left(\rIV(0,T)-\IV(0,T)\right)^2\middle|\Fint_T\right] + 2\sum_{j=1}^{M(T)}\E\left[\rIV(\tau_{j-1},\tau_j)^2\middle|\Fint_T\right].
		\label{eq:MSEdecomp2}
		\end{aligned} 
	\end{align}
	We begin by calculating the first term. Note that the following result only applies to sampling schemes $\btau$ that are $\Fint_T$-measurable.
	We denote the compensated jump process by $\Tilde{N}(t) = N(t)-\int_0^t\lambda(r)dr$, and get
	\begin{align*}
		\E\left[\rIV(\tau_{j-1},\tau_j)^2\middle|\Fint_T\right] &= \E\left[\left(\int_{\tau_{j-1}}^{\tau_j}\varsigma^2\left(r\right)dN\left(r\right)\right)^2\middle|\Fint_T\right]
		\\ &= \E\left[\left(\int_{\tau_{j-1}}^{\tau_j}\varsigma^2\left(r\right)d\Tilde{N}\left(r\right)+\int_{\tau_{j-1}}^{\tau_j}\varsigma^2\left(r\right)\lambda\left(r\right)dr\right)^2\middle|\Fint_T\right]
		\\ &= \E\left[\left(\int_{\tau_{j-1}}^{\tau_j}\varsigma^2\left(r\right)d\Tilde{N}\left(r\right)\right)^2+2\int_{\tau_{j-1}}^{\tau_j}\varsigma^2\left(r\right)d\Tilde{N}\left(r\right)\int_{\tau_{j-1}}^{\tau_j}\varsigma^2\left(r\right)\lambda\left(r\right)dr \right.
		\\ &~~~~~~~~~~+\left.\left(\int_{\tau_{j-1}}^{\tau_j}\varsigma^2\left(r\right)\lambda\left(r\right)dr\right)^2\middle|\Fint_T\right]
		\\ &= \E\left[\left(\int_{\tau_{j-1}}^{\tau_j}\varsigma^2\left(r\right)d\Tilde{N}\left(r\right)\right)^2\middle|\Fint_T\right]
		\\ &~~~~~~~~~~+2\E\left[\int_{\tau_{j-1}}^{\tau_j}\varsigma^2\left(r\right)d\Tilde{N}\left(r\right)\middle|\Fint_T\right] \IV\left(\tau_{j-1},\tau_j\right)
		\\ &~~~~~~~~~~+\IV\left(\tau_{j-1},\tau_j\right)^2.
	\end{align*}
	The second term above is zero due to the zero-mean martingale property of $\int_0^t\varsigma^2\left(r\right)d\Tilde{N}\left(r\right)$ w.r.t $\Fint_T$ based on Assumption \eqref{ass:intensity_restrictions} (see \citet[Lemma L3, page 24]{bremaud1981}).\footnote{The martingale property is w.r.t. the filtration $\mathcal{G}_t:=\Fint_T\vee \mathcal{F}_t$, i.e. with respect to the filtration of the smallest $\sigma$-algebras containing both $\Fint_T$ and $\mathcal{F}_t$. We specifically need the zero-mean property which is fulfilled in case of a doubly stochastic Poisson process since the trades arrive independently and are can not be recovered from the evolution of $\lambda$.} 
	To further simplify the first term, we need the quadratic variation $\big[\Tilde{N}\big]_t$ since by the It\^{o}'s isometry for martingales it holds that
	$$
	\E\left[\left(\int_{\tau_{j-1}}^{\tau_j}\varsigma^2\left(r\right)d\Tilde{N}\left(r\right)\right)^2\middle|\Fint_T\right]=\E\left[\int_{\tau_{j-1}}^{\tau_j}\varsigma^4\left(r\right)d\left[\Tilde{N}\right]_r\middle|\Fint_T\right].
	$$
	Further let $0=s_0<s_1<\ldots<s_n=t$ denote a partition of $[0,t]$ such that 
	$$\max_{1\leq k\leq n}\left|s_{k}-s_{k-1}\right|\to0$$ as $n\to\infty$.
	Then, using that $N\left(t\right)$ is a pure jump process and that $t\mapsto\int_{0}^{t}\lambda\left(r\right)dr$ is continuous, we have that 
	\begin{align*} 
		\left[\Tilde{N}\right]_t &= \plim_{n\to\infty}\sum_{k=1}^n\left(\Tilde{N}\left(s_k\right)-\Tilde{N}\left(s_{k-1}\right)\right)^2
		\\ &= \plim_{n\to\infty}\sum_{k=1}^n\left(N\left(s_k\right)-N\left(s_{k-1}\right)+\int_{s_{k-1}}^{s_k}\lambda\left(r\right)dr\right)^2
		\\ &= \plim_{n\to\infty} \sum_{k=1}^n \left\{\big(N\left(s_k\right)-N\left(s_{k-1}\right)\big)^2 + \left(\int_{s_{k-1}}^{s_k}\lambda\left(r\right)dr\right)^2\right\}
		\\ &= \left[N\right]_t+ \left[\int_{0}^{\cdot}\lambda\left(r\right)dr\right]_t
		= \sum_{0<s\leq t}\left(N\left(s\right)-N\left(s-\right)\right)^2
		\\ &= \sum_{0<s\leq t}\left(N\left(s\right)-N\left(s-\right)\right)
		= N\left(t\right).
	\end{align*}
	Hence, it follows that 
	\begin{align*}
		\E\left[\left(\int_{\tau_{j-1}}^{\tau_j}\varsigma^2\left(r\right)d\Tilde{N}\left(r\right)\right)^2\middle|\Fint_T\right] &= \E\left[\int_{\tau_{j-1}}^{\tau_j}\varsigma^4\left(r\right)dN\left(r\right)\middle|\Fint_T\right]
		\\ &= \E\left[\int_{\tau_{j-1}}^{\tau_j}\varsigma^4\left(r\right)d\Tilde{N}\left(r\right)+\int_{\tau_{j-1}}^{\tau_j}\varsigma^4\left(r\right)\lambda\left(r\right)dr\middle|\Fint_T\right]
		\\ &= \E\left[\int_{\tau_{j-1}}^{\tau_j}\varsigma^4\left(r\right)\lambda\left(r\right)dr\middle|\Fint_T\right]
		\\ &= \IQ\left(\tau_{j-1},\tau_j\right),
	\end{align*}
	where we apply the martingale property of $\int_0^t\varsigma^4\left(r\right)d\Tilde{N}\left(r\right)$. We again use the assumption that the sampling scheme $\btau$ is $\Fint_T$-measurable here.
	
	The first term in \eqref{eq:MSEdecomp2} now simplifies the following way:
	\begin{align*}
		\E\left[\left(\rIV(0,T)-\IV(0,T)\right)^2\middle|\Fint_T\right]&=\E\left[\left(\Stil-\Sno\right)^2\middle|\Fint_T\right]
		\\&=\E\left[\left(\int_0\varsigma^2\left(r\right)d\Tilde{N}\left(r\right)\right)^2\middle|\Fint_T\right]
		\\&=\IQ\left(0,T\right).     
	\end{align*}
	For the second term in \eqref{eq:MSEdecomp2} we accordingly find
	\begin{align*}
		2\sum_{j=1}^{M(T)}\E\left[\rIV(\tau_{j-1},\tau_j)^2\middle|\Fint_T\right]&=2\sum_{j=1}^{M(T)}\IV(\tau_{j-1},\tau_j)^2+2\sum_{j=1}^{M(T)}\IQ\left(\tau_{j-1},\tau_j\right)
		\\&=2\sum_{j=1}^{M(T)}\IV(\tau_{j-1},\tau_j)^2+2\IQ(0,T).
	\end{align*}
	Summing the results up yields claim (b) and finishes the proof:
	\begin{align*}
		\E\left[\left(\RV(\btau)-\IV(0,T)\right)^2\middle|\Fint_T\right]&=3\IQ(0,T)+2\sum_{j=1}^{M(T)}\IV(\tau_{j-1},\tau_j)^2.
	\end{align*}
\end{proof}

\pagebreak 
\section{Approximating the remainder term in Corollary~\ref{cor:MSE_jump_based}}
\label{sec:RemainderTerms}

Consider the remainder term from Corollary~\ref{cor:MSE_jump_based} given by
\begin{equation*}
    R(\boldsymbol{\tau}) = 4\sum_{j=1}^M \left((P_{\tau_{j}}-P_{\tau_{j-1}})^2-([P]_{\tau_{j}}-[P]_{\tau_{j-1}})\right) \text{rIV}(\tau_{j-1},\tau_{j}).
\end{equation*}
We start to only consider the initial term of the sum above and refer to the first sampling time as $\tau$.
Then, we have that
\begin{align}
    (P_\tau^2-[P]_\tau ) \cdot \text{rIV}(0,\tau) & = \sum_{0\le t_i < t_j \le \tau} 2 \varsigma(t_i) \varsigma(t_j) U_i U_j \sum_{0 \le t_k \le \tau} \varsigma^2(t_k) \nonumber \\
    & = \sum_{0\le t_i < t_j \le \tau} \sum_{0 \le t_k \le \tau}2 \varsigma(t_i) \varsigma(t_j) U_i U_j  \varsigma^2(t_k), 
    \label{eq:remainder_summation}
\end{align}
where we used that $P(t) = \sum_{0\le t_j \le t} \varsigma(t_j)U_j$. 
We will now see that in expectation, many of the terms in \eqref{eq:remainder_summation} are zero. 
Namely, for all $t_k$ such that $t_k\le t_j$, we can condition on $\Fil_{t_j-}$ (and apply Lemma \ref{lem:random_sum_cond_exp}) such that 
\begin{align*}
    \E[(P_\tau^2-[P]_\tau ) \cdot \text{rIV}(0,\tau)] & =  \E \left[\sum_{0\le t_i < t_j \le \tau} \sum_{t_j < t_k \le \tau}2 \varsigma(t_i) \varsigma(t_j) U_i U_j  \varsigma^2(t_k) \right].
\end{align*}
This last expression shows that the dependence between the tick variance after the jump time $t_j$, i.e., $\varsigma^2(t_k)$, and the product $\varsigma(t_i) \varsigma(t_j) U_i U_j$ is important. 

For the sake of argument, suppose that there exists a $j' \in \mathbb{N}$ such that for each $k\ge j + j'$, the tick variance $\varsigma^2(t_{k})$ is independent from $\varsigma(t_j)$ and $U_j$. 
This characterizes that the dependence of the $\varsigma$ process on its past and on the past of the price-changes dies out after some time (similar to $k$-dependence or $\alpha$-mixing). 
Then, we can also condition on $\Fil_{t_j-}$ and use the independence $\varsigma(t_k) \perp \varsigma(t_j), U_j$, if $t_k\ge t_{j+j'}$, as well as the independence of $\varsigma(t_i) \perp \varsigma(t_j), U_j$, if $t_i\le t_{j-j'}$, such that
\begin{align}
    \E[(P_\tau^2-[P]_\tau ) \cdot \text{rIV}(0,\tau)] & =  \E \left[\sum_{0 < t_j \le \tau} \left\{ \sum_{t_{j-j'} < t_i < t_j} 2 \varsigma(t_i) \varsigma(t_j) U_i U_j \sum_{t_j < t_k \le t_{j+j'} \wedge \tau} \varsigma^2(t_k) \right\}\right].
    \label{eqn:RemainderApproxDependence}
\end{align}
For many of the $t_j$'s in the above sum, the terms do not depend on the sampling time $\tau$. This is only the case if $t_j$ is such that $t_{j+j'}\ge \tau$. So we can approximate 
\begin{equation}
    \E[(P_\tau^2-[P]_\tau ) \cdot \text{rIV}(0,\tau)] \approx \E \left[\sum_{0 < t_j \le \tau} \left\{ \sum_{t_{j-j'} < t_i < t_j} 2 \varsigma(t_i) \varsigma(t_j) U_i U_j \sum_{t_j < t_k \le t_{j+j'}} \varsigma^2(t_k) \right\}\right],
    \label{eqn:RemainderApproxSparse}
\end{equation}
where the approximation is accurate if the sampling time $\tau$ is large in comparison to the time it takes for the dependence between the tick variance and the past price changes and the past tick variance to become negligible. 

Generalizing the previous argument to all sampling points, we get the following approximation for the entire remainder term,
\begin{equation}
    \E[R(\boldsymbol{\tau})] \approx \E \left[\sum_{0 < t_j \le T} \left\{ \sum_{t_{j-j'} < t_i < t_j} 2 \varsigma(t_i) \varsigma(t_j) U_i U_j \sum_{t_j < t_k \le t_{j+j'}} \varsigma^2(t_k) \right\}\right].
    \label{eqn:RemainderApprox}
\end{equation}
Most importantly, the approximation in \eqref{eqn:RemainderApprox} does not depend on the employed sampling scheme such that we conjecture that the efficiency result of Theorem~\ref{thm:EfficientSampling} (b) continues to hold under mild forms of dependencies, as can be seen from our simulation results.
Notice again that the accuracy of the above approximations depends on a dependence structure that dies out quick enough (in \eqref{eqn:RemainderApproxDependence}) and a relatively sparse sampling frequency (in \eqref{eqn:RemainderApproxSparse}).
Although some arguments in this section are informal, they offer valuable intuition and can serve as a foundation for more rigorous mathematical analysis in future research.

\pagebreak 
\section{Additional Empirical Results}
\label{sec:AdditionalResults}

$ $
\begin{figure}[h!] 
	\centering
	\includegraphics[width=1\textwidth]{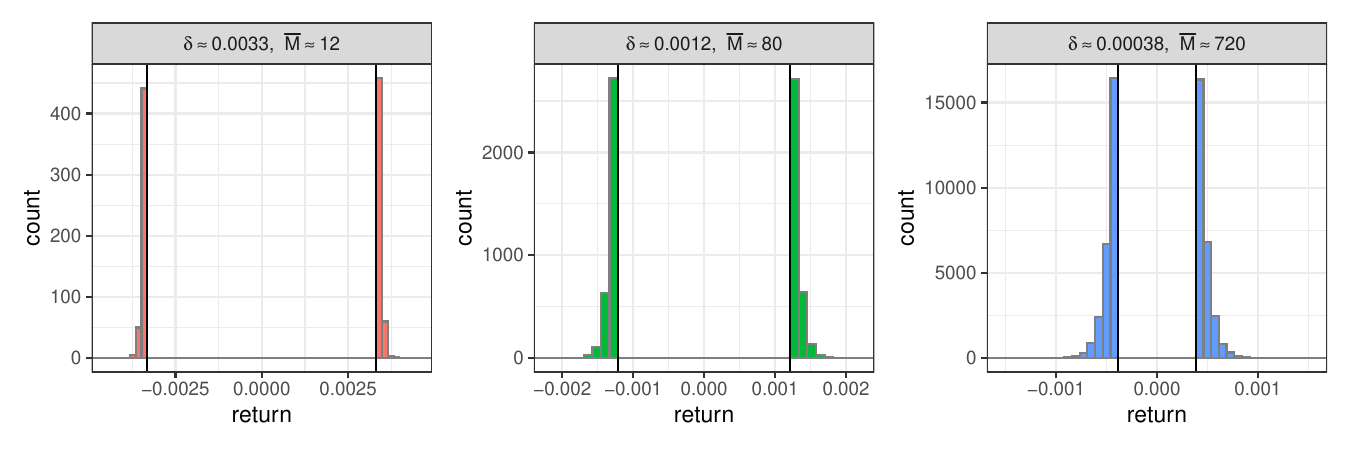}
	\caption{
        Histograms of the simulated HTS returns at different values of $\delta$ in \eqref{eq:HTS_Implementation} (and corresponding average values $\overline{M}$ shown in the plot titles).
        Here, we see the ``overshooting'' effect of the HTS returns in discrete price processes that becomes more severe for smaller values of $\delta$.
	}
	\label{fig:HTS_Overshooting}
\end{figure}

\begin{figure}[tb] 
	\centering
	\includegraphics[width=1\textwidth]{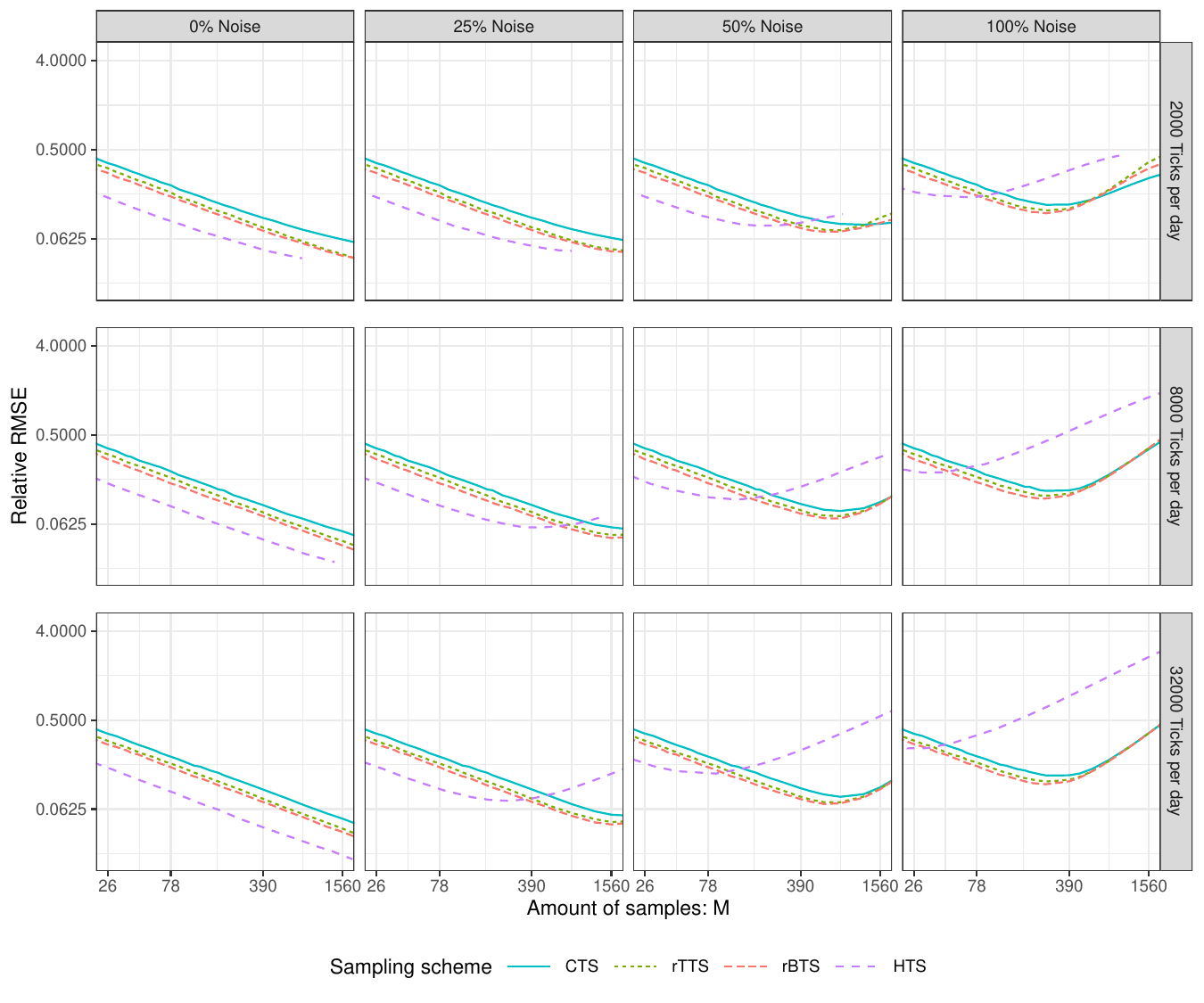}
    \caption{
		Relative RMSE of the RV estimator under the Hawkes-type TTSV process using different sampling schemes in color plotted against the (for HTS average) sampling frequencies $M$ on the horizontal axis.
        The plot columns refer to the (i.i.d.) noise magnitude described below \eqref{eq:SimulateNoise} and the plot rows refer to different amounts of expected ticks per day.}
	\label{fig:RMSE_TickPerDay}
\end{figure}

\begin{figure}[bt] 
	\centering
	\includegraphics[width=1\textwidth]{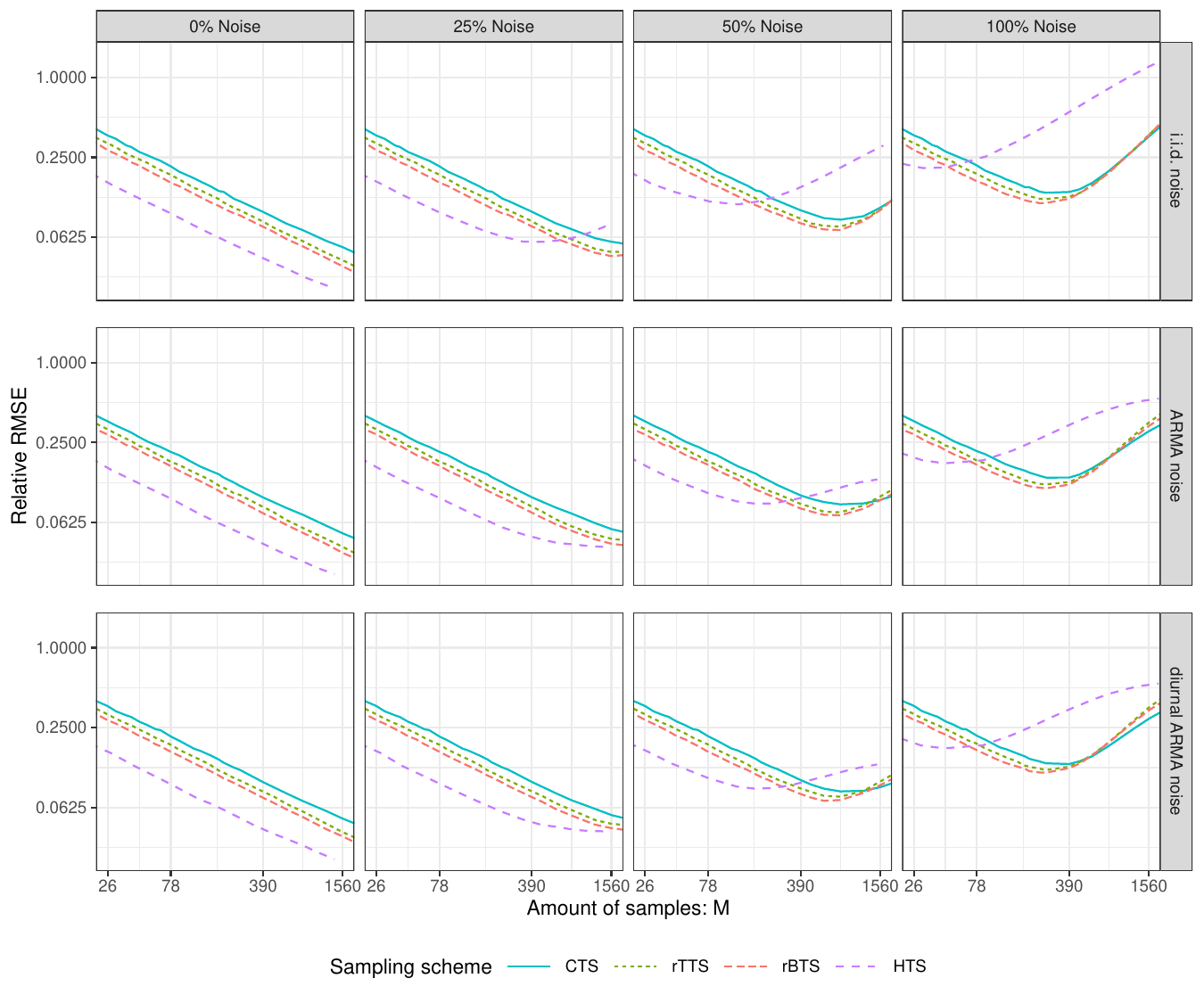}
    \caption{
		Relative RMSE of the RV estimator under the Hawkes-type TTSV process using different sampling schemes in color plotted against the (for HTS average) sampling frequencies $M$ on the horizontal axis.
        The plot rows refer to different specifications of the noise process and the plot columns refer to the noise magnitude described below \eqref{eq:SimulateNoise}.}
	\label{fig:RMSE_NoiseProcess}
\end{figure}

\begin{figure}[tb] 
	\centering
	\includegraphics[width=1\textwidth]{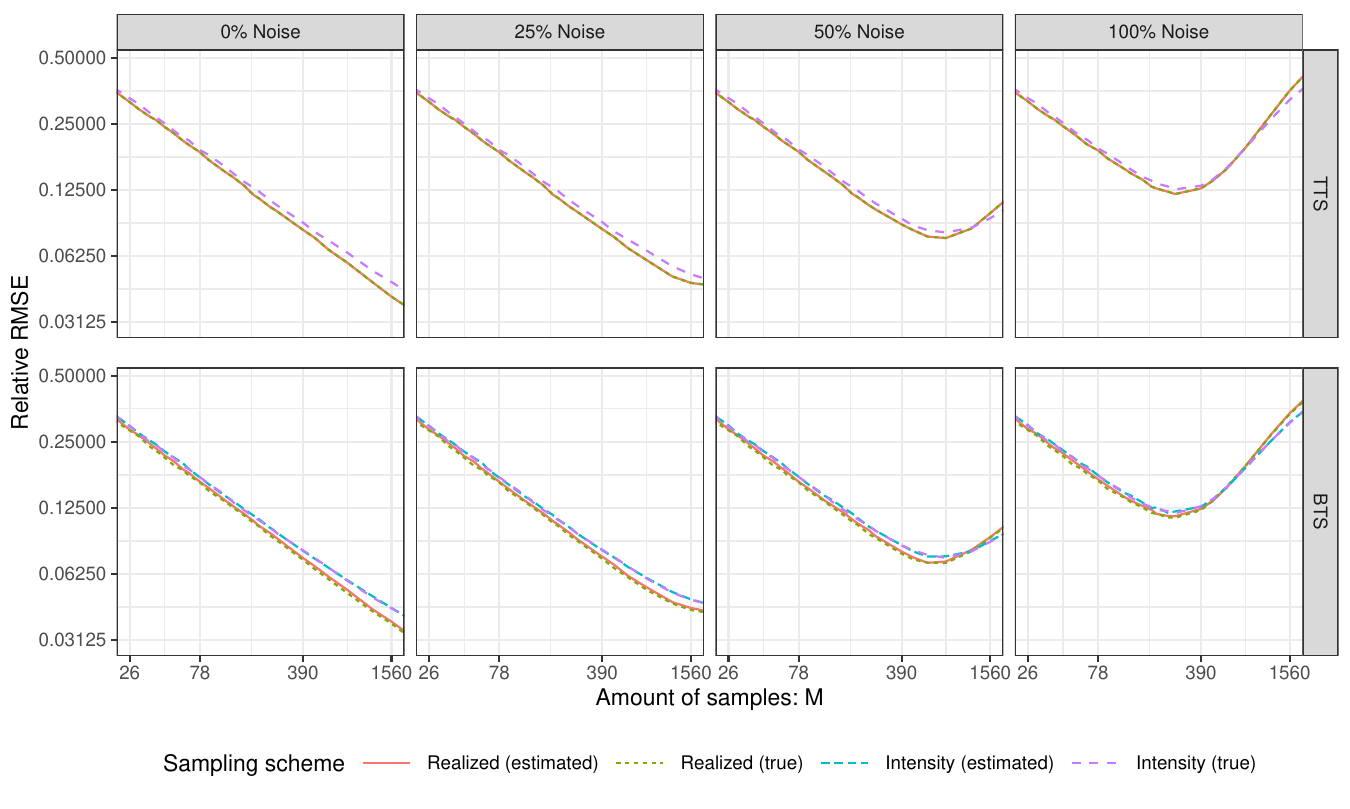}
    \caption{
		Relative RMSE of the RV estimator (for TTS and BTS in the plot rows) plotted against the sampling frequencies $M$ and for different realized and intensity based sampling schemes in color. The ``estimated'' schemes refer to estimation of the underlying intensities whereas the ``true'' versions employ the true (oracle) intensities.
	}
	\label{fig:RMSE_Intensity}
\end{figure}

\begin{figure}[tb] 
	\centering
	\includegraphics[width=1\textwidth]{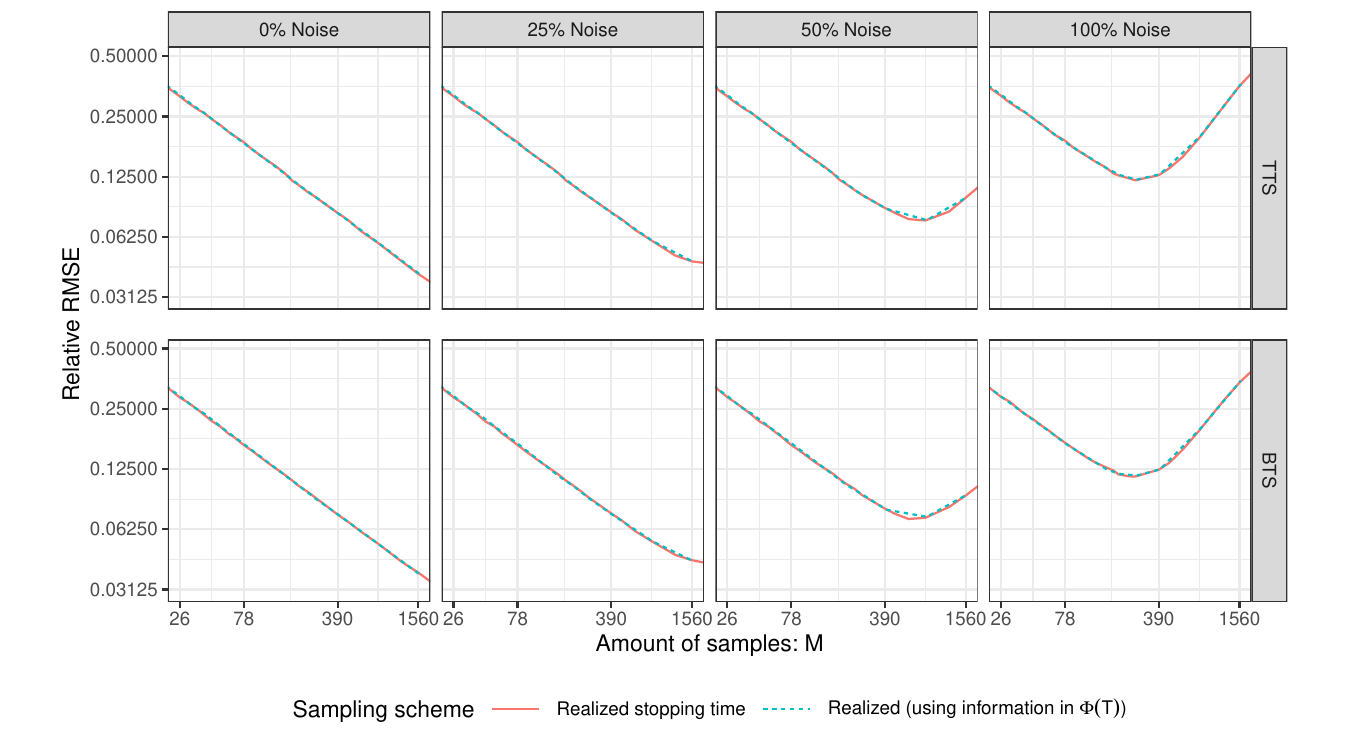}
	\caption{
		Relative RMSE of the RV estimator (for TTS and BTS in the plot rows) plotted against the (average) sampling frequencies $M$, where the colored lines refer to the stopping time versions (that generate random values for $M$) and the versions that use information $\Phi(T)$ to fix $M$; see the discussion in Section~\ref{sec:SamplingSchemes}.
	}
	\label{fig:RMSE_Stopping}
\end{figure}

\begin{table}[tb]
	\centering  
	\begin{tabular}{llrrlrrllrrlrrrr}
		\toprule
        \multicolumn{7}{c}{Sampling vs. CTS} & $\qquad \qquad$ &  \multicolumn{7}{c}{Sampling vs. rBTS}  \\
		\cmidrule{1-7}  	\cmidrule{9-15} 
		&& \multicolumn{2}{c}{MSE} &&  \multicolumn{2}{c}{QLIKE} &&&& \multicolumn{2}{c}{MSE} &&  \multicolumn{2}{c}{QLIKE} \\
		\cmidrule{3-4}  	\cmidrule{6-7}  	\cmidrule{11-12} 	\cmidrule{14-15} 	
		Sampling &  & pos & neg  &  & pos & neg& & Sampling & &pos & neg &  & pos & neg \\ 
		\midrule
        \multicolumn{15}{l}{\textbf{Panel A}: Matching $M_\delta$ to $M$ separately for every \textit{day and asset}:} \\
        &&&&&&&& CTS &  & 0 & 56 &  & 2 & 90 \\  
        rTTS &  & 46 & 0 &  & 64 & 8 && rTTS &  & 3 & 42 &  & 0 & 89 \\ 
        iBTS &  & 43 & 1 &  & 95 & 0 && iBTS &  & 4 & 29 &  & 14 & 27 \\ 
        rBTS &  & 56 & 0 &  & 90 & 2 \\ 
        HTS &  & 56 & 3 &  & 86 & 4  && HTS &  & 33 & 19 &  & 73 & 10 \\ 
        \midrule
        \multicolumn{15}{l}{\textbf{Panel B}: Matching (monthly average of) $M_\delta$ to $M$, separately on every \textit{month and asset}:} \\
        &&&&&&&& CTS &  & 0 & 56 &  & 2 & 90 \\ 
        rTTS &  & 45 & 0 &  & 65 & 8 && rTTS &  & 3 & 42 &  & 0 & 89 \\ 
        iBTS &  & 44 & 1 &  & 95 & 0 && iBTS &  & 4 & 29 &  & 14 & 27 \\ 
        rBTS &  & 56 & 0 &  & 90 & 2 \\ 
        HTS &  & 49 & 2 &  & 86 & 4 && HTS &  & 32 & 14 &  & 71 & 10 \\ 
        \midrule
        \multicolumn{15}{l}{\textbf{Panel C}: Matching (all-time average of)  $M_\delta$ to $M$, separately for every \textit{asset}:} \\
        &&&&&&&& CTS &  & 0 & 56 &  & 2 & 90 \\ 
        rTTS &  & 45 & 0 &  & 64 & 8 &&   rTTS &  & 4 & 42 &  & 0 & 89 \\ 
        iBTS &  & 43 & 1 &  & 95 & 0 &&   iBTS &  & 4 & 29 &  & 14 & 26 \\ 
        rBTS &  & 56 & 0 &  & 90 & 2 \\ 
        HTS &  & 47 & 4 &  & 85 & 4 &&   HTS &  & 31 & 15 &  & 67 & 10 \\ 
        \midrule
        \multicolumn{15}{l}{\textbf{Panel D}: Matching (average over days and assets) $M_\delta$ to $M$:} \\
        &&&&&&&& CTS &  & 0 & 55 &  & 2 & 90 \\ 
        rTTS &  & 47 & 0 &  & 64 & 8 && rTTS &  & 4 & 42 &  & 0 & 89 \\ 
        iBTS &  & 43 & 1 &  & 95 & 0 && iBTS &  & 4 & 29 &  & 14 & 27 \\ 
        rBTS &  & 56 & 0 &  & 90 & 2 \\ 
        HTS &  & 44 & 4 &  & 85 & 6 &&  HTS &  & 30 & 16 &  & 64 & 12 \\ 
		\bottomrule
	\end{tabular}
    \caption{Percentage values of significantly positive (``pos'') and negative (``neg'')  MSE and QLIKE loss differences between the sampling schemes mentioned in the column ``Sampling'' against the one in the title using the method of \cite{Patton2011RV}.
	The percentage values are computed over the 27 assets and the seven employed values of $M$ for the respective estimators.
    The four panels A--D correspond to different methods how the daily varying $M_\delta$ of HTS is matched to the fixed values of $M$ of the other sampling schemes, which is further described in footnote \ref{fn:HTS_Matching}. Panel A corresponds to the results presented in Table~\ref{tab:applSigPosNegValues}.}
	\label{tab:applSigPosNegValuesAggregation}
\end{table}

\begin{figure}[p]
	\centering
	\includegraphics[width=0.9\textwidth]{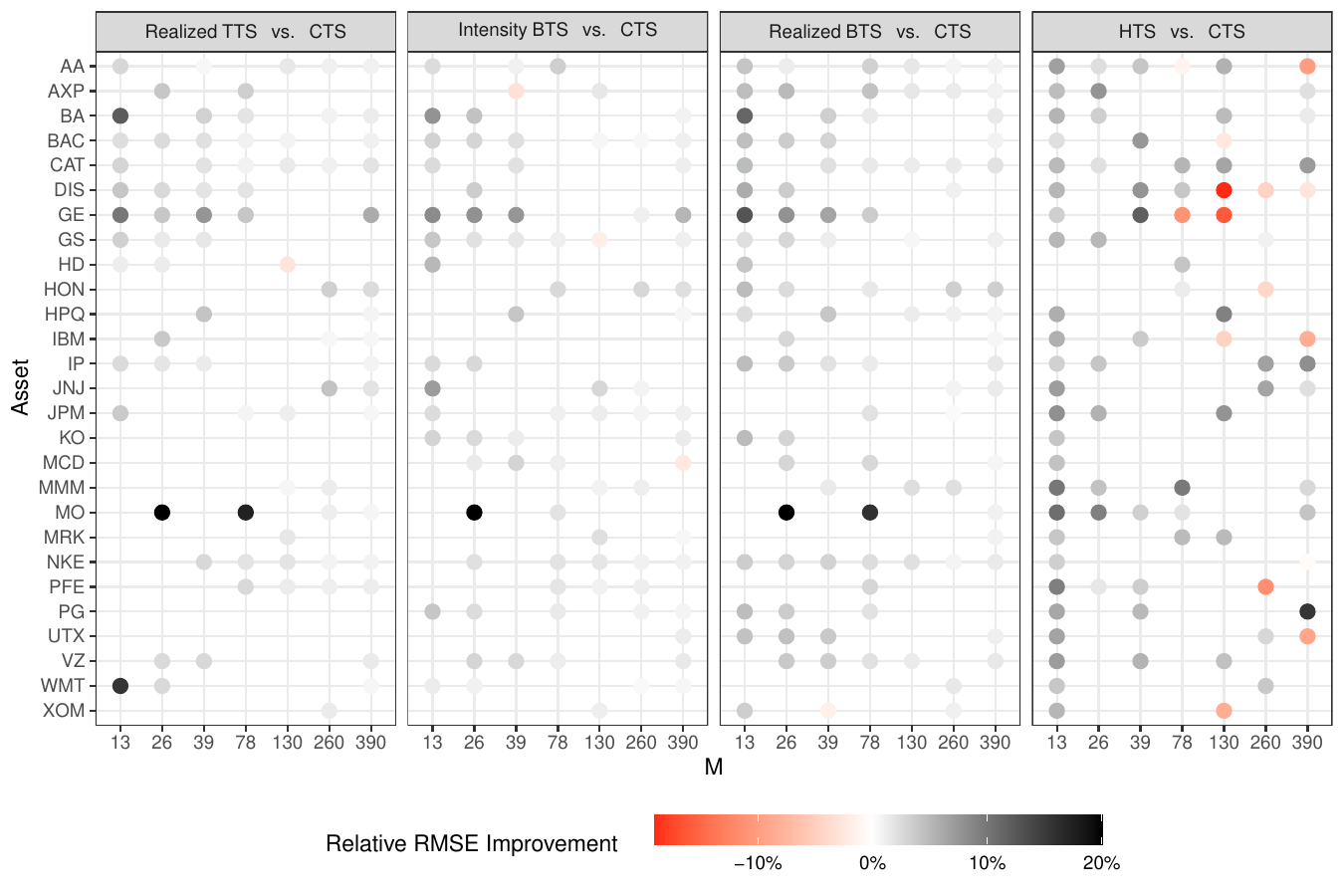}
	\includegraphics[width=0.9\textwidth]{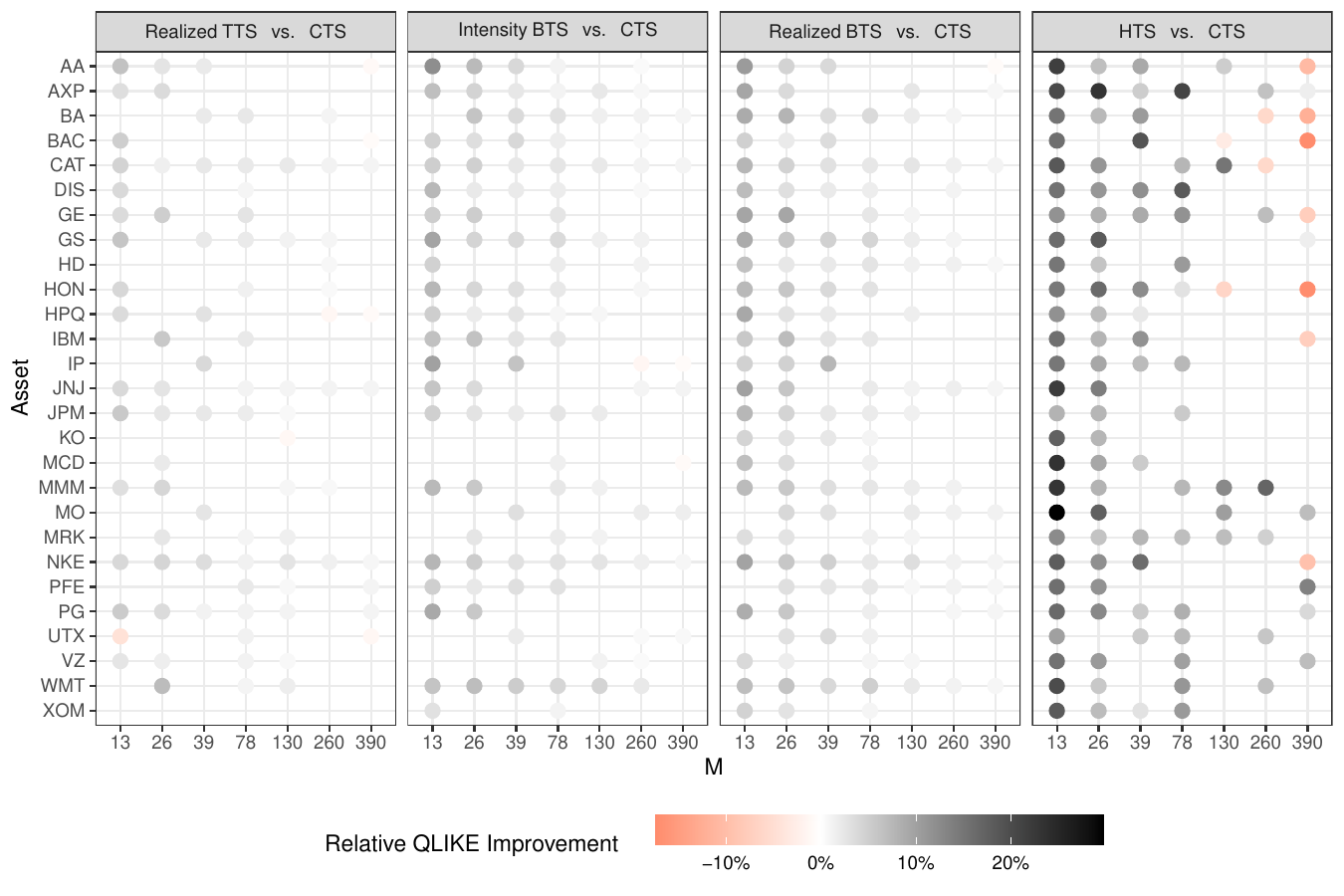}
	\caption{RMSE (top) and QLIKE (bottom) loss differences for the RV estimator based on different sampling schemes and a range of sampling frequencies $M$ for the 27 considered assets.
	Each point corresponds to a (at the $5\%$ level) significant  loss difference of the corresponding RV estimator to a \emph{benchmark CTS RV estimator} with the same sampling frequency.
    For the evaluation proxy, we use daily squared returns here.
	Insignificant loss differences are omitted.
	The color scale of the points shows the relative improvement in terms of RMSE/QLIKE, where black (red) colors refer to an improvement (decline).}
	\label{fig:appl_RV_eval_vsCTS_SqRet}
\end{figure}

\begin{figure}[p]
	\centering
	\includegraphics[width=0.9\textwidth]{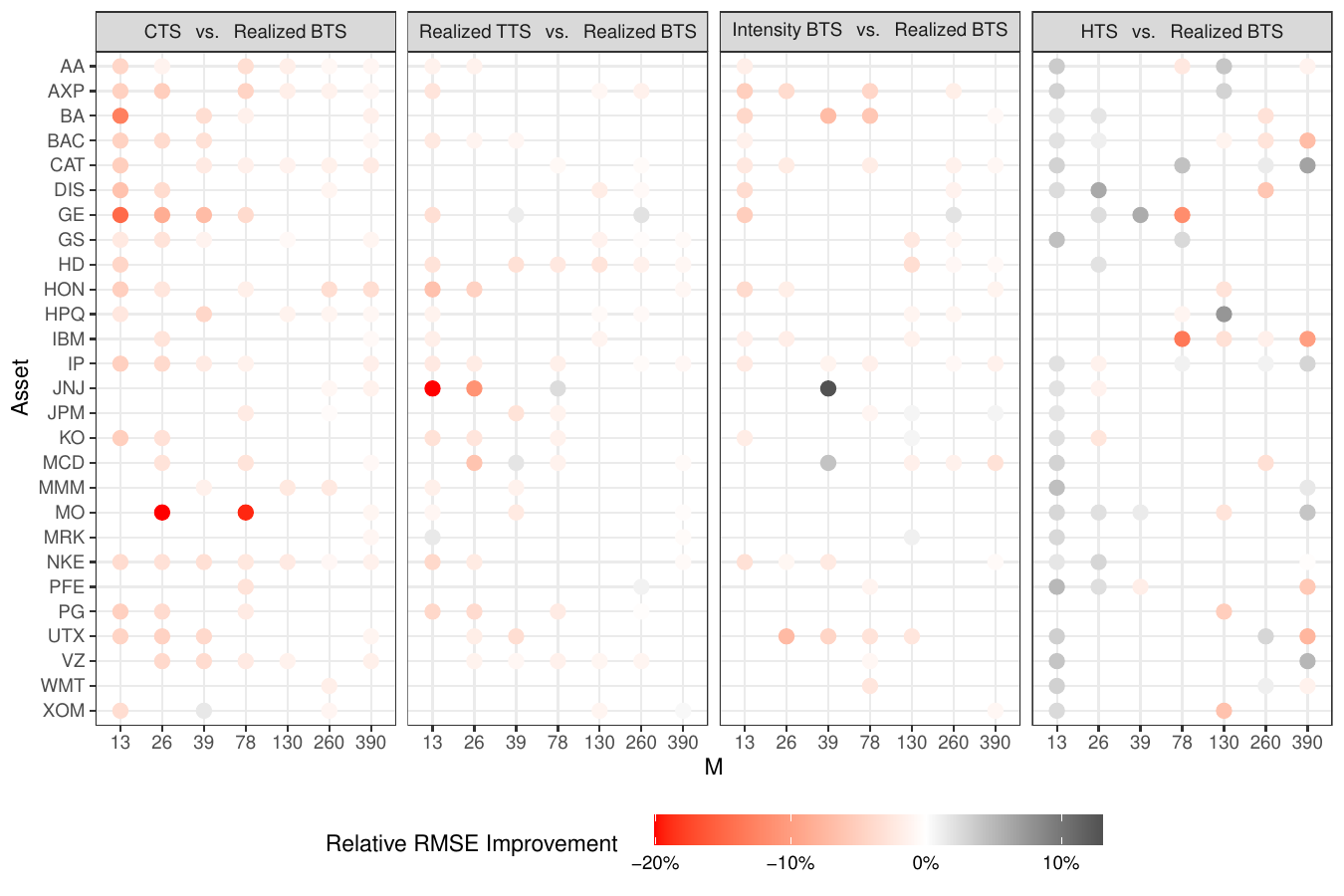}
	\includegraphics[width=0.9\textwidth]{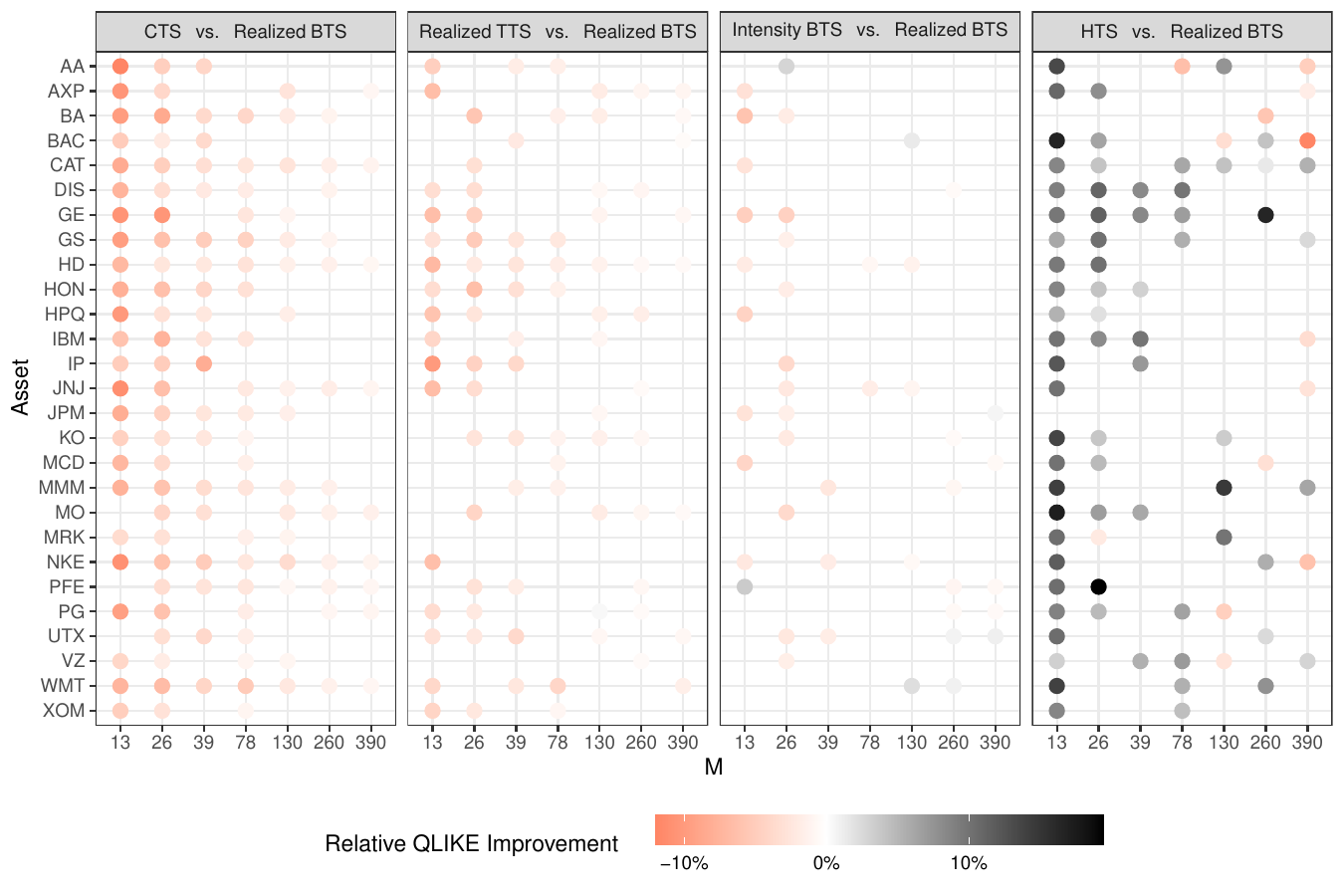}
	\caption{RMSE (top) and QLIKE (bottom) loss differences for the RV estimator based on different sampling schemes and a range of sampling frequencies $M$ for the 27 considered assets.
	Each point corresponds to a (at the $5\%$ level) significant  loss difference of the corresponding RV estimator to a \emph{benchmark rBTS RV estimator} with the same sampling frequency.
    For the evaluation proxy, we use daily squared returns here.
	Insignificant loss differences are omitted.
	The color scale of the points shows the relative improvement in terms of RMSE/QLIKE, where black (red) colors refer to an improvement (decline).}
	\label{fig:appl_RV_eval_vsrBTS_SqRet}
\end{figure}

\begin{figure}[tb]
	\centering
	\includegraphics[width=0.9\textwidth]{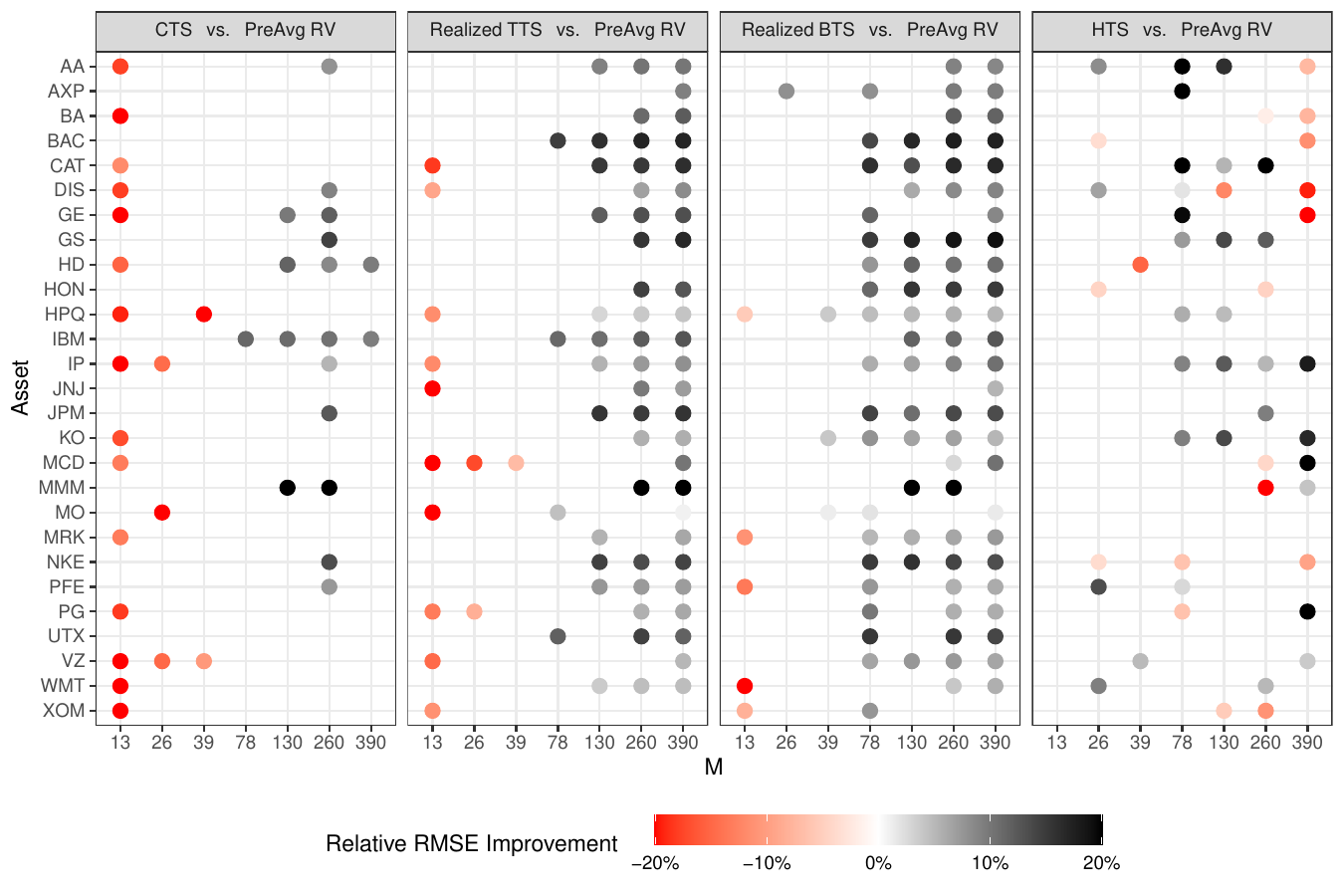}
	\caption{RMSE loss differences for the RV estimator based on different sampling schemes and a range of sampling frequencies $M$ for the 27 considered assets.
    Here, unlike Figure~\ref{fig:appl_RV_eval_vs_PAVG}, we use the (leaded) CTS RV estimator with $M=78$ as the proxy in the evaluation framework of \citet{Patton2011RV}.
	Each point corresponds to a (at the $5\%$ level) significant  loss difference of the corresponding RV estimator to a \emph{benchmark pre-averaging RV estimator} using all tick-level returns.
	Insignificant loss differences are omitted.
	The color scale of the points shows the relative improvement in terms of RMSE, where black (red) colors refer to an improvement (decline).}
	\label{fig:appl_RV_eval_vs_PAVG_CTSProxy}
\end{figure}	

\begin{figure}[tb]
	\centering
	\includegraphics[width=0.9\textwidth]{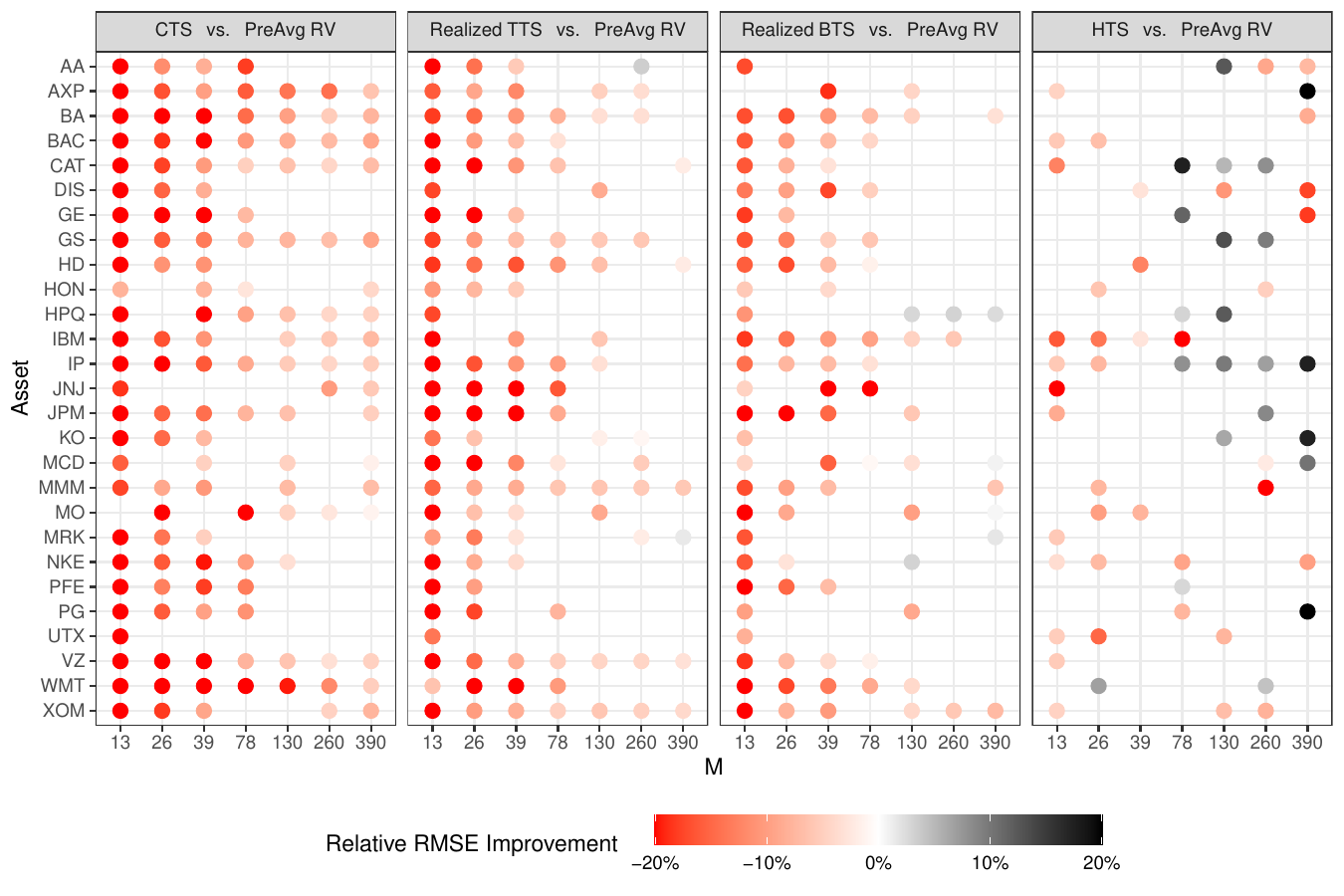}
	\caption{RMSE loss differences for the RV estimator based on different sampling schemes and a range of sampling frequencies $M$ for the 27 considered assets.
    Here, unlike Figure~\ref{fig:appl_RV_eval_vs_PAVG}, we use the (leaded) pre-averaging RV estimator as the proxy in the evaluation framework of \citet{Patton2011RV}.
	Each point corresponds to a (at the $5\%$ level) significant  loss difference of the corresponding RV estimator to a \emph{benchmark pre-averaging RV estimator} using all tick-level returns.
	Insignificant loss differences are omitted.
	The color scale of the points shows the relative improvement in terms of RMSE, where black (red) colors refer to an improvement (decline).}
	\label{fig:appl_RV_eval_vs_PAVG_PAVGProxy}
\end{figure}

\end{document}